\documentclass[11pt,twoside]{article}
\usepackage{tikz}
\usetikzlibrary{positioning}
\usepackage{amsmath,amsthm}
\usepackage{graphicx}
\usepackage{geometry}
\usepackage{multirow}
\usepackage{float}
\usepackage[symbol]{footmisc}
\usepackage{sectsty}
\usepackage{amssymb}

\usepackage[utf8]{inputenc}
\usepackage[T1]{fontenc}
\usepackage{lmodern}
\usepackage[english]{babel}
\usepackage{braket}
\usepackage{mathrsfs}
\usepackage{stmaryrd }
\usepackage{titleps}
\usepackage{tabularx}
\usepackage{enumitem}
\usepackage{color}
\usetikzlibrary{arrows.meta}

\newcommand{\field}[1]{\mathbb{#1}}
\newcommand{\N}{\field{N}}
\newcommand{\R}{\field{R}}
\newcommand{\C}{\field{C}}

\newcommand{\II}{\mathcal I}
\newcommand{\JJ}{\mathcal J}

\newcommand{\XX}{\mathfrak X}
\newcommand{\HH}{\mathscr H}
\newcommand{\KK}{\mathscr K}

\newcommand{\LL}{\mathscr L}

\newcommand{\eps}{\varepsilon}
\newcommand{\ph}{\varphi}

\newcommand{\ran}{\mathrm{Ran}}
\renewcommand{\ker}{\mathrm{Ker}}

\renewcommand{\d}[1]{\textup{d}#1}

\newcommand{\sgn}{\operatorname{sgn}}

\newcommand{\supp}{\operatorname{supp}}
\newcommand{\esssup}{\operatorname*{ess\;sup}}
\newcommand{\Ima}{\operatorname{Im}}
\newcommand{\Rea}{\operatorname{Re}}

\newcommand{\rst}{\! \upharpoonright \!} 

\DeclareMathOperator*{\limeps}{\lim\limits_{\eps \to 0}}
\DeclareMathOperator*{\limepsr}{\lim_{\eps \to 0}}

\newtheorem{theorem}{Theorem}[section]
\newtheorem{lemma}[theorem]{Lemma}
\newtheorem{prop}[theorem]{Proposition}

\newtheorem{kor}[theorem]{Corollary}

\usepackage[babel]{csquotes}

\sectionfont{\Large}
\subsectionfont{\Large}
\date{\today}
\usepackage[normalem]{ulem}

\sectionfont{\Large}
\subsectionfont{\large}
 
\def\Biggs#1{{\hbox{$\left#1\vbox to20pt{}\right.$}}} 
\font\notefont=cmsl8 
\title{From Short-Range to Contact Interactions in Two-dimensional Many-Body Quantum Systems}

\author{Marcel~Griesemer\footnote{marcel.griesemer@mathematik.uni-stuttgart.de} \;and Michael~Hofacker\footnote{michael.hofacker@mathematik.uni-stuttgart.de}\\  
\small Fachbereich Mathematik, Universit\"at Stuttgart, D-70569 Stuttgart, Germany}  
\date{}

\begin{document}
\maketitle

\setcounter{figure}{0}

\numberwithin{equation}{section}
\newpagestyle{mypage}{%
      \setfoot{}{\thepage}{}
}

\pagestyle{mypage}
\setcounter{page}{1}
\newcommand\Scalefrac[2]{\scalebox{0.85}{$\dfrac{#1}{#2}$}}
\newcommand\sym[0]{\scalebox{0.7}{$\mspace{-6mu}\raisebox{-2mm}{\textup{sym}}$}}
\newcommand\huger[0]{\scalebox{5}{$)$}}
\renewcommand{\abstractname}{Abstract}
\renewcommand{\thefootnote}{\fnsymbol{footnote}}

\begin{abstract}
\noindent Quantum systems composed of $N$ distinct particles in $\R^2$ with two-body contact interactions of TMS type are shown to arise as limits - in the norm resolvent sense - of Schr\"odinger operators with suitably rescaled pair potentials.  
\end{abstract} 

\section{Introduction} \label{sec1}
\normalfont
\normalsize

Many-particle quantum systems with short-range interactions are often described by simplified models with zero-range (contact) interactions.  Effective models of low energy nuclear physics, the Lieb-Liniger model for the 1d Bose gas, and the Fermi polaron model for an impurity within an ideal Fermi gas  are well-known models of this type \cite{Kolck1999,LiLi,Parish2013}.  The understanding of these models is that zero-range interactions provide an idealized description of short-ranged interparticle forces whose details are not known or considered irrelevant.  In the present paper we fully justify this idealization for 2d many-particle systems with (spin-independent) two-body forces. We prove norm resolvent convergence for suitably rescaled Schr\"odinger operators towards TMS Hamiltonians \cite{Figari}. This mathematically justifies the use of TMS Hamiltonians for describing contact interactions, and, moreover, provides a new way of defining them. 

In the Hilbert space
\begin{align}
     \HH:= L^2\left(\R^{2N}, \textup{d}(x_1,\ldots,x_N) \right), \nonumber
\end{align}
we consider Schr\"odinger operators of the form
\begin{align}\label{DefHeps}
   H_\eps:= \sum\limits_{i=1}^{N} (-\Delta_{x_i}/m_i)  - \sum_{\sigma\in \II}  g_{\eps, \sigma} \, V_{\sigma,\eps}(x_j-x_i), \qquad \quad \eps>0,
\end{align}
where $m_i>0$ is the mass of the $i$th particle, $\II$ denotes the set of all pairs $\sigma=(i,j)$, $i<j$, and 
\begin{equation} \label{Vscale}
     V_{\sigma,\eps}(r):=\eps^{-2}\,V_{\sigma}(r/\eps), \qquad \quad \eps>0.
\end{equation}
We assume $V_{\sigma}\in L^1\cap L^2(\R^2)$, $V_\sigma(-r)=V_\sigma(r)$ as well as some mild decay of $V_\sigma(r)$ as $|r|\to\infty$.
For $H_\eps$ to have a limit as $\eps\to 0$, the coupling constant $g_{\eps, \sigma}$ must have an asymptotic behavior, as $\eps\to 0$, adjusted to the space dimension. For $d=2$ we may set  
\begin{equation}\label{g-def}
    -\frac{1}{g_{\eps, \sigma}} = \mu_\sigma\big [a_\sigma\ln(\eps) + b_{\sigma} \big]+o(1)
\end{equation}
with $a_\sigma,b_\sigma\in \R$, $a_\sigma>0$. Here the reduced mass $\mu_\sigma$ of the pair $\sigma=(i,j)$ has been factored out. We assume that
\begin{equation}\label{V-cond}
   a_\sigma \geq \frac{1}{2\pi}\int V_{\sigma}(r)\,\d{r}.
\end{equation}
If equality holds in \eqref{V-cond}, we will see that the two-body interaction $V_\sigma$ gives rise to a non-trivial contact interaction, whose strength is determined by the parameter $b_\sigma$. In the case of inequality there is no contribution from $V_\sigma$. This is well-known for the case of $N=1$ particles, where our problem was solved long ago \cite{SolvableModels}. We use $\JJ\subset\II$ to denote the subset of pairs for which equality holds in \eqref{V-cond}.

Our main result asserts norm resolvent convergence of $H_\eps$ as $\eps\to 0$. In addition, we give an explicit expression for the resolvent of the limiting Hamiltonian, a characterization of the domain in terms of so-called TMS boundary conditions at the collision planes $\{x\in \R^{2N}\mid x_i=x_j\}$, and we show, in the case of equal masses $m_j=1$, that our limiting Hamiltonian agrees with the Hamiltonian $\textup{H}_{\beta}$ (for suitable $\beta=(\beta_{\sigma})_{\sigma \in \II}$)  of Dell'Antonio et al. (see \cite[Eqs. (5.3)-(5.4)]{Figari}).
To describe these results in detail, we need some auxiliary spaces and operators. 

Let  
\begin{equation}
\begin{aligned}
& \XX := \bigoplus_{\sigma \in \JJ} \XX_{\sigma},  \\
& \XX_{\sigma=(i,j)} := L^2\left(\R^{2(N-1)}, \textup{d}\left(R,x_1,...\widehat{x}_i...\widehat{x}_j...,x_N\right)\right), \label{DefHred}
\end{aligned}
\end{equation} 
where the hat, as in $\widehat{x}_i$, indicates omission of a variable.
Let $H_0 = \sum\limits_{i=1}^{N} (-\Delta_{x_i}/m_i)$ on $D(H_0)=H^2(\R^{2N})$ and let 
$$
    T_\sigma:D(T_{\sigma})\to \XX_{\sigma}
$$
denote the $H_0$-bounded trace operator, which, for smooth functions $\psi$, is given by
\begin{equation}\label{trace}
     (T_\sigma\psi)(R,x_1,...\widehat{x}_i...\widehat{x}_j...,x_N) = \psi(x_1,\ldots,x_N)\big|_{x_i=x_j=R}, \qquad \qquad \quad \sigma=(i,j).
\end{equation}
Then $G(z)_{\sigma} = T_\sigma R_0(z)$, with $R_0(z)=(H_0+z)^{-1}$, is a bounded operator from $\HH$ to $\XX_{\sigma}$. 
This operator is closely related to the Green's function of $H_0+z$, see, e.g., \eqref{G_(k,l)(z)*Fourier}, below, and the letter $G$ serves to remind us of the connection. Let $T$ and $G(z) = TR_0(z)$ denote the operator-valued vectors with components  $ T_\sigma$ and $G(z)_{\sigma}$, respectively. 
Our main result is the following:\footnote{Note our definition of the resolvent set at the end of this section.}

\begin{theorem}\label{theo1}
Suppose, for all $\sigma\in \II$, that $V_{\sigma} \in L^1\cap L^2(\R^2)$, $V_\sigma(-r) = V_\sigma(r)$, and there exists some $s>0$ such that $\int |r|^{2s} \left|V_{\sigma}(r)\right| \, \d{r}< \infty$. Let $H_{\eps}$ be defined by \eqref{DefHeps}-\eqref{V-cond}. Then, as $\eps\to 0$, $H_{\eps}$ converges in norm resolvent sense to a self-adjoint, semibounded operator $H$.
The resolvent of $H$, for $z \in \rho(H_0)\cap \rho(H)$, obeys
\begin{align} \label{Hres}
(H+z)^{-1} = R_0(z) +  G(\overline{z})^{*} \Theta(z)^{-1} G(z)
\end{align}
with some invertible and unbounded operator $\Theta(z):D\subset \XX\to \XX$, whose domain is independent of $z$.
\end{theorem}

\noindent\emph{Remarks:} 
\begin{enumerate}
\item $\Theta(z)$ is an operator matrix with entries $\Theta(z)_{\sigma\nu}:\XX_{\nu}\to\XX_{\sigma}$. For $\sigma\neq\nu$ these entries are bounded operators given by $\Theta(z)_{\sigma\nu} = -T_\sigma G(\overline{z})_{\nu}^{*}=-T_\sigma (T_{\nu}R_0(\overline{z}))^{*}$. The unbounded operators $\Theta(z)_{\sigma\sigma}$ are given explicitly in Equation \eqref{DefTheta(z,P)diag}, below.
\item Antisymmetric wave functions from $H^2(\R^{2N})$ belong to the kernel of all trace operators $T_\sigma$. Hence for identical particles, by \eqref{Hres}, $(H+z)^{-1} = R_0(z)$ on the subspace $\HH_{\rm anti}\subset\HH$ of antisymmetric wave functions. In fact, there are no contact interactions among identical, spin-aligned fermions in space dimensions $d\geq 2$ \cite{GH2021}. Theorem \ref{theo1} implies $\lim_{\eps\to 0}(H_{\eps}+z)^{-1} = R_0(z)$ on $\HH_{\rm anti}$ and its proof gives a rate of convergence \cite{GH2021}. 
\end{enumerate}

The main novelty in Theorem \ref{theo1}, compared to previous results of similar type, is that convergence is established in norm resolvent sense. Norm resolvent convergence implies convergence of the spectrum, which is not true for the weaker strong resolvent convergence \cite{Teschl}. In the present case, where $\sigma(H_\eps) = [\Sigma_\eps,\infty)$, the norm resolvent convergence implies $\sigma(H) = [\Sigma,\infty)$ with $\Sigma=\lim_{\eps\to 0}\Sigma_\eps$.

The analog of  Theorem \ref{theo1} in the case $N=1$ is well-known and true without a condition of the type \eqref{V-cond}. But if $N>1$ this condition is necessary: if \eqref{V-cond} is not satisfied for some pair $\sigma$, then one may expect strong resolvent convergence at best. To see this, consider a two-particle system with pair potential $V$ and reduced mass $\mu=1$. If $0<a<\frac{1}{2\pi}\int V(r)\,\d{r}$ and $-g_\eps^{-1} = a\ln(\eps)+b$, then the Hamiltonian in the center of mass frame has a negative eigenvalue running towards $-\infty$ as $\eps\to 0$ \cite{Simon76}. Due to the center of mass motion, the Hamiltonian $H_\eps$ then has essential spectrum filling the entire real axis in the limit $\eps\to 0$. This is not compatible with norm resolvent convergence towards a semibounded Hamiltonian.

By the following corollary, domain vectors of $H$ satisfy the so-called TMS boundary conditions at the collision planes $\{x\in \R^{2N} \mid x_i=x_j\}$.

\begin{kor}\label{cor2}
Under the assumptions of Theorem~\ref{theo1}, a vector $\psi\in \HH$ belongs to $D(H)$ if and only if the following holds:
For some (and hence all) $z\in\rho(H_0)\cap \rho(H)$ there exist $\psi_0\in D(H_0)$ and $w \in D$ such that 
\begin{equation}\label{TMS1}
   \psi = \psi_0 + G(\overline{z})^{*} w   
\end{equation}
and \vspace*{-2mm}
\begin{equation}\label{TMS2}
    T\psi_0 =  \Theta(z)w.   
\end{equation}
In this case, \vspace*{-3mm}
\begin{align}
(H+z)\psi = (H_0+z)\psi_0. \label{Hact}
\end{align}
The vectors $\psi_0$ and $w$ are uniquely determined by $\psi\in D(H)$ and  $z\in\rho(H_0)\cap \rho(H)$.
\end{kor}

Choosing $w=0$ in Corollary \ref{cor2}, we see that $D(H_0) \cap \ker\,T \subset D(H)$ and that $H=H_0$ on $D(H_0) \cap \ker\,T$. Thus $H$ is a self-adjoint extension of $H_0\rst \ker\,T$. It is well-known that a Krein formula, like \eqref{Hres}, is the characteristic equation for the resolvent of all such extensions \cite{Pos2001,Pos2008}. 

Recently, in a study of the 2d stochastic heat equation, Gu, Quastel and Tsai have derived a result very similar to Theorem \ref{theo1} for $N$ identical particles \cite{Quastel2019}. In  \cite{Quastel2019} the two-body potentials are compactly supported smooth functions and convergence in strong resolvent sense is established. In the case of bosons, Gu et al.~give a formula for the resolvent that is similar to the one in Theorem~\ref{theo1} and the subsequent remark.   In one space dimension, results similar to Theorem~\ref{theo1} for $N=3$ quantum particles are obtained in \cite{Basti2018}, and for arbitrary $N\geq 1$ in \cite{GHL2019}. The well-known Lieb-Liniger model  with repulsive $\delta$-interactions is derived from a trapped 3d Bose gas with two-body potentials in \cite{SeiYin}. TMS Hamiltonians like $H$ in Theorem \ref{theo1} have also been described as resolvent limits of $N$-body Hamiltonians, where the regularized two-body contact interaction is an integral operator, rather than a potential, and the regularization is achieved by  an ultraviolet cutoff \cite{Figari,DR,GriesemerLinden2} or a reversed heat flow \cite{Turgut2013}. In these cases the convergence is easier to establish than in the case studied here. Nevertheless, all previous approximation results of this kind in 2d with $N\geq 2$ particles establish \emph{strong} resolvent convergence only. In three space dimensions, TMS Hamiltonians for $N\geq 3$ bosons are symmetric but not self-adjoint, and all (non-trivial) self-adjoint extensions are unbounded below \cite{MinlosFaddeev1, MinlosFaddeev2}. A possible solution to this problem in the case of $N= 3$ bosons is worked out in the recent paper  \cite{basti2021}, where a semi-bounded self-adjoint Hamiltonian containing a three-body contact interaction is constructed.

To prove Theorem \ref{theo1}, we further develop the methods and tools introduced in our previous paper \cite{GHL2019}. The key elements along with some auxiliary spaces and tools are described in Section \ref{sec2}. Sections 3 and 4 provide all preparations needed for the proof of Theorem \ref{theo1}, which is given in Section 5. In addition, we derive a lower bound on $\sigma(H)$ in Section 5. In Section 6 we compute the quadratic form of the Hamiltonian $H$ and we show it agrees, in the case of equal masses, with a quadratic form derived by Dell'Antonio et al. \cite{Figari}, and which, in recent years, has become a standard starting point for investigations of TMS Hamiltonians. 
  
We conclude this introduction with some remarks on our notations and with the proof of Corollary \ref{cor2}.  In this paper the resolvent set $\rho(H)$ of a closed operator $H$ is defined as the set of
$z \in \C$ for which $H + z: D(H) \subset \HH \rightarrow \HH$ is a bijection. This differs by a minus sign from the
conventional definition and it means that the spectrum $\sigma(H)$ is the complement of $-\rho(H)$.
The $L^2$-norm will be denoted by $\|\cdot\|$, without index, while all other
norms carry the space as an index, as e.g. in $\|V\|_{L^1}$.

\begin{proof}[Proof of Corollary~\ref{cor2}]
In this proof, $T : D(H_0)\to \XX$ and $G(z) : \HH\to \XX$ are the operator-valued vectors defined in terms of the components $T_\sigma$ and $G(z)_{\sigma}$. We know from \cite{CFP2018} and Proposition \ref{prop5.1} that a point $z\in \rho(H_0)$ belongs to $\rho(H)$ if and only if $\Theta(z)$ has a bounded inverse.
 
Suppose $\psi\in D(H)$ and $z\in\rho(H_0)\cap \rho(H)$. Define $\ph := (H+z)\psi$, $\psi_0:=R_0(z)\ph$, and $w:=\Theta(z)^{-1} T\psi_0$.
Then, by \eqref{Hres},
\begin{align*}
   \psi = (H+z)^{-1}\ph &= R_0(z)\ph +  G(\overline{z})^{*} \Theta(z)^{-1} G(z)\ph\\
   &= \psi_0 +  G(\overline{z})^{*}w. \nonumber
\end{align*}

Conversely, if \eqref{TMS1} and \eqref{TMS2} hold for some $z\in\rho(H_0)\cap \rho(H)$, let $\ph:=(H_0+z)\psi_0$. Then $T\psi_0 = T R_0(z)\ph = G(z)\ph$ and $w = \Theta(z)^{-1} T\psi_0 = \Theta(z)^{-1}G(z)\ph$. Hence, by \eqref{TMS1} and \eqref{Hres},
\begin{align*}
    \psi = \psi_0 + G(\overline{z})^{*} w
    &= R_0(z)\ph + G(\overline{z})^{*} \Theta(z)^{-1} G(z)\ph\\
    &= (H+z)^{-1}\ph \in D(H). \nonumber
\end{align*}
It follows that $(H+z)\psi =\ph= (H_0+z)\psi_0$. In particular, $\psi_0= (H_0+z)^{-1}(H+z)\psi$ is uniquely determined by $\psi$, and, since $\Theta(z)$ is invertible, it follows from \eqref{TMS2} that $w\in D$ is unique as well.
\end{proof}

\section{Auxiliary operators and strategy of the proof}
\label{sec2}

The proof of Theorem \eqref{theo1} starts with the new expression \eqref{Hepsnew}, below, for $H_\eps$, which allows for the explicit representation of the resolvent in terms of a generalized Konno-Kuroda formula, see Equation~\eqref{Krein}.  We now describe the various operators occurring in \eqref{Hepsnew} and \eqref{Krein}. To this end, we need the auxiliary Hilbert spaces 
\begin{align}
\widetilde{\XX}&:= \bigoplus_{\sigma \in \II} \widetilde{\XX}_{\sigma}, \nonumber \\
\widetilde{\XX}_{\sigma=(i,j)} &:= L^2\left(\R^{2N}, \textup{d}\left(r,R,x_1,...\widehat{x}_i...\widehat{x}_j...,x_N \right) \right). \label{HTilde}
\end{align}
The integration  variables $r$ and $R$ in \eqref{HTilde} correspond to the relative and center of mass coordinates 
\begin{align}
r=r_{\sigma}:=x_j-x_i, \qquad R=R_{\sigma}:= \frac{m_ix_i+m_jx_j}{m_i+m_j} \label{DefrR}
\end{align}
of the particles constituting the pair $\sigma=(i,j)$. This change of coordinates is implemented unitarily by the operator $\KK_{(i,j)}: \HH \rightarrow \widetilde{\XX}_{(i,j)}$ defined by
\begin{eqnarray} \label{DefK}
\lefteqn{\big(\KK_{(i,j)}\Psi \big)\left(r,R,x_1,...\widehat{x}_i...\widehat{x}_j...,x_N\right)} \nonumber \\
&&:= \Psi\Bigg(x_1,...,x_{i-1},R-\frac{m_j r}{m_i+m_j},x_{i+1},...,x_{j-1},R+\frac{m_i r}{m_i+m_j}, x_{j+1},..., x_{N}\Bigg).
\end{eqnarray}
The adjoint thereof is the operator $\KK_{(i,j)}^{*}: \widetilde{\XX}_{(i,j)} \rightarrow \HH$ with
\begin{align}
\big(\KK_{(i,j)}^{*} \Psi\big)(x_1,...,x_{N})=  \Psi\left(x_j-x_i,\frac{m_ix_i+m_jx_j}{m_i+m_j} ,x_{1},...\widehat{x}_i...\widehat{x}_j...,x_N \right). \label{DefK*}
\end{align}
Let $U_\eps\in \LL(L^2(\R^2))$ denote the unitary rescaling that, on $\LL(\widetilde{\XX}_{(i,j)})$, is given by
\begin{align}
 \left(U_{\eps} \Psi\right)(r,\underline{X}):= \eps \,\Psi(\eps r, \underline{X}), \qquad \qquad \quad \underline{X}=(R,x_1,...\widehat{x}_i...\widehat{x}_j...,x_N) ,  \label{DefUeps}
\end{align}
and let
\begin{align*}
      v_\sigma(r) &= |V_\sigma(r)|^{1/2}, \\
      u_\sigma(r) &= J_\sigma |V_\sigma(r)|^{1/2},\qquad J_\sigma:=\sgn(V_\sigma),
\end{align*}
so that $V_\sigma = u_\sigma v_\sigma$. Then, in terms of the above operators, we define new operators $A_{\eps, \sigma},B_{\eps, \sigma}: D(A_{\eps, \sigma}) \subseteq \HH \rightarrow \widetilde{\XX}_{\sigma}$ by
\begin{align}
A_{\eps, \sigma} &:= (v_{\sigma} \otimes 1) \,\eps^{-1} U_{\eps} \KK_{\sigma}, \label{Aepsijnew}\\
B_{\eps, \sigma} &:= (u_{\sigma} \otimes 1) \,\eps^{-1} U_{\eps} \KK_{\sigma} = J_\sigma A_{\eps, \sigma}, \label{Bepsijnew}
\end{align}
where $D(A_{\eps, \sigma})$ is determined by the domain of the multiplication operator $v_{\sigma} \otimes 1$. Obviously,  $A_{\eps, \sigma}$ and $B_{\eps, \sigma}$ are densely defined and closed. These operators allow us to write the two-body interaction in the form 
\begin{align}
      V_{\sigma,\eps}(x_j-x_i) = \KK_{\sigma}^{*}(V_{\sigma,\eps} \otimes 1) \KK_{\sigma} = A_{\eps, \sigma}^{*} B_{\eps, \sigma} \label{Potdecomp}
\end{align}
and therefore
\begin{align} \label{Hepsnew}
        H_{\eps}=H_0 - \sum_{\sigma \in \II} g_{\eps, \sigma} A_{\eps, \sigma}^{*} B_{\eps, \sigma}, \qquad \quad \eps>0.
\end{align}
Since $V_\sigma\in L^2(\R^2)$, it is clear from \eqref{Potdecomp} that $A_{\eps, \sigma}^{*} B_{\eps, \sigma}$ as well as $A_{\eps, \sigma}^{*} A_{\eps, \sigma}$ is infinitesimally $H_0$-bounded. So $H_\eps$ is self-adjoint on $D(H_0)$, and moreover, the hypotheses of Corollary \ref{kor:KK} are satisfied. This means that 
\begin{align}\label{DefLambdaepssigmany}
    \Lambda_{\eps}(z)_{\sigma \nu} =g_{\eps, \sigma}^{-1} \delta_{\sigma,\nu} - B_{\eps, \sigma}R_0(z) A_{\eps, \nu}^{*} , \qquad \qquad \sigma, \nu \in \II 
\end{align} 
are bounded operators and that $\Lambda_{\eps}(z)$ has a bounded inverse  if and only if $z\in \rho(H_\eps)$. Then
\begin{equation}\label{Krein}
     (H_{\eps}+z)^{-1} = R_0(z) + \sum_{\sigma, \nu \in \II} 
\left(A_{\eps, \sigma}R_0(\overline{z})\right)^{*} (\Lambda_{\eps}(z)^{-1})_{\sigma \nu}\, B_{\eps, \nu} R_0(z). 
\end{equation}

In Section 3 we prove existence of the limit
\begin{equation}\label{S(z)} 
     S(z)_{\sigma}= \limeps \; A_{\eps, \sigma}R_0(z), \qquad \qquad \sigma \in \II 
\end{equation}
for some, and hence for all $z \in \rho(H_0)$. This works in all dimensions $d \in \{1,2,3\}$. Obviously, $\limeps \; B_{\eps, \sigma}R_0(z)=J_\sigma S(z)_{\sigma}$. We will see that 
\begin{equation}\label{S(z)2} 
    S(z)_{\sigma}\Psi= v_{\sigma} \otimes G(z)_{\sigma}\Psi,
\end{equation}
where $G(z)_{\sigma}= T_{\sigma} R_0(z)$  and $T_\sigma$ is the trace operator introduced in \eqref{trace}.

The hard part, which is non-trivial even for   $N=1$, is the convergence of  $\Lambda_{\eps}(z)^{-1}$ in space dimensions $d \geq 2$.  The limit is, in general, not the inverse of an operator.  Our analysis is based on the decomposition 
\begin{equation*}
      \Lambda_{\eps}(z) =\Lambda_{\eps}(z)_{\textup{diag}}  + \Lambda_{\eps}(z)_{\textup{off}} 
\end{equation*}
into diagonal and off-diagonal parts of the operator matrix \eqref{DefLambdaepssigmany}. We then show that both $(\Lambda_{\eps}(z)_{\textup{diag}})^{-1}$ and $(\Lambda_{\eps}(z)_{\textup{diag}})^{-1} \Lambda_{\eps}(z)_{\textup{off}}$ have limits as $\eps\to 0$. The diagonal parts $B_{\eps, \sigma}R_0(z) A_{\eps, \sigma}^{*} $ contain a divergent contribution that must be cancelled by the divergence of $g_{\eps, \sigma}^{-1}$ as $\eps\to 0$. It turns out that $(g_{\eps, \sigma}^{-1}\mspace{-1mu}-\mspace{-1mu}B_{\eps, \sigma}R_0(z) A_{\eps, \sigma}^{*})^{-1}$ has a vanishing limit unless $\int\mspace{-2mu}V_\sigma(r)\,\textup{d}r = 2\pi a_\sigma$.
In the end, we arrive at
\begin{align}\label{Lambdafac}
\lim_{\eps \to 0} \,(\Lambda_{\eps}(z)^{-1})_{\sigma \nu} = 
\begin{cases} \dfrac{\Ket{u_{\sigma}}\Bra{v_{\nu}}}{\Braket{u_{\sigma}| v_{\sigma}} \Braket{u_{\nu}| v_{\nu}}} \otimes (\Theta(z)^{-1})_{\sigma \nu}  \qquad \quad &\textup{if}\; \sigma, \nu \in \JJ \\
\;\; 0 \qquad &\textup{else}
\end{cases}
\end{align}
for large enough $z>0$, with some closed and invertible operator  $\Theta(z)$ in the Hilbert space $\XX$.  Combining \eqref{Lambdafac} with \eqref{S(z)} and \eqref{S(z)2}, it follows that the expression on the right hand side of \eqref{Krein} has the limit \eqref{Hres}. It is then a standard argument to show that \eqref{Hres} defines the resolvent of a self-adjoint operator.

\section{The limit of $A_{\eps, \sigma}R_0(z)$}  \label{sec3}
Let $\eps>0$, $z \in \rho(H_0)$ and $\sigma \in \II$. Then the operator $A_{\eps, \sigma}R_0(z)$ is bounded because, by \eqref{Aepsijnew}, $A_{\eps, \sigma}$ is closed and $D(H_0) \subset D(A_{\eps, \sigma})$. In this section we prove that
\begin{align}
S(z)_{\sigma}= \limepsr A_{\eps, \sigma}R_0(z), \qquad \qquad \quad z \in \rho(H_0) \label{LimitSsigma} 
\end{align}
exists in the norm of $\LL(\HH, \widetilde{\XX}_{\sigma})$.
In fact, we have an explicit expression for $S(z)_{\sigma}$ in terms of the trace $T_{\sigma}$ defined in \eqref{trace}. Clearly,
\begin{align}
 T_{\sigma} = \tau \KK_{\sigma}, \label{DefTsigma}
\end{align}  
where $\KK_{\sigma} \in \LL(\HH,  \widetilde{\XX}_{\sigma})$ is the change of coordinates defined in \eqref{DefK} and,
for smooth functions $\psi$ with compact support,
\begin{equation}
\left( \tau \psi \right)(R,x_3,...,x_N):= \psi(r,R,x_3,...,x_N)\vert_{r=0}. \nonumber
\end{equation}
We shall use \eqref{DefTsigma} for the definition of $T_{\sigma}$, where the trace $\tau$ is defined by 
\begin{align}
\widehat{\tau \psi}(\underline{P}) := \frac{1}{2\pi} \int\limits_{\R^2} \mspace{-3mu} \widehat{\psi}(p,\underline{P}) \; \d{p},
\label{Deftau}
\end{align}
$(p,\underline{P})$ being conjugate to $(r,R,x_3,...,x_N)$, and
\begin{align}
D(\tau)= \left\{ \psi \in L^2(\R^{2N}) \, \Bigg \vert \, \int \mspace{-3mu} \left( \int \mspace{-3mu} |\widehat{\psi}(p,\underline{P})| \; \d{p} \right)^2  \d{\underline{P}}  < \infty \right\}. \nonumber
\end{align}
For lack of a good reference, we state and prove the following result about $\tau$:

\begin{lemma} \label{lm3.1}
For all $s>1$, $H^s(\R^{2N}) \subset D(\tau)$ and $\tau: H^s(\R^{2N}) \rightarrow H^{s-1}(\R^{2N-2})$ is a bounded operator. 
\end{lemma}

\noindent\emph{Remark:} From the Definition \eqref{DefTsigma} of $T_{\sigma}$, from $\KK_{\sigma}H^2(\R^{2N})=H^2(\R^{2N})$, and from Lemma \ref{lm3.1}, it follows that $H^2(\R^{2N}) \subset D(T_{\sigma})$.

\begin{proof}
We leave it to the reader to check that $H^s(\R^{2N}) \subset D(\tau)$. For $\psi \in H^s(\R^{2N})$, we have, by Definition \eqref{Deftau} and by Cauchy-Schwarz inequality,
\begin{align}
|\widehat{\tau \psi}(\underline{P})|^2 &\leq \frac{1}{4\pi^2} \int \mspace{-3mu} |\widehat{\psi}(p,\underline{P})|^2
(1+ |p|^2 + |\underline{P}|^2)^s \; \d{p} \cdot \int \mspace{-3mu} (1+ |p|^2 + |\underline{P}|^2)^{-s} \; \d{p}, \label{tauineq1}
\end{align}
where
\begin{align}
 \frac{1}{4\pi^2}\int \mspace{-3mu} (1+ |p|^2 + |\underline{P}|^2)^{-s} \; \d{p} = C_s (1+|\underline{P}|^2)^{1-s}, \label{tauineq2}
\end{align}
and $C_s < \infty$ because $s>1$. From \eqref{tauineq1} and \eqref{tauineq2}, it follows that $\| \tau \psi \|_{H^{s-1}} \leq C_{s}^{1/2}  \| \psi \|_{H^{s}}$. 
\end{proof}

In view of $A_{\eps, \sigma}=[(\eps^{-1}v_{\sigma}U_{\eps}) \otimes 1] \KK_{\sigma}$, proving \eqref{LimitSsigma} is a problem in $\LL(L^2(\R^2))$, which we solve in the following lemma. For a given pair $\sigma$, let $V=V_{\sigma}$, $v=v_{\sigma}$ and let $\Ket{v}\Bra{\delta}: H^2(\R^2) \rightarrow L^2(\R^2)$ denote the rank-one operator given by
\begin{align}
\Ket{v}\Bra{\delta} \psi:= \psi(0) v, \nonumber
\end{align}
where $\psi(0):=(2\pi)^{-1} \int \widehat{\psi}(p) \, \d{p}$. Then the following holds:
\begin{lemma} \label{lm3.2}
If $\int (1+|r|^{2s}) \left|V(r)\right|\,\d{r} < \infty$ for some $s \in (0,1)$, then
\begin{align}
v \eps^{-1} U_{\eps} \rightarrow \Ket{v}\Bra{\delta}, \qquad \qquad (\eps \downarrow 0) \nonumber
\end{align}
in the norm of $\LL(H^2(\R^2), L^2(\R^2))$ and the rate of convergence is at least $O(\eps^s)$.
\end{lemma}
\begin{proof}
Due to the Sobolev embedding $H^2(\R^2) \hookrightarrow C^{0,s}(\R^2)$, valid for $s \in (0,1)$, we have that
\begin{align}
\left|\left(v \eps^{-1}U_{\eps} \psi - \Ket{v}\Bra{\delta} \psi \right)(r)\right| &= \left| v(r) \right| \left| \psi(\eps r)- \psi(0) \right| \nonumber \\
& \leq  \left| v(r) \right| c_s |\eps r|^{s} \| \psi \|_{H^2}. \nonumber
\end{align}
This implies the statement of the lemma.
\end{proof}

We now come to the proof of \eqref{LimitSsigma}:
\begin{prop} \label{prop3.3}
If $z \in \rho(H_0)$ and $\int (1+|r|^{2s}) \left|V_{\sigma}(r)\right|\,\d{r} < \infty$ for some $s \in (0,1)$, then, as $\eps\downarrow 0$, \vspace*{-4mm}
\begin{align}
\left\| A_{\eps, \sigma}R_0(z) - S(z)_{\sigma} \right\| = O(\eps^s), \nonumber
\end{align}
where $S(z)_{\sigma} \in \LL(\HH, \widetilde{\XX}_{\sigma})$ is given by
\begin{align}
S(z)_{\sigma}\psi = v_{\sigma} \otimes (G(z)_{\sigma}\psi) \label{DefSsigma}
\end{align}
and $G(z)_{\sigma}= T_{\sigma}R_0(z) \in \LL(\HH, \XX_{\sigma})$.
\end{prop}

\begin{proof}
By \eqref{Aepsijnew}, we have that 
$A_{\eps, \sigma}R_0(z) =  (v_{\sigma} \otimes 1) \,\eps^{-1} U_{\eps} \KK_{\sigma}R_0(z)$. Comparing this with
$S(z)_{\sigma}\psi = v_{\sigma} \otimes (\tau \KK_{\sigma} R_0(z)\psi) = (\Ket{v_{\sigma}}\Bra{\delta} \otimes 1) \KK_{\sigma} R_0(z)\psi$,
we see that the proposition follows from Lemma \ref{lm3.2} since $\KK_{\sigma}$ defines a bounded operator in $H^2(\R^{2N})$.
\end{proof}
\section{Convergence of $\Lambda_{\eps}(z)^{-1}$} \label{sec4}
The goal of this section is to show that the limit $\lim_{\eps \to 0} \Lambda_{\eps}(z)^{-1}$ exists for all
large enough $z>0$, provided that, for all pairs $\sigma$,  $V_{\sigma} \in L^1 \cap L^2(\R^2)$ and $\int |r|^{2s}\left|V_{\sigma}(r)\right| \, \d{r} < \infty$ for some $s > 0 $.
As explained at the end of Section \ref{sec2}, this proof is based on a decomposition of $\Lambda_{\eps}(z)$ into diagonal and off-diagonal parts that we now define.

Recall from \eqref{DefLambdaepssigmany} that, for given $\eps, z>0$, the operator matrix $\Lambda_{\eps}(z) \in \LL(\widetilde{\XX})$ has the components
\begin{align}
\Lambda_{\eps}(z)_{\sigma \nu}=g_{\eps, \sigma}^{-1} \delta_{\sigma,\nu}- \phi_{\eps}(z)_{\sigma \nu} , \qquad \qquad \sigma, \nu \in \II, \nonumber
\end{align} 
where
\begin{align}
\phi_{\eps}(z)_{\sigma \nu}:=B_{\eps, \sigma}R_0(z)A_{\eps, \nu}^{*} \nonumber
\end{align}
defines a bounded operator $\phi_{\eps}(z)_{\sigma \nu} \in \LL(\widetilde{\XX}_{\nu},\widetilde{\XX}_{\sigma})$ for all pairs $\sigma, \nu \in \II$.
Using the definitions  \eqref{Bepsijnew} and \eqref{Aepsijnew} for $B_{\eps, \sigma}$ and $A_{\eps, \nu}$, respectively, leads to the explicit formula
\begin{align}
\phi_{\eps}(z)_{\sigma \nu}= \eps^{-2} \, (u_{\sigma} \otimes 1) \, U_{\eps} \KK_{\sigma}R_0(z) \KK_{\nu}^{*} U_{\eps}^{*} (v_{\nu} \otimes 1). \label{Defphisigmany}
\end{align}
We shall see that, for fixed $z>0$, the off-diagonal operators $\phi_{\eps}(z)_{\sigma \nu} \;(\sigma \neq \nu)$ are uniformly bounded in $\eps>0$, whereas the specific behavior of the coupling constant $g_{\eps, \sigma}$ (see \eqref{g-def}) cancels the singular part in the diagonal operators $\phi_{\eps}(z)_{\sigma \sigma}$ if and only if $a_{\sigma} = \int  V_{\sigma}(r) \, \d{r} / (2 \pi) > 0$. Therefore, we write the operator matrix $\Lambda_{\eps}(z)$ in the form
\begin{align}
\Lambda_{\eps}(z) =\Lambda_{\eps}(z)_{\textup{diag}}  + \Lambda_{\eps}(z)_{\textup{off}}, \label{Lambdadecomp}
\end{align}
where the diagonal part $\Lambda_{\eps}(z)_{\textup{diag}}$ and the off-diagonal part $\Lambda_{\eps}(z)_{\textup{off}}$ are given by
\begin{align}
\left(\Lambda_{\eps}(z)_{\textup{diag}}\right)_{\sigma \nu} = \left( g_{\eps, \sigma}^{-1} - \phi_{\eps}(z)_{\sigma \sigma} \right) \delta_{\sigma \nu}, \qquad \qquad \sigma, \nu \in \II, \label{DefLambdadiag}
\end{align}
and
\begin{align}
\left(\Lambda_{\eps}(z)_{\textup{off}}\right)_{\sigma \nu} =
\phi_{\eps}(z)_{\sigma \nu}(\delta_{\sigma \nu}-1), \qquad \qquad \sigma, \nu \in \II,
\label{DefLambdaoff}
\end{align}
respectively. Section \ref{sec4.1} is devoted to $ \Lambda_{\eps}(z)_{\textup{diag}} $ and Section \ref{sec4.2} is devoted to $ \Lambda_{\eps}(z)_{\textup{off}} $.
\subsection{The limit of $(\Lambda_{\eps}(z)_{\textup{diag}})^{-1}$} \label{sec4.1}
We first show that $ \Lambda_{\eps}(z)_{\textup{diag}} $ is invertible for small enough $\eps > 0$ and large enough $z>0$ and then we compute the limit $ \lim_{\eps \to 0} (\Lambda_{\eps}(z)_{\textup{diag}})^{-1}$. This can be done for each component 
$\Lambda_{\eps}(z)_{\sigma \sigma} = g_{\eps, \sigma}^{-1} - \phi_{\eps}(z)_{\sigma \sigma}$ separately and without loss of generality we may choose $\sigma=(1,2)$. 
In the following, $v:=v_{(1,2)}$, $\KK:=\KK_{(1,2)}$, $\mu:=\mu_{(1,2)}$, $\phi_{\eps}(z):=\phi_{\eps}(z)_{(1,2)(1,2)},$ etc.

We begin by computing the kernel of  $\phi_{\eps}(z)$. To this end, we note that 
\begin{align}
\KK R_0(z) = (\widetilde{H}_0+z)^{-1}\KK, \label{H0H0Tilde}
\end{align}
where $\widetilde{H}_0$ is the free Hamiltonian $H_0$ expressed in the relative and center of mass coordinates of the pair $(1,2)$ (cf. \eqref{DefrR}), i.e.
\begin{align}
\widetilde{H}_0 = -\frac{\Delta_r}{\mu} - \frac{\Delta_R}{m_1+m_2} + 
\sum\limits_{i=3}^{N} \left( -\frac{\Delta_{x_i}}{m_i} \right). \label{DefH0Tilde}
\end{align} 
Starting from the explicit formula \eqref{Defphisigmany}, the identity \eqref{H0H0Tilde} and $\KK \KK^{*}=1$ imply that
\begin{align}
\phi_{\eps}(z)= \eps^{-2} \, (u \otimes 1) \, U_{\eps} (\widetilde{H}_0 + z)^{-1} U_{\eps}^{*} (v \otimes 1). \nonumber
\end{align}
Hence, after a Fourier transform in $(R,x_3,...,x_N)$, $\phi_{\eps}(z)$ acts pointwise in the conjugate variable $\underline{P}=(P,p_3,...,p_N)$ by the operator 
\begin{align}
\phi_{\eps}(z, \underline{P})&= \mu \eps^{-2}\,u \, U_{\eps} (-\Delta_r + \mu(z + Q))^{-1} U_{\eps}^{*} v \nonumber \\
&= \mu \, u\, (-\Delta_r + \eps^2\mu(z + Q))^{-1} \, v \quad \in \LL(L^2(\R^2)), \label{phi1212P} 
\end{align}
where 
\begin{align}
Q:=\frac{P^2}{m_1+m_2}+ \displaystyle\sum\limits_{i=3}^{N} \frac{p_i^2}{m_i}, \label{DefQ}
\end{align} and the scaling property of the Laplace operator
with respect to the unitary scaling $U_{\eps}$ has been used in \eqref{phi1212P}.
As the resolvent in \eqref{phi1212P} is divergent in the limit $\eps \downarrow 0$, the first step in the analysis of $\lim_{\eps \to 0} (\Lambda_{\eps}(z)_{\textup{diag}})^{-1}$ is to study the asymptotic behavior of $\phi_{\eps}(z, \underline{P})$ for small $\eps>0$.

A similar but less specific version of the following lemma can be found in \cite[Proposition 3.2]{Simon76}. 
\begin{lemma}\label{lm4.1}
Let $V \in L^1 \cap L^2(\R^2)$ satisfy $\int  |r|^{2s}\left|V(r)\right|\,\d{r}< \infty$ for some $s \in (0,2)$.
Then there exists a constant $C=C(s,V)>0$ such that
\begin{align}
\forall \alpha >0: \qquad &\left\| u \, (-\Delta + \alpha)^{-1} \, v + (2\pi)^{-1} \left[ \left(\ln\left(\sqrt{\alpha}/2\right) + \gamma\right)\mspace{-2mu} \Ket{u} \mspace{-3mu}\Bra{v} + L  \right] \right\|_{\textup{HS}} \leq C \alpha^{s/2}, \label{alphaest1} \\
\forall \alpha \geq 1: \qquad &\left\| u \, (-\Delta + \alpha)^{-1} \, v \right\|_{\textup{HS}} \leq C, \label{alphaest2}
\end{align}
where $\|\cdot\|_{\textup{HS}}$ denotes the Hilbert-Schmidt norm in $L^2(\R^2)$, $\gamma$ is the Euler–Mascheroni constant and $L$ is a Hilbert-Schmidt operator in $L^2(\R^2)$ that is defined in terms of the kernel
\begin{align}
u(r) \, \ln(|r-r'|) \, v(r'), \qquad \qquad r \neq r'. \nonumber
\end{align}
\end{lemma}

\begin{proof} First, we observe that $u(-\Delta+\alpha)^{-1}v$ is associated with the integral kernel
\begin{align}
u(r)\, G_{\alpha}(r-r') \,v(r'),    \label{Lkernel}
\end{align}
where $G_{\alpha}=G_{\alpha}^2$ denotes the Green's function of $-\Delta + \alpha: H^2(\R^2) \rightarrow L^2(\R^2)$ (we refer to the appendix for details about $G_{\alpha}$). By \cite[Chapter I.5]{SolvableModels}, we have the explicit description
\begin{align}
G_{\alpha}(x)=G_{1}(\sqrt{\alpha}x)=\frac{i}{4} H^{(1)}_0(i\sqrt{\alpha}|x|),  \qquad\qquad 0\neq x \in \R^2,\label{GreenHankel}
\end{align}
where $H^{(1)}_0$ denotes the Hankel function of first kind and order zero.
It is known (see e.g. \cite[Chapter 9.1]{HandbookAS}) that
\begin{align}
iH^{(1)}_0(ir)&=-2\pi^{-1} \left(\ln\left(r/2\right)+\gamma \right) + O(r^2|\ln(r)|), \qquad \qquad (r \downarrow 0). \label{Hankelnull}
\end{align}
As $G_1$ is smooth in $\R^2\setminus \{0\}$ and exponentially decaying as $|x| \rightarrow \infty$ (see Lemma \ref{lmA2}), we conclude from \eqref{GreenHankel} and \eqref{Hankelnull} that there is some constant $\lambda=\lambda(s)>0$ such that
\begin{align}
\left|G_1(x)+(2\pi)^{-1} \left(\ln\left(|x|/2\right)+\gamma\right)\right| \leq \lambda|x|^s, \qquad \qquad x  \in \R^2\setminus \{0\}. \nonumber
\end{align}
Using again \eqref{GreenHankel}, this implies that
\begin{align}
\left|G_\alpha(x)+(2\pi)^{-1} \left(\ln\left(\sqrt{\alpha}|x|/2\right)+\gamma\right)\right| \leq \lambda\alpha^{s/2}|x|^s \nonumber
\end{align}
uniformly in $\alpha>0$ and $x  \in \R^2\setminus \{0\}$. Therefore, \eqref{alphaest1} follows from
\begin{align}
&\left\| u \, (-\Delta + \alpha)^{-1} \, v + (2\pi)^{-1} \left[ \left(\ln\left(\sqrt{\alpha}/2\right) + \gamma\right)\mspace{-2mu} \Ket{u} \mspace{-3mu}\Bra{v} + L  \right] \right\|_{\textup{HS}}^2 \nonumber \\
& \quad=\int \mspace{-3mu} \d{r} \,\d{r'} \; |u(r)|^2 \left| G_{\alpha}(r-r') + (2\pi)^{-1} \left[\ln\left(\sqrt{\alpha}|r-r'|/2\right)+\gamma\right] \right|^2 |v(r')|^2 \nonumber\\
&\quad\leq \lambda^{2} \alpha^{s}  \int \mspace{-3mu} \d{r} \,\d{r'} \; |V(r)| \left| r-r' \right|^{2s} |V(r')| \nonumber \\
&\quad\leq 2^{2s+1}\lambda^{2} \alpha^{s} \|V\|_{L^1}  \int |r|^{2s} |V(r)|\, \d{r},
\nonumber 
\end{align}
where the elementary inequality $(a+b)^{2s} \leq 2^{2s}(a^{2s}+b^{2s})\;(a,b >0)$ was used for the last inequality. To show that $L$ indeed defines a Hilbert-Schmidt operator, we note that $\left\| L \right\|_{\textup{HS}}^2=I_1+I_2$, where
\begin{align}
I_1&=\int\limits_{|r| \leq 1} \mspace{-10mu} \d{r} \; \ln(|r|)^2 \int \mspace{-3mu} \d{r'} \; |V(r+r')| |V(r')| \leq 2\pi \|V\|^2 \int\limits_{0}^{1}  x\ln(x)^2 \, \d{x} < \infty \nonumber
\end{align}
by the Cauchy-Schwarz inequality and
\begin{align}
I_2&=\int\limits_{|r| \geq 1} \mspace{-10mu} \d{r} \;\ln(|r|)^2 \int \mspace{-3mu} \d{r'} \;  |V(r+r')| |V(r')| \nonumber \\
& \leq \sup_{x\geq 1}\left(x^{-2s}\ln(x)^2\right)\int \mspace{-3mu} \d{r} \, \d{r'} \;  |r|^{2s} |V(r+r')| |V(r')| \nonumber \\
&= \sup_{x\geq 1}\left(x^{-2s}\ln(x)^2\right)\int \mspace{-3mu} \d{r} \, \d{r'} \;  |r-r'|^{2s} |V(r)| |V(r')| \nonumber \\
& \leq 2^{2s+1} \sup_{x\geq 1}\left(x^{-2s}\ln(x)^2\right)\|V\|_{L^1}\int  |r|^{2s}\left|V(r)\right|\, \d{r} < \infty.  \nonumber
\end{align}

It remains to prove \eqref{alphaest2}. To this end, we note that for fixed $x \neq 0$ the function $\alpha \rightarrow G_{\alpha}(x)$ is strictly monotonically decreasing in $\alpha>0$ (cf. Lemma \ref{lmA1} $(v)$). Therefore, the Cauchy-Schwarz inequality yields for $\alpha\geq 1$ the uniform estimate
\begin{align}
\left\| u \, (-\Delta + \alpha)^{-1} \, v \right\|_{\textup{HS}}^2 &\leq \int \mspace{-3mu} \d{r} \,\d{r'} \; |V(r)| \left| G_{1}(r-r')\right|^2 |V(r')|  \nonumber \\ &\leq \|V\|^2 \|G_{1}\|^2 < \infty, \nonumber
\end{align}
where we used that $G_1 \in L^2(\R^2)$ (cf. Lemma \ref{lmA1} $(iv)$). 
\end{proof}

In the next lemma we show that $g_{\eps}^{-1} - \phi_{\eps}(z)$ is invertible for all large enough $z>0$ and small enough $\eps >0$. It is here, where the asymptotics  \eqref{g-def} of $g_{\eps}$, i.e. 
\begin{align}
g_{\eps}^{-1}=-\mu(a \ln(\eps) +b) + o(1), \qquad \qquad (\eps \downarrow 0)  \label{g-def2}
\end{align}
for some $a>0$, will play a crucial role.
\begin{lemma}\label{lm4.2}
Let the hypotheses of Lemma \ref{lm4.1} be satisfied. Then there exist $\eps_0,z_0>0$ such that $g_{\eps}^{-1}-\phi_{\eps}(z,\underline{P})$
is invertible for all $\underline{P}$, $\eps \in (0, \eps_0)$ and $z > z_0$. Moreover,
\begin{align}
\left\|\left(g_{\eps}^{-1}- \phi_{\eps}(z,  \underline{P}) \right)^{-1} \right\| \leq \widetilde{C} \begin{cases}
|\ln(\eps)|^{-1}  & \textup{if} \;\int V(r)\,\d{r}/(2\pi) < a 
\\ \max\left\{|\ln(\eps)|^{-1}, \ln(\mu(z+Q))^{-1}\right\}
& \textup{if} \; \int V(r)\,\d{r}/(2\pi) = a,
\end{cases} \label{invdiagabs}
\end{align}
where the constant $\widetilde{C}>0$ is independent of $\underline{P}$, $z> z_0$ and $\eps \in (0,\eps_0) $.
\end{lemma}

\begin{proof}
Let $C>0$ denote the constant from Lemma \ref{lm4.1}.
By \eqref{g-def2}, we find $\eps_0 \in (0,1)$ such that $|g_{\eps} \mu C| \leq 1/2$ and $|g_{\eps}| \leq \lambda |\ln(\eps)|^{-1}$ for all $\eps \in (0, \eps_0)$, where the constant $\lambda>0$ is independent of $\eps \in (0, \eps_0)$. 

If $\alpha_{\eps}:=\eps^2\mu(z+Q) \geq 1$, then it follows from \eqref{phi1212P} and \eqref{alphaest2} that
\begin{align}
\left\| g_{\eps} \phi_{\eps}(z,\underline{P}) \right\| \leq 
|g_{\eps}| \mu \left\| u \, (-\Delta + \alpha_{\eps})^{-1} \, v \right\|
\leq |g_{\eps} \mu C| \leq \Scalefrac{1}{2}.  \nonumber
\end{align}
Hence, $g_{\eps}^{-1} -\phi_{\eps}(z,\underline{P})$ is invertible and
\begin{align}
\left\|\left(g_{\eps}^{-1} - \phi_{\eps}(z,\underline{P})\right)^{-1} \right\| =\left\|g_{\eps}\left(1 - g_{\eps}\phi_{\eps}(z,\underline{P})\right)^{-1} \right\| \leq 2 \lambda |\ln(\eps)|^{-1} \nonumber
\end{align}
for all $\eps \in (0, \eps_0)$. This establishes \eqref{invdiagabs} in the case of $\alpha_{\eps}  \geq 1$.

In all the following, $\alpha_{\eps}=\eps^2\mu(z+Q) < 1$. Then, it follows from \eqref{phi1212P} and \eqref{alphaest1} that
\begin{align}
\phi_{\eps}(z,  \underline{P}) = -\frac{\mu}{4\pi} \ln(\alpha_{\eps}) \Ket{u}\Bra{v} + O(1), \label{phi1212asymp}
\end{align}
where the norm of the operator $O(1)$ is uniformly bounded in  $\underline{P}$, $z>0$ and $\eps \in (0,\eps_0)$ as long as $\alpha_{\eps} < 1$.
To determine the inverse of $g_{\eps}^{-1}-\phi_{\eps}(z,\underline{P})$, we apply the identity (cf. \cite[Chapter I.1, Eq. (1.3.47)]{SolvableModels})
\begin{align}
\left( 1 +R +\beta \,\Ket{\widetilde{\eta}}  \Bra{\eta}\right)^{-1} \mspace{-5mu} = (1+R)^{-1} - \frac{ (1+R)^{-1}\Ket{\widetilde{\eta}}  \Bra{\eta}(1+R)^{-1} }{\beta^{-1}+
\Braket{\eta|(1+R)^{-1}\widetilde{\eta}}}, \label{Resid}
\end{align}
which is valid under the assumption that $\beta \in \C$, $\widetilde{\eta}, \eta \in \mathcal{H}$, $R,(1+R)^{-1} \in 
\LL(\mathcal{H})$ and $\beta^{-1}+\Braket{\eta|(1+R)^{-1}\widetilde{\eta}} \neq 0$ in some separable Hilbert space $\mathcal{H}$. This identity (with $R=0$) and $|g_{\eps}|= O(|\ln(\eps)|^{-1})$ yield
\begin{align}
\left( g_{\eps}^{-1}+ \frac{\mu}{4\pi} \ln(\alpha_{\eps}) \Ket{u}\Bra{v} \right)^{-1} &=
g_{\eps} \left( 1+\frac{g_{\eps}\mu}{4\pi} \ln(\alpha_{\eps}) \Ket{u}\Bra{v} ) \right)^{-1}  \nonumber \\
&= g_{\eps}^2\,f(z,\underline{P}, \eps)^{-1} \Ket{u}\Bra{v} + O(|\ln(\eps)|^{-1}), \label{invasymp}
\end{align}
provided that
\begin{align}
f(z,\underline{P}, \eps):=-\frac{4\pi}{\mu\ln(\alpha_{\eps})}-g_{\eps}\int V(r) \,\d{r}\neq 0.\nonumber
\end{align}
To derive a lower bound on $f(z,\underline{P}, \eps)$, we consider the two cases $\int  V(r) \, \d{r}/(2\pi) < a$ and $\int  V(r) \, \d{r}/(2\pi) = a$ separately. In both cases we assume that $z > z_0$ and $\eps \in (0, \eps_0)$ for some $z_0  \geq \mu^{-1}$ and $\eps_0 \in (0,1)$, which we will fix below. This implies that $\mu (z+Q) \geq 1$ and hence
\begin{align}
0 < -\ln(\alpha_{\eps}) \leq  2 |\ln(\eps)|. \label{lnest}
\end{align}

If  $\int  V(r) \, \d{r}/(2\pi) < a$, then we use \eqref{g-def2} and \eqref{lnest} to derive the lower bound
\begin{align}
f(z,\underline{P}, \eps) &\geq (\mu|\ln(\eps)|)^{-1} \left(2\pi  - \frac{1}{a} \int  V(r) \, \d{r} \right) + O(\ln(\eps)^{-2}) \qquad \qquad (\eps \downarrow 0). \label{flowerbound}
\end{align}
Hence, there exists $\eps_0 \in (0,1)$ such that $f(z,\underline{P}, \eps) > 0$ for all  $\underline{P}$, $z > z_0\geq \mu^{-1}$ and $\eps \in (0, \eps_0)$. Moreover, inserting \eqref{flowerbound} in \eqref{invasymp} and using $g_{\eps}^2= O(\ln(\eps)^{-2})$  shows that
\begin{align}
\left\| \left( g_{\eps}^{-1}+\frac{\mu}{4\pi} \ln(\alpha_{\eps}) \Ket{u}\Bra{v} \right)^{-1}\right\| =O(|\ln(\eps)|^{-1}) \qquad \qquad (\eps \downarrow 0). \nonumber
\end{align}
In view of \eqref{phi1212asymp}, this  proves \eqref{invdiagabs} for $\int  V(r)\, \d{r} / (2\pi) <  a$ and $\alpha_{\eps}<1$.

If  $\int  V(r) \, \d{r}/(2\pi) = a$, then we again use \eqref{g-def2} and \eqref{lnest} to derive a more refined version of \eqref{flowerbound}:
\begin{align}
|\ln(\eps)|^{2} |f(z,\underline{P}, \eps)| &= \frac{|\ln(\eps)|^2}{|\ln(\alpha_{\eps})|} \left|4 \pi \mu^{-1} + g_{\eps} \ln(\alpha_{\eps}) \int V(r) \, \d{r}\right| \nonumber \\
& \geq \frac{|\ln(\eps)|}{2\mu} \left|4 \pi  + g_{\eps} \mu \ln(\alpha_{\eps}) \int V(r) \, \d{r}\right| \nonumber \\
& = \frac{|\ln(\eps)|}{2\mu } \left|4 \pi  - \frac{\int  V(r) \,\d{r}}{a \ln(\eps)} \bigg( 2\ln(\eps) + \ln(\mu(z+Q))\bigg) + O\left(\frac{|\ln(\alpha_{\eps})|}{\ln(\eps)^2} \right) \right|\nonumber \\
& = \pi \mu^{-1} \ln(\mu(z+Q)) + O(1),  \nonumber
\end{align}
where the remainder $O(1)$ is uniformly bounded in  $\underline{P}$, $z > \mu^{-1} $ and $\eps \in (0,\eps_0)$.
By choosing $z_0 \geq \mu^{-1}$ large enough, we conclude that $|f(z,\underline{P}, \eps)| \geq \pi\ln(\mu(z+Q)) /(2\mu\ln(\eps)^2)>0$  for all $\underline{P}$, $z > z_0$ and $\eps \in (0,\eps_0)$. With the help of this bound, the desired estimate for the operator norm of $\left(g_{\eps}^{-1}- \phi_{\eps}(z,  \underline{P}) \right)^{-1}$ can be obtained similarly to the case  $\int  V(r) \, \d{r}/(2\pi) < a$. 
\end{proof}

An immediate consequence of Lemma \ref{lm4.2} is that $g_{\eps}^{-1}- \phi_{\eps}(z)$ is invertible for all small enough $\eps>0$ and large enough $z>0$. Yet, we shall see that the limit $\eps \to 0$ of the inverse operator essentially depends on the leading order of the coupling constant $g_{\eps}$ or, more precisely, on $a$. If $ \int V(r) \,  \d{r} /(2\pi) < a$, then it follows from \eqref{invdiagabs} that 
$ \lim_{\eps \to 0} \left(g_{\eps}^{-1}- \phi_{\eps}(z) \right)^{-1}$ vanishes for $z> z_0 > 0$, while for $ \int V(r) \,  \d{r} /(2\pi) = a$ this limit will turn out to be non-trivial. To prove this, we first recall that \eqref{phi1212P} and \eqref{Lkernel} imply that for fixed $\underline{P}$ the operator $\phi_{\eps}(z,\underline{P})$ is associated with the kernel
\begin{align}
\mu \mspace{1mu} u(r) \, G_{\eps^2\mu(z+Q)}(r-r') \, v(r'). \nonumber
\end{align}
In the notation of \cite[Chapter I.5]{SolvableModels} for the one-particle case, $(-g_{\eps}) \cdot \phi_{\eps}(z,\underline{P})$ agrees with the operator $B_{\eps}(k)$ for $k^2:=-\mu(z+Q)<0$, $\lambda_1:=1/a$ and $\lambda_2:=-b/a^2$. In view of \cite[Chapter I.5, Eq. (5.61)]{SolvableModels}\footnote[1]{There is a typo in \cite[Chapter I.5, Eq. (5.61)]{SolvableModels}: The term in curly braces is to be inverted.}, we thus expect that
\begin{align}
\lim_{\eps \to 0} \left(g_{\eps}^{-1}- \phi_{\eps}(z) \right)^{-1} =  \frac{\Ket{u}\Bra{v}}{\Braket{u|v}^2} \otimes D(z)^{-1}, \label{diaglim12}
\end{align}
where $D(z)$ is an (unbounded) operator in $\XX_{\sigma}$, which,
after passing to Fourier space, acts as the multiplication operator
\begin{align}
D(z,\underline{P})&:= \frac{\mu}{2\pi} \left( \ln\left( \frac{\sqrt{\mu(z+Q)}}{2}\right) + \gamma + 2 \pi \alpha \right), \label{DefD(z,P)} \\
\alpha &:=-\frac{b}{\Braket{u|v}} + \frac{\Braket{v|Lu}}{2 \pi  \Braket{u|v}^2}. \label{Defalpha}
\end{align} 
Here, $\gamma$ denotes the Euler–Mascheroni constant and $L$ is the Hilbert-Schmidt operator defined in Lemma \ref{lm4.1}. We now turn to the proof of \eqref{diaglim12}:
\begin{prop} \label{prop4.3}
Let $\int V(r) \, \d{r} /(2\pi) = a$ and let the hypotheses of Lemma \ref{lm4.1} be satisfied. 
Then $D(z)$ is invertible and \eqref{diaglim12} holds true for all large enough $z>0$.
\end{prop} 

\begin{proof}
Clearly, $D(z)$ is invertible if $ \ln\left( \sqrt{\mu z}/2\right) + \gamma + 2 \pi \alpha >0$, i.e. $z>z_1:=4\mu^{-1} \exp(-4\pi \alpha-2\gamma)$. 
Let $\eps_0, z_0>0$ be chosen as in Lemma \ref{lm4.2} and, for the rest of this proof, let $z> \max(z_0,z_1)$ be fixed. Then \eqref{diaglim12} is equivalent to the statement that
\begin{align}
\lim_{\eps \to 0} \sup_{\underline{P}} \left\|  \left(g_{\eps}^{-1}- \phi_{\eps}(z, \underline{P}) \right)^{-1} -D(z,\underline{P})^{-1} \frac{\Ket{u}\Bra{v}}{\Braket{u|v}^2}   \right\| = 0  \label{Puniformest}
 \end{align}
 with $\left\|\,\cdot\, \right\|$ denoting the operator norm in $\LL(L^2(\R^2))$.

The idea of the following proof is that both operators in the difference vanish as $|\underline{P}| \rightarrow \infty$, while for $|\underline{P}| \leq K$ the arguments from the one-particle case still work. To make this more explicit, we fix $\delta>0$. Then it follows from Lemma \ref{lm4.2} and \eqref{DefD(z,P)} that
\begin{align}
\left\|  \left(g_{\eps}^{-1}- \phi_{\eps}(z, \underline{P}) \right)^{-1} \right\| + \left\|D(z,\underline{P})^{-1}  \frac{\Ket{u}\Bra{v}}{\Braket{u|v}^2}   \right\| < \delta, \nonumber
\end{align} 
provided that $\eps>0$ is sufficiently small and $|\underline{P}|$, and hence $Q$, are sufficiently large. To prove \eqref{Puniformest}, it therefore suffices to show that for any $K>0$ there exists $\eps_K>0$ such that 
\begin{align} \sup_{\underline{P}\,:\, Q \leq K } \left\|  \left(g_{\eps}^{-1}- \phi_{\eps}(z, \underline{P}) \right)^{-1}-D(z,\underline{P})^{-1} \frac{\Ket{u}\Bra{v}}{\Braket{u|v}^2}     \right\| < \delta \label{Puniformest2}
 \end{align}
holds for all $\eps \in (0,\eps_K$).

Let $K>0$ be fixed. Then \eqref{alphaest1} implies that, for some constant $C=C(s,V)>0$,
\begin{align}
\left\| \phi_{\eps}(z, \underline{P}) + \frac{\mu}{2\pi} \left[ \left(\ln\left(
\frac{\eps}{2}\sqrt{\mu(z+Q)}\right) + \gamma\right)\mspace{-2mu} \Ket{u} \mspace{-3mu}\Bra{v} + L  \right] \right\|_{\textup{HS}} \leq C\mu^{1+s/2}(z+K)^{s/2}\eps^s  \nonumber
\end{align}
uniformly in $\underline{P}$ with $Q=Q(\underline{P}) \leq K$.
Combining this with the asymptotic behavior of $g_{\eps}$ in the limit $\eps \to 0$ (see \eqref{g-def2}), we see that
\begin{align}
&  (-g_{\eps})\cdot\phi_{\eps}(z, \underline{P}) \nonumber \\
& \,= \mspace{-1mu}-\frac{1}{2 \pi a}  \Ket{u}\mspace{-2mu}\Bra{v}- \frac{1}{2\pi a\ln(\eps)}\mspace{-1mu}
  \left[ \mspace{-2mu} \left\{  \ln \mspace{-3mu} \left(\mspace{-2mu} \frac{\sqrt{\mu(z+Q)}}{2}\right)  \mspace{-2mu}+  \mspace{-2mu}\gamma \mspace{-2mu}- \mspace{-2mu} \frac{b}{a} \mspace{-2mu}\right\} \mspace{-2mu} \Ket{u}\mspace{-2mu}\Bra{v} + L  \mspace{-2mu}  \right] \mspace{-2mu} + o\left(\frac{1}{|\ln(\eps)|} \right)  \label{phieps(z)asy}
\end{align}
is valid in Hilbert-Schmidt norm uniformly in $\underline{P}$ with $Q=Q(\underline{P}) \leq K$. 
As this coincides with the expansion of $B_{\eps}(k)$ from the one-particle case (cf. \cite[Chapter I.5, Eq. (5.56)]{SolvableModels}), the proof of \eqref{Puniformest2} now follows the line of arguments from that case. For the convenience of the reader, we spell out the details in the following.

To start with, we derive from \eqref{phieps(z)asy} the equation
\begin{align}
\left(g_{\eps}^{-1}- \phi_{\eps}(z, \underline{P}) \right)^{-1} = g_{\eps} \Big(1 + \beta(Q) \Ket{u}\Bra{v} + R_{\eps}  \Big)^{-1}, \label{phi1212invasymp} 
\end{align}
where
\begin{equation}
\begin{aligned}
\beta(Q)&:= -\frac{1}{2 \pi a}\left(1 + \frac{\widetilde{\beta}(Q)}{\ln(\eps)} \right), \\
 \widetilde{\beta}(Q)&:=  \ln \mspace{-3mu} \left(\mspace{-2mu} \frac{\sqrt{\mu(z+Q)}}{2}\right)  \mspace{-1mu}+  \mspace{-1mu}\gamma \mspace{-2mu}- \mspace{-2mu} \frac{b}{a}, \\
 R_{\eps}&:= -\frac{1}{2 \pi a \ln(\eps)} \, L + o\left(|\ln(\eps)|^{-1}\right). \nonumber
\end{aligned}
\end{equation}
The expansion of $R_{\eps}$ holds uniformly in $\underline{P}$ with $Q=Q(\underline{P}) \leq K$. To obtain an expression for the inverse on the right side of \eqref{phi1212invasymp}, we now apply the identity \eqref{Resid}. This yields that
\begin{align}
\Big(1 + \beta(Q) \Ket{u}\Bra{v} + R_{\eps}  \Big)^{-1}= (1+R_{\eps})^{-1} - \frac{ (1+R_{\eps})^{-1}\Ket{u}  \Bra{v}(1+R_{\eps})^{-1}} { \beta(Q)^{-1} + \Braket{v | (1+R_{\eps})^{-1} u }}  \label{Residuse}
\end{align}
as operators in $\mathscr{L}(L^2(\R^2))$. By the definitions of $\beta(Q)$ and $R_{\eps}$, and by $2 \pi a=  \Braket{u | v} $,
\begin{align}
\beta(Q)^{-1}=-  \Braket{u | v} \left( 1-  \frac{\widetilde{\beta}(Q)}{\ln(\eps)} \right) +  O\left(\ln(\eps)^{-2}\right), \qquad \qquad Q \leq K, \label{Defbeta(Q)}
\end{align}
and 
\begin{equation}
\Braket{v | (1+R_{\eps})^{-1} u } =  \Braket{u | v} + \frac{\Braket{v|L u}}{2 \pi a \ln(\eps)} + o\left(|\ln(\eps)|^{-1}\right), \label{Repsid}
\end{equation}
where  $\| R_{\eps} \|= O(|\ln(\eps)|^{-1})$ was used. Note that the sum of 
\eqref{Defbeta(Q)} and \eqref{Repsid} is of order $1/|\ln(\eps)|$ because $\Braket{u | v}$ cancels. It follows that the second term in \eqref{Residuse} is of order $|\ln(\eps)|$. Hence we may ignore the first summand,  $(1+R_{\eps})^{-1}$,
in Equation \eqref{Residuse} and we may replace the numerator in the second summand by $ \Ket{u} \mspace{-2mu}\Bra{v}$.  It follows that, uniformly in  $\underline{P}$ with $Q=Q(\underline{P}) \leq K$, 
\begin{align}
\Big(1  \mspace{-1mu}+ \mspace{-1mu} \beta(Q) \Ket{u} \mspace{-2mu}\Bra{v} + \mspace{-2mu} R_{\eps}  \Big)^{-1} \mspace{-4mu} = -\frac{\ln(\eps)}{\Braket{u|v}}
 \bigg[  \ln \mspace{-3mu} \left(\mspace{-2mu} \frac{\sqrt{\mu(z+Q)}}{2}\right)  \mspace{-1mu}+  \mspace{-2mu}\gamma +2 \pi \alpha  \bigg]^{-1} \mspace{-3mu} \Ket{u} \mspace{-2mu}\Bra{v} + o\left( \left|\ln(\eps)\right| \right) \nonumber
\end{align}
with $\alpha$ defined by \eqref{Defalpha}. Finally, it follows  from \eqref{phi1212invasymp} and from the asymptotics of $g_{\eps}$ with $a= \Braket{u|v}/(2 \pi)$ that \eqref{Puniformest2} with $D(z,\underline{P})$ from \eqref{DefD(z,P)} holds true for all sufficiently small $\eps>0$. 
This proves the lemma. 
\end{proof}

For later convenience, we now collect the conclusions of this section. To this end, we need to reinstall the index $\sigma$, which we dropped at the beginning of the section. 
Let $\sigma \in \II$ be a fixed pair. If  $V_{\sigma} \in L^1 \cap L^2(\R^2)$ and $\int |r|^{2s}\left|V_{\sigma}(r)\right| \,\d{r}< \infty$ for some $s \in (0,2)$, then the analogues of Lemma \ref{lm4.2} and Proposition \ref{prop4.3} show that $\Lambda_{\eps}(z)_{\sigma \sigma}=g_{\eps, \sigma}^{-1}- \phi_{\eps}(z)_{\sigma \sigma}$ is invertible for all sufficiently small $\eps>0$ and sufficiently large $z>0$ and
\begin{align}
\limeps (\Lambda_{\eps}(z)_{\sigma \sigma})^{-1} =\begin{cases}
\;0  \hspace*{6cm} &\textup{if} \; \int V_{\sigma}(r) \, \d{r} /(2\pi) <  a_{\sigma} \\  \dfrac{\Ket{u_{\sigma}}\Bra{v_{\sigma}}}{\Braket{u_{\sigma}| v_{\sigma}}^2} \otimes (\Theta(z)_{\sigma \sigma})^{-1} &\textup{if} \;\int V_{\sigma}(r) \, \d{r}/(2\pi)= a_{\sigma},
\end{cases}
\label{Lambdadiaginvlim}
\end{align}
where $\Theta(z)_{\sigma \sigma}$ is an invertible operator in $\XX_{\sigma}$. For $\sigma=(1,2)$, $\Theta(z)_{\sigma \sigma}=D(z)$ with $D(z)$ defined by \eqref{DefD(z,P)} and \eqref{Defalpha}. For general pairs $\sigma=(i,j)$, the operator $\Theta(z)_{\sigma \sigma}$ acts pointwise in $\underline{P}_{\sigma}=(P,p_1,...\widehat{p}_i...\widehat{p}_j ...,p_N)$ by multiplication with
\begin{equation}
\Theta(z, \underline{P}_{\sigma})_{\sigma \sigma} = \frac{\mu_{\sigma}}{4 \pi} \left[
 \ln\mspace{-2mu}\Bigg(z + \frac{P^2}{m_i+m_j} + \sum\limits_{\substack{n=1 \\ n \neq i,j}}^{N} \frac{p_n^2}{m_n}\Bigg) + \frac{\beta_{\sigma}}{\pi} \right], \qquad \qquad \sigma=(i,j), \label{DefTheta(z,P)diag}
\end{equation}
where 
\begin{align}
\beta_{\sigma}&=2\pi \left( \ln(\sqrt{\mu_{\sigma}}/2) + \gamma +2 \pi \alpha_{\sigma} \right), \label{Defbetasigma} \\
\alpha_{\sigma} &=-\frac{b_{\sigma}}{\Braket{u_{\sigma}|v_{\sigma}}} + \frac{\Braket{v_{\sigma}|Lu_{\sigma}}}{2 \pi  \Braket{u_{\sigma}|v_{\sigma}}^2}. \label{Defalphasigma}
\end{align}

\subsection{Analysis of $\Lambda_{\eps}(z)_{\textup{off}}$} \label{sec4.2}
If $N > 2$, then $\Lambda_{\eps}(z)$ has an off-diagonal part  $\Lambda_{\eps}(z)_{\textup{off}}$ defined by Eq. \eqref{DefLambdaoff}. We will see in this section that  $\Lambda_{\eps}(z)_{\textup{off}}$ is uniformly bounded in $\eps >0$ and $z>0$ and that, upon introducing a cutoff, it has a limit as $\eps \downarrow 0$. These results will allow us in Section \ref{sec4.3} to prove existence of  $\limepsr \Lambda_{\eps}(z)^{-1}$ for sufficiently large $z>0$. With this goal in mind, the results of this section are formulated for $z>0$ only, although most of them still hold in a slightly modified way for general $z \in \rho(H_0)$.
\subsubsection{Uniform boundedness of $\Lambda_{\eps}(z)_{\textup{off}}$ in $\eps>0$ and $z>0$} \label{sec4.2.1}
Recall from Eqs. \eqref{DefLambdaoff} and \eqref{Defphisigmany} that the non-vanishing components of $\Lambda_{\eps}(z)_{\textup{off}}$ are given by
\begin{align}
\Lambda_{\eps}(z)_{\sigma \nu} = - \eps^{-2} \, (u_{\sigma} \otimes 1) \, U_{\eps} \KK_{\sigma}R_0(z) \KK_{\nu}^{*} U_{\eps}^{*} (v_{\nu} \otimes 1), \qquad \qquad  \sigma, \nu \in \II, \sigma \neq \nu. \label{Defphisigmany2}
\end{align}
We will prove in this section, among other things, that for $\sigma \neq \nu$ the norm $\|\Lambda_{\eps}(z)_{\sigma \nu} \|$ is uniformly bounded in  $\eps >0$ and $z>0$, the main result being Proposition \ref{prop4.4} below. The presence of the distinct coordinate transformations $\KK_{\sigma}$ and $\KK_{\nu}$ in \eqref{Defphisigmany2} makes the proof very technical and somewhat tedious. Since the tools used in this proof are not needed anymore in the sequel, it is possible and advisable to take Proposition \ref{prop4.4} for granted and to skip the proof at first reading. \newline  
With this said, we now start developing the tools for proving Proposition \ref{prop4.4}.
To estimate the norm of $\Lambda_{\eps}(z)_{\sigma \nu} $, we compute its integral kernel in terms of the Green's function $G_{z,\underline{m}}$ of $H_0$, where the vector $\underline{m}:=(m_1,...,m_N)$ collects the masses of the $N$ particles. By a simple scaling argument, 
\begin{align}
G_{z,\underline{m}}(x_1,...,x_N)= \left(\prod_{i=1}^{N} m_i\right)
G^{2N}_{z}(\sqrt{m_1}x_1,...,\sqrt{m_N}x_N),\label{DefG_(z,m)}
\end{align}
where $G^{2N}_{z}$ denotes the usual Green's function of $-\Delta+z: H^2(\R^{2N}) \rightarrow L^2(\R^{2N})$.
As the choice of the pair $\sigma$ is immaterial for our estimates, we can assume that
$\sigma=(1,2)$ and $\nu=(k,l) \neq (1,2)$.
Now, with $U_{\eps}$, $\KK_{(1,2)}$ and $\KK_{(k,l)}^{*}$ defined by \eqref{DefUeps}, \eqref{DefK} and \eqref{DefK*}, respectively, we find that \eqref{Defphisigmany2} defines an integral operator and, for $\psi \in \widetilde{\XX}_{(k,l)}$,  
\begin{align}
& \left(\Lambda_{\eps}(z)_{(1,2)(k,l)} \,\psi\right)(r,R,x_3,...,x_N) 
 \nonumber \\ &\;= -\eps^{-2}  u_{(1,2)}(r) \mspace{-2mu} \int \mspace{-4mu} \textup{d}x_1' \cdots \textup{d}x_N' \, G_{z,\underline{m}} \mspace{-3mu} \left( \mspace{-3mu} R-\Scalefrac{\eps m_2 r}{m_1+m_2}-x_1', R+\Scalefrac{\eps m_1 r}{m_1+m_2}-x_2',x_3-x_3',...,x_{N}-x_{N}' \mspace{-3mu} \right) \nonumber \\
 &\; \;\;\; \cdot v_{(k,l)}\left(\Scalefrac{x_l'-x_k'}{\eps}\right) \, \psi\mspace{-2mu}\left(\Scalefrac{x_l'-x_k'}{\eps},\Scalefrac{m_kx_k'+m_lx_l'}{m_k+m_l},x_1',...\widehat{x_k'}...\widehat{x_l'}...,x_N' \right) \nonumber \\ &\;=  -u_{(1,2)}(r) \mspace{-2mu} \int \mspace{-3mu} \textup{d}x_1' \cdots \textup{d}x_N' \,\d{r'} \, \d{R'}  \,   G_{z,\underline{m}} \mspace{-2mu} \left( R-\eps c_{21} r-x_1', R+\eps c_{12} r-x_2',x_3-x_3',...,x_{N}-x_{N}' \right) \nonumber \\
 & \;\;\;\; \cdot v_{(k,l)}\left(r'\right) \, \psi\mspace{-2mu}\left(r', R', x_1',...\widehat{x_k'}...\widehat{x_l'}...,x_N' \right)  \delta\mspace{-2mu} \left(x_k'-R'+\eps c_{lk} r'\right) \delta\mspace{-2mu}\left(x_l'-R'-\eps c_{kl} r'\right)\mspace{-2mu}, \label{intkernelLambda}
\end{align}
where 
\begin{align}
c_{ij}:= \frac{m_i}{m_i+m_j}, \qquad \qquad i,j =1,...,N. \nonumber
\end{align}
The second equation of \eqref{intkernelLambda} was obtained by the substitution
\begin{align}
r':=\frac{x_l'-x_k'}{\eps}, \qquad R':=\frac{m_kx_k'+m_lx_l'}{m_k+m_l}, \nonumber
\end{align}
and two more integrations  were introduced that are compensated by delta distributions. In the following, we distinguish two cases: In the first case 
$\sigma=(1,2)$ and $\nu = (k,l)$ have one particle in common, so $k \in \{1,2\}$ and $l \geq 3$. In the second case $\sigma$ and $\nu$ are composed of distinct particles, which means that $3\leq k<l \leq N$.

If $\sigma=(1,2)$ and $\nu=(1,l)$ with $l \geq 3$, then, after the evaluation of the delta distributions in $x_1'$ and $x_l'$, \eqref{intkernelLambda} shows that the operator $\Lambda_{\eps}(z)_{\sigma \nu} $ simply acts by convolution in $(x_3,...\widehat{x_l}...,x_N)$.
Consequently, the explicit formula \eqref{DefG_(z,m)} for $G_{z,\underline{m}}$ and Lemma \ref{lmA1} $(vi)$ imply that  $\Lambda_{\eps}(z)_{\sigma \nu} $ acts pointwise in the conjugate variable $\underline{p}_{(1,l)}:=(p_3,...\widehat{p_l}...,p_N)$ by the operator $\Lambda_{\eps}(z, \underline{p}_{(1,l)})_{\sigma \nu}$ that has the kernel
\begin{align}
-u_{\sigma}(r)\left( m_1m_2 m_l \, G_{z+Q_{\nu}}^{6}\mspace{-3mu} \left(X_{\sigma \nu,\eps}\right) \right)\mspace{-2mu}v_{\nu}\mspace{-2mu}\left(r'\right), \qquad \qquad \sigma=(1,2), \nu=(1,l), l \geq 3,  \label{kernelLambda121j}
\end{align}
where
\begin{align}
&X_{(1,2) (1,l),\eps}:=\left(\begin{array}{c}\sqrt{m_1}(R- R'-\eps( c_{21}r- c_{l1}r')) \\\sqrt{m_2}(R-x_2'+\eps c_{12}r) \\ \sqrt{m_l} ( x_l-R'-\eps c_{1l} r') \end{array}\right) \in \R^6 \nonumber
\end{align}
and
\begin{align}
Q_{(k,l)}:=\displaystyle\sum\limits_{\substack{n=3 \\ n \neq k,l}}^{N} \frac{p_n^2}{m_n}, \qquad \qquad \nu=(k,l) \in \II. \nonumber 
\end{align}

Similar considerations for $\sigma =(1,2)$ and $\nu =(2,l)$ with $l \geq 3$ show that 
$\Lambda_{\eps}(z)_{\sigma \nu}$ acts pointwise in $\underline{p}_{(2,l)}:=(p_3,...\widehat{p_l}...,p_N)$ by the operator $\Lambda_{\eps}(z, \underline{p}_{(2,l)})_{\sigma \nu}$ with kernel
\begin{align}
-u_{\sigma}(r)\left( m_1m_2 m_l \, G_{z+Q_{\nu}}^{6}\mspace{-3mu} \left(X_{\sigma \nu,\eps}\right) \right)\mspace{-2mu}v_{\nu}\mspace{-2mu}\left(r'\right),  \qquad \qquad \sigma=(1,2), \nu=(2,l), l \geq 3,\label{kernelLambda122j}
\end{align}
where
\begin{align}
X_{(1,2)(2,l),\eps}:=\left(\begin{array}{c} \sqrt{m_1}(R- x_1'-\eps c_{21} r) \\ \sqrt{m_2}(R-R'+\eps(c_{12}r+c_{l2}r')) \\ \sqrt{m_l} ( x_l-R'-\eps c_{2l} r') \end{array}\right) \in \R^6.\nonumber
\end{align}
By inspection, the kernels \eqref{kernelLambda121j} and \eqref{kernelLambda122j} only differ by the permutations $x_1' \leftrightarrow x_2'$, $m_1 \leftrightarrow m_2$, $v_{(1,l)} \leftrightarrow v_{(2,l)}$ and the reflection $r\rightarrow -r$, which will allow us to analyze them simultaneously.

So far, we have considered all operators that appear in the case of $N\leq 3$ particles.
Let now $N>3$, $\sigma=(1,2)$ and $\nu=(k,l)$ with $3 \leq k < l \leq N$. Then, after the evaluation of the delta distributions in $x_k'$ and $x_l'$, it follows from \eqref{intkernelLambda} that $\Lambda_{\eps}(z)_{\sigma \nu}$ acts by convolution in $(x_3,...\widehat{x_k}...\widehat{x_l}...,x_N)$. Hence, the action of $\Lambda_{\eps}(z)_{\sigma \nu}$ is pointwise in the conjugate variables $\underline{p}_{(k,l)}:=(p_3,...\widehat{p_k}...\widehat{p_l}...,p_N)$ by the operator $\Lambda_{\eps}(z, \underline{p}_{(k,l)})_{\sigma \nu}$ with kernel
\begin{align}
-u_{\sigma}(r) \left(m_1 m_2 m_k m_l \, G_{z+Q_{\nu}}^{8} \mspace{-2mu}\left(X_{\sigma \nu,\eps}\right)\right)v_{\nu}\mspace{-2mu}\left(r'\right),  \qquad \quad \sigma=(1,2), \nu=(k,l), 3 \leq k < l \leq N,\label{kernelLambda12kl}
\end{align}
where
\begin{align}
X_{(1,2)(k,l),\eps}&:=\left(\begin{array}{c} \sqrt{m_1}(R -x_1'-\eps c_{21}r) \\ \sqrt{m_2}(R -x_2'+\eps c_{12}r) \\
\sqrt{m_k}(x_{k}-R'+\eps c_{lk} r') \\ \sqrt{m_l}(x_{l}-R'-\eps c_{kl} r') \end{array}\right) \in \R^8. \nonumber
\end{align}

Besides the bound for the norm of $\Lambda_{\eps}(z)_{\sigma \nu}$, we shall also estimate the error caused by cutting off the potential outside some ball of radius $h>0$ in the following proposition.
This, in turn, will reduce the proof of our convergence result (see Proposition \ref{prop4.8}) to the case of compactly supported potentials. For any pair $\sigma \in \II$ and $h>0$, let
\begin{equation}
V_{\sigma}^h(x):= \begin{cases}
V_{\sigma}(x) \qquad \quad &\textup{if} \;\, |x| \leq h \\
0    &\textup{else}
\end{cases} \nonumber
\end{equation}
and let $\Lambda_{\eps}^h(z)_{\sigma \nu}$ denote the operator $\Lambda_{\eps}(z)_{\sigma \nu}$, where the potentials $V_{\sigma}$ and $V_{\nu}$ are replaced by $V_{\sigma}^h$ and $V_{\nu}^h$, respectively. This means that the kernel of  $\Lambda_{\eps}^h(z)_{\sigma \nu}$ emerges from the kernel of $\Lambda_{\eps}(z)_{\sigma \nu}$ by replacing $u_{\sigma}$ and $v_{\nu}$  with $u_{\sigma}^h:=\sgn(V_{\sigma}) |V_{\sigma}^h|^{1/2}$ and $v_{\nu}^h:=|V_{\nu}^h|^{1/2}$, respectively.
\begin{prop}\label{prop4.4}
For all pairs $\sigma=(i,j), \nu=(k,l)$ with $\sigma \neq \nu$ the operator $\Lambda_{\eps}(z)_{\sigma \nu}$
is uniformly bounded in $z>0$ and $\eps>0$. Explicitly, it holds that
\begin{equation}
\|\Lambda_{\eps}(z)_{\sigma \nu} \| \leq C(\sigma, \nu) \,\| V_{\sigma} \|_{L^1}^{1/2}\| V_{\nu} \|_{L^1}^{1/2}, \label{Lambdanormbound}
\end{equation} 
where $ C(\sigma, \nu) :=  (4\sqrt{2})^{-1} m_i m_j m_k m_l/\min(m_i,m_j, m_k,m_l)^3$. Furthermore, for any $h>0$,
\begin{align}
\|\Lambda_{\eps}(z)_{\sigma \nu} -\Lambda_{\eps}^h(z)_{\sigma \nu}\| &\leq  C(\sigma, \nu) \left( \| V_{\sigma} \|_{L^1} \| V_{\nu} \|_{L^1}  -  \| V_{\sigma}^h \|_{L^1} \| V_{\nu}^h \|_{L^1} \right)^{1/2}. \label{Lambdacutest}    
\end{align}
\end{prop}

\begin{proof}
Without loss of generality, we may assume that $\sigma=(1,2)$ and then we have to establish \eqref{Lambdanormbound} and \eqref{Lambdacutest} for all pairs $\nu \neq (1,2)$. We first consider the case $\nu=(1,l)$ with $l \geq 3$.

The proofs of \eqref{Lambdanormbound} and \eqref{Lambdacutest} are similar: It follows from \eqref{kernelLambda121j} that in both cases we have to estimate the norm of an operator that, for fixed $\underline{p}_{(1,l)}$, is given by a kernel of the form
\begin{align}
m_1 m_2 m_l \, W(r,r') \,G_{z+Q}^{6}\left( X_{\eps}\right),\nonumber
\end{align}
where  $X_{\eps}:=X_{(1,2)(1,l),\eps}$ and $Q:=Q_{(1,l)}$ for short. Explicitly, we have $W(r,r')=u_{(1,2)}(r)v_{(1,l)}(r')$ in the case of \eqref{Lambdanormbound} and $W(r,r')=u_{(1,2)}(r)v_{(1,l)}(r')-u_{(1,2)}^h(r)v_{(1,l)}^h(r')$ in the case of \eqref{Lambdacutest}. Therefore, we  only demonstrate the desired estimate in the case of \eqref{Lambdanormbound}.

For $\psi \in L^2(\R^{2} \times \R^{2} \times \R^{2})$, the Cauchy-Schwarz inequality in the $r'$-integration yields 
\begin{align}
&(m_1m_2m_l)^{-2} \|\Lambda_{\eps}(z,\underline{p}_{(1,l)})_{(1,2)(1,l)} \psi \|^2  \nonumber \\
&\quad = \int \mspace{-3mu} \d{r}\,\d{R}\,
\d{x_{l}} \left| \int \mspace{-3mu} \d{r'} \,  \d{R'} \,
 \d{x_{2}'} \; W(r,r')\, G^{6}_{z+Q}(X_{\eps})\, \psi(r',R',x_{2}') \right|^2 \nonumber \\
 &\quad\,= \int \mspace{-3mu} \d{r}\,\d{R}\, \d{x_{l}} \left| \int \mspace{-3mu} \d{r'} \,  W(r,r') \int \mspace{-3mu} \d{R'} \,  \d{x_{2}'} \; G^{6}_{z+Q}(X_{\eps})\, \psi(r',R',x_{2}') \right|^2 \nonumber \\
 & \quad \leq \int \mspace{-3mu} \d{r}\,\d{R}\, \d{x_{l}}  \left\{  \int \mspace{-3mu} \d{r'} \,  W(r,r') ^2 \right\}  \int \mspace{-3mu} \d{r'} \left| \int \mspace{-3mu} \d{R'} \, \d{x_{2}'} \; G^{6}_{z+Q}(X_{\eps})\, \psi(r',R',x_{2}') \right|^2 \nonumber \\
 & \quad \leq     \left\{  \int \mspace{-3mu} \d{r} \, \d{r'} \,  W(r,r')^2  \right\} \cdot \sup_{r \in \R^2} \left(  \int \mspace{-3mu} \d{r'} \d{R}\, \d{x_{l}} \left| \int \mspace{-3mu} \d{R'} \,
 \d{x_{2}'} \; G^{6}_{z+Q}(X_{\eps})\, \psi(r',R',x_{2}') \right|^2 \right)\mspace{-3mu}, \label{CSuse} 
\end{align}
where 
\begin{align}
\int \mspace{-3mu} \d{r} \, \d{r'} \;  W(r,r')^2   = \| V_{(1,2)} \|_{L^1}\| V_{(1,l)} \|_{L^1}. \label{L2normW}
\end{align}
For a further estimate of \eqref{CSuse}, we fix $r \in \R^2$. Then triangle inequality and the sequence of substitutions $R'+\eps (c_{21}r-c_{l1}r') \rightarrow  R'$, $x_l+\eps c_{21}r- \eps r' \rightarrow x_l$, $x_2'- \eps  c_{12}r \rightarrow x_2'$ in the first step and the monotonicity of the Green's function w.r.t. $z$ and $m_1,m_2,m_l$ (cf. Lemma \ref{lmA1} $(v)$) in the second step yield that
\begin{align}
&\int \mspace{-3mu} \d{r'} \, \d{R}\,\d{x_{l}} \left| \int \mspace{-3mu} \d{R'} \, \d{x_{2}'} \, G^{6}_{z+Q}\mspace{-1mu}\left(X_{\eps} \right)\,  \psi(r',R',x_{2}') \right|^2 \nonumber \\
& \quad = \int \mspace{-3mu} \d{r'} \, \d{R}\,\d{x_{l}} \left( \int \mspace{-3mu} \d{R'} \, \d{x_{2}'} \, G_{z+Q}^6 \mspace{-2mu} \left(\sqrt{m_1}(R-R'),\sqrt{m_2}(R-x_2'),\sqrt{m_l}(x_l-R') \right)
\widetilde{\psi}(r'\mspace{-1mu},\mspace{-1mu}R'\mspace{-1mu},\mspace{-1mu}x_2') \right)^2 \nonumber \\
& \quad \leq \int \mspace{-3mu} \d{r'} \left\|F \widetilde{\psi}(r',\,\cdot\,)\right\|^2, \label{CSuse2}
\end{align}
where $\widetilde{\psi} \in L^2(\R^2 \times \R^2 \times \R^2)$ is given by
\begin{align}
\widetilde{\psi}(r',R',x_2'):=\left|\psi\left(r',R'- \eps (c_{21}r-c_{l1}r'), x_2'+ \eps c_{12}r  \right)\right| \nonumber
\end{align}
and $F: L^2\left(\R^2 \times \R^2 , \textup{d}\left(R,x_2\right) \right) \rightarrow L^2\left(\R^2 \times \R^2, \textup{d}\left(R,x_l\right) \right)$ is defined by the kernel
\begin{align}
G_{z}^6\left(\sqrt{m}(R-R'),\sqrt{m}(R-x_2'),\sqrt{m}(x_l-R')\right), \qquad \qquad m=\min(m_1,m_2, m_l).\nonumber
\end{align}
By Lemma \ref{lmA3}, $F$ is bounded and $\|F\| \leq (4\sqrt{2}m^2)^{-1}$.
Using this together with the fact that $\|\psi\|=\|\widetilde{\psi}\|$, we obtain from \eqref{CSuse2} that
\begin{align}
&\int \mspace{-3mu} \d{r'} \, \d{R}\,\d{x_{l}} \left| \int \mspace{-3mu} \d{R'} \, \d{x_{2}'} \, G^{6}_{z+Q}\mspace{-1mu}\left(X_{\eps} \right)\,  \psi(r',R',x_{2}') \right|^2 
\leq (32 \min(m_1,m_2, m_l)^4)^{-1} \|\psi \|^2,  \label{kernelL2}
\end{align}
where the right side is independent of $r \in \R^2$ and $\underline{p}_{(1,l)}$.
Hence, \eqref{Lambdanormbound} for $\sigma=(1,2)$ and $\nu=(1,l)$ now follows by combining \eqref{CSuse}, \eqref{L2normW} and \eqref{kernelL2}.

As already mentioned, the bound \eqref{Lambdacutest} is obtained in the same manner with the only difference that $W(r,r')=u_{\sigma}(r)v_{\nu}(r')-u_{\sigma}^h(r)v_{\nu}^h(r')$ in this case. Then, \eqref{L2normW} has to be replaced by the identity
\begin{align}
\int \mspace{-3mu} \d{r} \, \d{r'} \;  W(r,r')^2   = \| V_{\sigma} \|_{L^1} \| V_{\nu} \|_{L^1}  -  \| V_{\sigma}^h \|_{L^1} \| V_{\nu}^h \|_{L^1}, \nonumber
\end{align}
which is a consequence of the identities $u_{\sigma} \cdot u_{\sigma}^h=|V_{\sigma}^h|$ and $v_{\nu} \cdot v_{\nu}^h=|V_{\nu}^h|$. This completes the proof of \eqref{Lambdanormbound} and \eqref{Lambdacutest} if $\sigma=(1,2)$ and $\nu=(1,l)$ with $l \geq 3$. 

If $\sigma=(1,2)$ and $\nu=(2,l)$ with $l \geq 3$, then the kernel of $\Lambda_{\eps}(z)_{\sigma \nu}$ only differs from the kernel of $\Lambda_{\eps}(z)_{(1,2)(1,l)}$  by the permutations $x_1' \leftrightarrow x_2'$, $m_1 \leftrightarrow m_2$, $v_{(1,l)} \leftrightarrow v_{(2,l)}$ and the (unitary) reflection $r\rightarrow -r$. Therefore, the above proof also works in this case. 

If $\sigma=(1,2)$ and $\nu=(k,l)$ with $3 \leq k < l \leq N$, then the operator $\Lambda_{\eps}(z)_{\sigma \nu}$ acts pointwise in $\underline{p}_{(k,l)}=(p_3,...\widehat{p_k}...\widehat{p_l}...,p_N)$ by the operator $\Lambda_{\eps}(z, \underline{p}_{(k,l)})_{\sigma \nu}$ that is defined by the kernel \eqref{kernelLambda12kl}. In this case, the above estimates have to be slightly adjusted. The role of the operator $F$ is now played by the operator $B \in \LL\left(L^2(\R^2 \times \R^2\times \R^2 )\right)$ from Lemma \ref{lmA3}. As the bounds for $\|F\|$ and for $\|B\|$ in Lemma \ref{lmA3} differ by a factor of $m^{-1}$, while there is one mass factor more in front of the Green's function for $\nu=(k,l)$ with $3 \leq k < l \leq N$ than for $\nu=(1,l)$ with $l \geq 3$, we again obtain \eqref{Lambdanormbound} and \eqref{Lambdacutest} with $C(\sigma, \nu)$ given in the statement of the proposition.
\end{proof}

\subsubsection{The off-diagonal limit operators} \label{sec4.2.2}
This section is a preparation for the next one, where we shall be concerned with the convergence, as $\eps \to 0$, of
\begin{align}
\Lambda_{\eps}(z)_{\sigma \nu}=-B_{\eps, \sigma}R_0(z)A_{\eps, \nu}^{*}, \qquad \qquad \sigma \neq \nu. \label{DefLambdaepssigmany2}
\end{align}

From Proposition \ref{prop3.3} we know that
\begin{align}
 R_0(z) A_{\eps,\nu}^{*}  \rightarrow S(z)_{\nu}^{*} =  G(z)_{\nu}^{*} \Bra{v_{\nu}}, \qquad \qquad (\eps \downarrow 0)  \label{nulimit}
\end{align}
in $\LL(\widetilde{\XX}_{\nu}, \HH)$, and
\begin{align}
B_{\eps, \sigma} \rightarrow \Ket{u_{\sigma}} T_{\sigma}, \qquad \qquad (\eps \downarrow 0)  \label{Bsigmalimit}
\end{align}
in $\LL(H^2(\R^{2N}), \widetilde{\XX}_ {\sigma})$. 
Here, $ \Ket{u_{\sigma}}: \XX_{\sigma} \rightarrow \widetilde{\XX}_{\sigma}$ is defined by $ \Ket{u_{\sigma}} \psi = u_{\sigma} \otimes \psi$ and $ \Bra{v_{\nu}}: \widetilde{\XX}_{\nu} \rightarrow \XX_{\nu}$ is the adjoint of $ \Ket{v_{\nu}}$. The negative of the formal composition of the limits in \eqref{nulimit} and \eqref{Bsigmalimit} is the operator  
\begin{align}
\Ket{u_{\sigma}}\Bra{v_{\nu}} \otimes  \Theta(z)_{\sigma \nu}\nonumber
\end{align}
with
\begin{align}
\Theta(z)_{\sigma \nu}:=-T_{\sigma} G(z)_{\nu}^{*}, \qquad \qquad \sigma \neq \nu. \label{DefThetaoff}
\end{align}

In the remainder of this section, we show that \eqref{DefThetaoff} defines an element $\Theta(z)_{\sigma \nu} \in  \LL(\XX_{\nu}, \XX_{\sigma})$. We begin by computing representations of $T_{\sigma}$ and $G(z)_{\nu}^{*}$ in Fourier space. Using the definitions \eqref{DefTsigma}, \eqref{Deftau} and \eqref{DefK} of $T_{\sigma}$, $\tau$ and $\KK_{\sigma}$, respectively, we find for $\sigma=(i,j)$ and $\psi \in D(T_{\sigma})$ that 
\begin{align}
&(\widehat{T_{(i,j)}\psi})\left(P,p_1,...\widehat{p_i}...\widehat{p_j}...,p_N\right)=
\frac{1}{2\pi}\mspace{-3mu} \int \mspace{-3mu} \d{p} \; \widehat{\KK_{(i,j)}\psi}\left(p,P,p_1,...\widehat{p_i}...\widehat{p_j}...,p_N\right) \nonumber \\
& \quad= \frac{1}{2\pi}\mspace{-3mu}\int\mspace{-3mu} \d{p} \; \widehat{\psi}\mspace{-2mu}\left(p_1,...,p_{i-1}, \frac{m_i P}{m_i+m_j}-p,p_{i+1},...,p_{j-1}, \frac{m_jP}{m_i+m_j}+p , p_{j+1}, ..., p_{N} \right). \label{TsigmaFourier}
\end{align}
To compute the Fourier transform of $G(z)_{\nu}^{*} w$ for $z>0$, $\nu=(k,l)$, and $w \in \XX_{\nu}$, we use
$G(z)_{\nu}=T_{\nu}R_0(z)$ in combination with \eqref{TsigmaFourier} and the substitution 
\begin{align}
P:=p_k+p_l, \quad p:= \frac{m_k p_l - m_l p_k}{m_k+m_l}. \nonumber
\end{align}
After a straightforward computation, we find that $\Braket{ w| G(z)_{\nu} \psi}=\Braket{ G(z)_{\nu}^{*}w| \psi}$ for any $\psi \in \HH$, where
\begin{align}
\widehat{ (G(z)_{\nu}^{*} w) }\left(p_1,...,p_N\right) = \frac{1}{2\pi} \left(z+ \sum_{n=1}^{N}\frac{p_n^2}{m_n} \right)^{-1} \mspace{-12mu}\cdot \widehat{w}\mspace{-3mu}\left(p_k+p_l,p_{1}, ...  \widehat{p}_k... \widehat{p}_l ..., p_N\mspace{-3mu}\right). \label{G_(k,l)(z)*Fourier}
\end{align}
With these identities at hand, we now come to the main result of this section:

\begin{prop}\label{prop4.5}
Let $\sigma = (i,j) \neq (k,l)= \nu$ and $z>0$ be given. Then  $\Theta(z)_{\sigma \nu}= -T_{\sigma}G(z)_{\nu}^{*}$ defines a bounded operator and $\|\Theta(z)_{\sigma \nu} \| \leq \max(m_i,m_j,m_k, m_l)/4$.
\end{prop}

In the proof of Proposition \ref{prop4.5} we will need the following result taken from \cite[Lemma 3.1]{Figari}. For the convenience of the reader, we give a short proof with a different constant here.

\begin{lemma}\label{lm4.6}
For all $f,g \in L^2(\R^2)$, 
\begin{align}
\int \frac{|f(x)| |g(x')| }{|x|^2+|x'|^2} \, \d{x} \, \d{x'} \leq \pi^2 \|f\| \|g\|. \label{intbound} 
\end{align}
\end{lemma}
\begin{proof} Let $K$ denote the integral operator in $L^2(\R^2)$ that is defined by the kernel
$K(x;x')=(|x|^2+|x'|^2)^{-1}$. Using the Schur test with $h(x)=|x|^{-1}$, it is straightforward to verify that $K$ defines a bounded operator and that
\begin{align}
\| K \| &\leq \esssup_{x' \in \R^2}\left(|x'| \int\limits_{\R^2} \frac{1}{|x|^2+|x'|^2} \frac{1}{|x|} \, \d{x}\right) = \pi^2. \nonumber
\end{align}
Hence, $\Braket{|f|, K|g|} \leq \pi^2 \|f\| \|g\|$, which establishes \eqref{intbound}.
\end{proof}

\begin{proof}[Proof of Proposition \ref{prop4.5}] Without loss of generality, we may assume that $\sigma=(1,2)$. Moreover, it suffices to consider the two cases $\nu=(1,3)$ and $\nu=(3,4)$ corresponding to pairs with one common particle and  no common particle, respectively. 

To show that $ \ran\,G(z)_{\nu}^{*} \subset D(T_{\sigma})$, or equivalently, that  $ \ran\, \KK_{\sigma}G(z)_{\nu}^{*} \subset D(\tau)$, we have to verify that, for all  $w \in \XX_{\nu}$,
\begin{align}
\widehat{\phi}(P,p_3,...,p_N):= \frac{1}{2\pi}\mspace{-3mu}\int\mspace{-3mu} \d{p} \,  \left|  \widehat{G(z)_{\nu}^{*} w}\mspace{-2mu}\left(\frac{m_1 P}{m_1+m_2}-p, \frac{m_2P}{m_1+m_2}+p , p_{3}, ..., p_{N} \right) \right|\nonumber
\end{align}
defines an $L^2$-function. To this end, by the Riesz-Lemma, it suffices to show that 
\begin{align}
\int  \mspace{-3mu} \d{P} \, \d{p_3} \cdots \d{p_N} \, |(\widehat{\psi} \widehat{\phi})(P,p_3,...,p_N)| \leq \textup{const}. \|\psi\|\nonumber
\end{align}
for all $\psi \in \XX_{\sigma}$. The substitution 
\begin{align}
P:=p_1+p_2, \quad p:= \frac{m_1 p_2 - m_2 p_1}{m_1+m_2} \nonumber
\end{align} 
and the identity \eqref{G_(k,l)(z)*Fourier} for $G(z)_{\nu}^{*}w$ show that
\begin{align}
&\int \mspace{-3mu} \d{P} \, \d{p_3} \cdots \d{p_N} \, |\widehat{\psi}(P,p_3,...,p_N)| |\widehat{\phi}(P,p_3,...,p_N)| \nonumber\\
&\,=\frac{1}{2\pi}\int \mspace{-3mu}\d{p_1} \cdots \d{p_N} \,
|\widehat{\psi}(p_1+p_2, p_3,..., p_N)| \left|  \widehat{G(z)_{\nu}^{*} w}\mspace{-2mu}\left(p_1, ..., p_{N} \right) \right| \nonumber\\
&\,\leq \frac{\max(m_2,\mspace{-2mu}m_l)}{4\pi^2} \int \mspace{-5mu}\d{p_1} \cdots \d{p_N}
\frac{|\widehat{\psi}(p_1+p_2, p_3,..., p_N)| \mspace{-2mu} \left|\widehat{w}\mspace{-3mu}\left(p_k+p_l,p_{1}, ...  \widehat{p}_k... \widehat{p}_l ..., p_N\mspace{-3mu}\right) \right| }{p_2^2 + p_l^2}.  \label{vphibound}
\end{align}

If $\nu=(k,l)=(1,3)$, let $\underline{p}:=(p_4,...,p_N)$ and set
\begin{equation}
\begin{aligned}
f(p_3,\underline{p}):= \left( \int \mspace{-3mu} \d{p_1}\; |\widehat{\psi}(p_1,p_3,\underline{p})|^2 \right)^{1/2} \\
g(p_2,\underline{p}):= \left( \int \mspace{-3mu} \d{p_1} \; |\widehat{w}(p_1,p_2,\underline{p})|^2 \right)^{1/2}.\nonumber
\end{aligned} 
\end{equation}  
Then $f(\cdot, \underline{p}), g(\cdot, \underline{p}) \in L^2(\R^2)$ for almost all $\underline{p}$. Hence, after applying the Cauchy-Schwarz inequality in the $p_1$-integration, it follows from Lemma \ref{lm4.6} that  
\begin{align}
&\int \mspace{-3mu} \d{\underline{p}} \, \d{p_2}\, \d{p_3} \, \d{p_1} \;
\frac{|\widehat{\psi}(p_1+p_2, p_3,\underline{p})| \left|\widehat{w}\mspace{-3mu}\left(p_1+p_3,p_{2},\underline{p}\right) \right| }{p_2^2 + p_3^2} \nonumber \\
& \quad \leq \int \mspace{-3mu} \d{\underline{p}} \, \d{p_2}\, \d{p_3} \; \frac{|f(p_3,\underline{p})| \left|g(p_2,\underline{p}) \right| }{p_2^2 + p_3^2} \leq \pi^2 \| \psi\| \| w\|. \label{nu=(1,3)}
\end{align}

If $\nu=(k,l)=(3,4)$, we set $\underline{p}:=(p_5,...,p_N)$ and
\begin{equation}
\begin{aligned}
f(p_4,\underline{p}):= \left( \int \mspace{-3mu} \d{p_1}\, \d{p_3} \; |\widehat{\psi}(p_1,p_3,p_4,\underline{p})|^2 \right)^{1/2} \\
g(p_2,\underline{p}):= \left( \int \mspace{-3mu} \d{p_1}\, \d{p_3} \; |\widehat{w}(p_3,p_1,p_2,\underline{p})|^2 \right)^{1/2}.\nonumber
\end{aligned} 
\end{equation}  
Using the Cauchy-Schwarz inequality in the $(p_1,p_3)$-integration, we conclude from Lemma \ref{lm4.6} that  
\begin{align}
&\int \mspace{-3mu} \d{\underline{p}} \, \d{p_2}\, \d{p_4} \, \d{p_1} \, \d{p_3} \;
\frac{|\widehat{\psi}(p_1+p_2, p_3,p_4,\underline{p})| \left|\widehat{w}\mspace{-3mu}\left(p_3+p_4,p_{1},p_{2}, \underline{p}\right) \right| }{p_2^2 + p_4^2} \nonumber \\
& \quad \leq \int \mspace{-3mu}\d{\underline{p}} \, \d{p_2}\, \d{p_4} \; \frac{|f(p_4,\underline{p})| \left|g(p_2,\underline{p}) \right| }{p_2^2 + p_4^2} \leq \pi^2 \| \psi\| \| w\|.\label{nu=(3,4)}
\end{align}

From  \eqref{vphibound}, \eqref{nu=(1,3)} and \eqref{nu=(3,4)}, it follows that $\widehat{\phi} \in L^2$ and hence $\ran\,G(z)_{\nu}^{*} \subset D(T_{\sigma})$. Furthermore, the above estimates  now show that 
\begin{align}
\left| \Braket{\psi|T_{\sigma}G(z)_{\nu}^{*} w}  \right| \leq \frac{\max(m_1,m_2,m_3,m_4)}{4} \|\psi \| \|w\|.\nonumber
\end{align} 
Therefore,  $\Theta(z)_{\sigma \nu}= -T_{\sigma}G(z)_{\nu}^{*}$ is a bounded operator, whose norm satisfies the desired estimate.
\end{proof}

\subsubsection{Cutoff functions and convergence} \label{sec4.2.3}
The goal of this section is to prove Proposition \ref{prop4.8} on the convergence, as $\eps \to 0$, of  $\Lambda_{\eps}(z)_{\sigma \nu}$ with a suitable space cutoff $\chi_{\sigma \nu,c}$ defined below. In the analysis of $\limepsr \Lambda_{\eps}(z)^{-1}$ in Section \ref{sec4.3}, this cutoff will be removed again. 

We begin by motivating the cutoff: By \eqref{Bsigmalimit}, $\lim_{\eps \to 0} B_{\eps, \sigma}$ exists in $\LL(H^2(\R^{2N}), \widetilde{\XX}_{\sigma})$. Here, the Sobolev index $2$ could be reduced to $1+\delta$ for some $\delta > 0$ but not further because of the presence of the trace operator $T_{\sigma}$ in \eqref{Bsigmalimit}. Moreover, by \eqref{nulimit}, $\lim_{\eps \to 0} R_0(z)  A_{\eps, \nu}^{*}$ exists in $\LL(\widetilde{\XX}_{\nu}, \HH)$, or perhaps in $\LL(\widetilde{\XX}_{\nu},H^{1-\delta}(\R^{2N}))$, which is not enough to prove that $\Lambda_{\eps}(z)_{\sigma \nu}$ given by \eqref{DefLambdaepssigmany2} has a limit as $\eps \rightarrow 0$. The space cutoff to be introduced below removes the singularity due to $T_{\nu}^{*}$. 
More explicitly, for $\nu=(1,2)$, $v=v_{\nu}$ and $\widetilde{R}_0(z)=(\widetilde{H}_0 +z)^{-1}$ (see \eqref{H0H0Tilde} and \eqref{DefH0Tilde}), the operator $\widetilde{R}_0(z) \eps^{-1} U_{\eps}^{*} (v \otimes 1)$ has the kernel
\begin{align}
 G_{z,\underline{\widetilde{m}}}\left(r-\eps r', R-R',x_3-x_3',...,x_N-x_N'\right) v(r'), \nonumber
\end{align}
 where $\underline{\widetilde{m}}:=(\mu,m_1+m_2,m_3,...,m_N)$.
 Hence, if $v$ has compact support and $|r| \geq c >0$ is enforced by a space cutoff 
 $\chi_c$, then, for sufficiently small $\eps>0$, the singularity of $G_{z,\underline{\widetilde{m}}}$ at the origin is avoided and, by Lemma \ref{lmA2}, this means that the kernel becomes a smooth function of $(r,R,x_3,...,x_N)$.

Let $\chi \in C^{\infty}(\R,[0,1])$ be a function with 
$\chi(x)=0$ for $x\leq 1$ and $\chi(x)=1$ for $x \geq 2$ and set $\chi_c(r):=\chi(|r|/c)$ for $r \in \R^2$ and $c>0$.
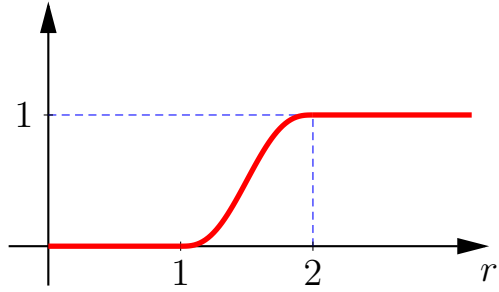
\begin{figure}[H]
 \begin{center}
 \begin{tikzpicture}[scale=0.58,>={Triangle[length=12,width=6]}]]
  \draw[->, thick] (-0.9cm,0) -- (10cm,0) node[align=center,below=0.1cm] {\Large $r$};
  \draw[->, thick] (0,-0.9cm) -- (0,5.6cm) node[yshift=-0.3cm,xshift=-0.4cm] { };
  \draw[densely dashed, blue] (6,0) -- (6,3);   \draw[densely dashed, blue] (0,3) -- (6,3);
  \foreach \x/\xtext in { 1,2}
  \draw (3 * \x cm,3pt) -- (3 * \x cm,-3pt) node[anchor=north] {\Large $\xtext$};
  \foreach \y/\ytext in {1}
  \draw (3pt,\y*3 cm) -- (-3pt,\y*3 cm) node[anchor=east] {\Large $\ytext$};
  \draw[red, domain=1:2, line width=2pt] 
  plot[smooth] (3*\x,{3*( 10*(\x-1)^3 -15*(\x-1)^4 + 6*(\x-1)^5)}) ;
 \draw[color=red, line width=2pt] (0,0) -- (3,0);
 \draw[color=red, line width=2pt] (6,3) -- (9.6,3);
 \end{tikzpicture} 
 \end{center}
\caption{The cutoff function $\chi$} 
\end{figure} 
\noindent Then we have the following result:
\begin{lemma}\label{lm4.7}
Assume that $v \in L^2(\R^2)$ has compact support $\supp(v) \subset \{x \in \R^2 \,| \, |x|\leq h\}$ for some $h>0$
and let $c,z>0$. Then, for all $n \in \N_0$, the limit 
\begin{align}
\lim\limits_{\eps \to 0} (\chi_c \otimes 1) \widetilde{R}_0(z) \eps^{-1} U_{\eps}^{*} (v \otimes 1) = (\chi_c \otimes 1) \KK S(z)^{*}\nonumber
\end{align}
exists in  $\LL(L^2(\R^{2N}),H^n(\R^{2N}))$, where $\KK=\KK_{(1,2)}$ and   $S(z)=S(z)_{(1,2)}$.
\end{lemma}
 
\begin{proof} Clearly, the lemma will follow if we show that
\begin{align}
\lim\limits_{\eps \to 0} \, (\partial^{\alpha}\chi_c)  \widetilde{R}_0(z) \eps^{-1} U_{\eps}^{*} (v \otimes 1) = (\partial^{\alpha} \chi_c) \KK S(z)^{*} \qquad \qquad \textup{in}\; \LL(L^2(\R^{2N}),H^{n}(\R^{2N})) \label{cutconv}
\end{align}
for all $n \in \N_0$ and all multi-indices $\alpha \in \N_0^{2}$, where $\partial^{\alpha} \chi_c=(\partial^{\alpha}\chi_c) \otimes 1$ for short. Using \eqref{H0H0Tilde} in combination with
the adjoint of \eqref{Aepsijnew} and \eqref{nulimit}, it follows that 
\begin{align}
\lim\limits_{\eps \to 0} \widetilde{R}_0(z) \eps^{-1} U_{\eps}^{*} (v \otimes 1) 
= \KK S(z)^{*}  \qquad \quad \textup{in} \; \LL(L^2(\R^{2N})),\nonumber
\end{align}
so \eqref{cutconv} holds for $n=0$. 

We now proceed by induction and assume that \eqref{cutconv} holds for some $n \in \N_0$ and all $\alpha \in \N_0^2$. To show that \eqref{cutconv} also holds for $n+1$, we are going to use that
\begin{align}
\eta \widetilde{R}_0(z) = \widetilde{R}_0(z) \left[ \widetilde{H}_0, \eta  \right] \widetilde{R}_0(z) + \widetilde{R}_0(z)\eta , \qquad \qquad \eta= \partial^{\alpha}\chi_c, \label{Comm1}
\end{align}
where $[\cdot, \cdot]$ denotes the commutator and
\begin{align}
\lim\limits_{\eps \to 0} \widetilde{R}_0(z) \eta \, \eps^{-1} U_{\eps}^{*} (v \otimes 1) 
= 0 \qquad \qquad \textup{in}\; \LL(L^2(\R^{2N}),H^{n+1}(\R^{2N})) \label{Aftercommlim}
\end{align}
because $ \eta \, U_{\eps}^{*} (v \otimes 1)=0$ for $\eps h < c$.
Let $D_k: H^1(\R^{2}) \rightarrow L^2(\R^{2})$ denote the partial derivative $D_k \varphi:= \partial_{k} \varphi$ for $k \in \{1,2\}$. From the definition $\eqref{DefH0Tilde}$ of $\widetilde{H}_0$ and from $\Delta_r = D_1 D_1 + D_2 D_2$, it follows that 
\begin{align}
\mu  \left[ \widetilde{H}_0,\eta \right] =  \left[ -\Delta_r , \eta\right] 
= -\sum_{k=1}^2 \left(D_k (\partial_k \eta) + (\partial_k  \eta) D_k \right)
=  (\Delta_r \eta) - 2 \sum_{k=1}^2 D_k (\partial_k \eta), \label{Comm2}
\end{align}
where $\mu=\mu_{(1,2)}$ is a reduced mass. Equations \eqref{Comm1} and \eqref{Comm2} imply that
\begin{align}
\eta \widetilde{R}_0(z) = \widetilde{R}_0(z) \eta + \mu^{-1} \widetilde{R}_0(z)(\Delta_r \eta) \widetilde{R}_0(z) - 2 \mu^{-1} \sum_{k=1}^2  D_k \widetilde{R}_0(z)(\partial_k \eta)  \widetilde{R}_0(z).  \label{Commid}
\end{align}
After multiplying \eqref{Commid} with $\eps^{-1} U_{\eps}^{*} (v \otimes 1)$ from the right,
we get from \eqref{Aftercommlim} combined with the induction hypothesis and the fact that 
$ \widetilde{R}_0(z) $ and $D_k \widetilde{R}_0(z)$ belong to  $\LL(H^n(\R^{2N}),H^{n+1}(\R^{2N}))$ that $\lim_{\eps \to 0} \eta \widetilde{R}_0(z) \eps^{-1} U_{\eps}^{*} (v \otimes 1)$ exists  in $\LL(L^2(\R^{2N}),H^{n+1}(\R^{2N}))$. Furthermore, the limit operator has to be the one in \eqref{cutconv} because convergence in $\LL(L^2(\R^{2N}),H^{n+1}(\R^{2N}))$ implies convergence in $\LL(L^2(\R^{2N}),H^{n}(\R^{2N}))$. Hence, \eqref{cutconv} holds with $n+1$ in place of $n$ and the proof is complete.
\end{proof}
 
With the help of Lemma \ref{lm4.7} we can now prove convergence of the regularized operators $(1 \otimes \chi_{\sigma \nu,c}) \Lambda_{\eps}(z)_{\sigma \nu}$ for pairs  $\sigma=(i,j) \neq (k,l)= \nu$, and a cutoff function $\chi_{\sigma \nu,c}$ defined by
\begin{align}
\chi_{(i,j)(k,l),c}\left(R,x_1,...\widehat{x}_i...\widehat{x}_j...,x_N\right):= \begin{cases} 
\chi_c\left(x_l -R \right) \qquad \quad &\textup{if}\; k \in \{i,j\}, l \notin \{i,j\}, \\
\chi_c\left(x_k-R  \right) \qquad \quad &\textup{if}\; k \notin \{i,j\}, l \in \{i,j\}, \\
\chi_c\left(x_l-x_k  \right) \qquad \quad &\textup{if}\; k,l \notin  \{i,j\}.
\end{cases} \label{Defchisigmanu}
\end{align}

\begin{prop}\label{prop4.8}
For all pairs $\sigma, \nu \in \II$, $\sigma \neq \nu$ and $c,z>0$ the limit
\begin{align}
\lim\limits_{\eps \to 0} (1 \otimes \chi_{\sigma \nu,c})\, \Lambda_{\eps}(z)_{\sigma \nu} = 
(1 \otimes \chi_{\sigma \nu,c})\, \Lambda(z)_{\sigma \nu} \label{Lambdacutlim}
\end{align}
exists in $\LL(\widetilde{\XX}_{\nu}, \widetilde{\XX}_{\sigma})$, where
\begin{align}
\Lambda(z)_{\sigma \nu}:= \Ket{u_\sigma} \Bra{v_{\nu}} \otimes \Theta(z)_{\sigma \nu} \label{DefLambda(z)}
\end{align}
and $\Theta(z)_{\sigma \nu}=-T_{\sigma} G(z)_{\nu}^{*} \in \LL(\XX_{\nu}, \XX_{\sigma})$.
\end{prop}

\begin{proof} Assume, for the moment, that for all $h>0$, 
\begin{align}
\limeps(1 \otimes \chi_{\sigma \nu,c})\, \Lambda_{\eps}^{h}(z)_{\sigma \nu} = (1 \otimes \chi_{\sigma \nu,c})\, \Lambda^{h}(z)_{\sigma \nu}, \label{compactlim}
\end{align}  
where $\Lambda_{\eps}^{h}(z)_{\sigma \nu}$ was introduced in Section \ref{sec4.2.1} and $\Lambda^{h}(z)_{\sigma \nu}:= \ket{u_{\sigma}^h} \bra{v_{\nu}^h} \otimes\Theta(z)_{\sigma \nu}$.  
By Proposition \ref{prop4.5}, $\Theta(z)_{\sigma \nu} = -T_{\sigma} G(z)_{\nu}^{*}$ defines a bounded operator and $\|\Theta(z)_{\sigma \nu} \| \leq \widetilde{C}(\sigma, \nu)$, so the definition \eqref{DefLambda(z)} and the corresponding identity for $\Lambda^{h}(z)_{\sigma \nu}$ imply that
\begin{align}
\|\Lambda(z)_{\sigma \nu} -\Lambda^{h}(z)_{\sigma \nu} \| & \leq \left\|\Ket{u_\sigma} \Bra{v_{\nu}} - \ket{u_{\sigma}^h} \bra{v_{\nu}^h} \right\| \| \Theta(z)_{\sigma \nu}\| \nonumber \\
&\leq \widetilde{C}(\sigma, \nu) \left( \| V_{\sigma} \|_{L^1} \| V_{\nu} \|_{L^1}  -  \| V_{\sigma}^h \|_{L^1} \| V_{\nu}^h \|_{L^1} \right)^{1/2},\label{Lambdacutest0}    
\end{align}
which is the equivalent of \eqref{Lambdacutest} for $\eps=0$. Now, the general case of \eqref{Lambdacutlim}, where $u_{\sigma}$ and $v_{\nu}$ do not have compact support, follows from \eqref{compactlim}, \eqref{Lambdacutest} and \eqref{Lambdacutest0} by a simple $\delta/3$-argument.

It remains to prove \eqref{compactlim}. As the choice of the pair $\nu$ is immaterial for the following arguments, it suffices to consider the case   $\nu=(1,2)$ only. Moreover, we may assume that $\supp(u_{\sigma}) \cup \supp(v_{\nu}) \subset  \{x \in \R^2 \,| \, |x|\leq h\}$, i.e. $u_{\sigma}^h=u_{\sigma}$ and $v_{\nu}^h=v_{\nu}$. Then \eqref{Defphisigmany2} becomes
\begin{align}
\Lambda_{\eps}^{h}(z)_{\sigma \nu}= -\eps^{-2} \, (u_{\sigma} \otimes 1) \, U_{\eps} \KK_{\sigma}R_0(z) \KK_{\nu}^{*} U_{\eps}^{*} (v_{\nu} \otimes 1). \label{DefLambdah}
\end{align}
From Lemma \ref{lm3.2}, combined with Eqs. \eqref{DefTsigma} and \eqref{Deftau}, it follows that
\begin{align}
\limeps \eps^{-1} \, (u_{\sigma} \otimes 1) \, U_{\eps} \KK_{\sigma} =  (\Ket{u_{\sigma}}\Bra{\delta} \otimes 1  ) \KK_{\sigma} = \Ket{u_{\sigma}}T_{\sigma} \qquad \quad \textup{in} \;\LL(H^2(\R^{2N}), \widetilde{\XX}_{\sigma}). \label{Convfirst}
\end{align}
Next, setting
\begin{align}
\chi_{\nu,c}(x_1,...,x_N):= \chi_{c}(|x_2-x_1|), \qquad \qquad \nu=(1,2),\nonumber
\end{align}
the identities $R_0(z) \KK_{\nu}^{*} = \KK_{\nu}^{*} \widetilde{R}_0(z)$, $ \chi_{\nu,c} \KK_{\nu}^{*}=\KK_{\nu}^{*} (\chi_{c} \otimes 1)$ and Lemma \ref{lm4.7} imply that
\begin{align}
\lim\limits_{\eps \to 0} \chi_{\nu,c} R_0(z) \KK_{\nu}^{*} \eps^{-1} U_{\eps}^{*} (v_{\nu} \otimes 1) =  \chi_{\nu,c}  S(z)_{\nu}^{*}  \qquad \quad \textup{in} \; \LL(\widetilde{\XX}_{\nu},H^n(\R^{2N}))
\label{Convsecond}
\end{align}
for any $n \in \N_0$. Hence, we conclude from \eqref{Convfirst} and \eqref{Convsecond} that 
\begin{align}
\lim\limits_{\eps \to 0} -\eps^{-2} \, (u_{\sigma} \otimes 1) \, U_{\eps} \KK_{\sigma} \chi_{\nu,c} R_0(z) \KK_{\nu}^{*} U_{\eps}^{*} (v_{\nu} \otimes 1)=  - \Ket{u_{\sigma}} T_{\sigma} \chi_{\nu,c}  S(z)_{\nu}^{*} \label{Gesconv}
\end{align}
in $\LL(\widetilde{\XX}_{\nu}, \widetilde{\XX}_{\sigma})$. To complete the proof of \eqref{compactlim}, we have to move the cutoff to the left on both sides of \eqref{Gesconv}. We start by considering the limit operators. For a given $\psi \in \widetilde{\XX}_{\nu}$, let  $\widetilde\psi:= S(z)_{\nu}^{*} \psi$. Then  \eqref{Convsecond} shows that $ \chi_{\nu,c} \widetilde\psi \in H^n(\R^{2N})$ for any $n \in \N$, so it follows from a standard Sobolev embedding theorem (see e.g. \cite[Theorem 8.8]{AnalysisLL}) that $ \chi_{\nu,c} \widetilde\psi$ defines a smooth function. In particular, $\chi_{\nu,c} \widetilde\psi \in D(T_{\sigma})$ and $T_{\sigma}(\chi_{\nu,c} \widetilde\psi)$ is explicitly given by \eqref{trace}. Now, a direct calculation, using the defining relations for  $\chi_{\nu,c}$ and $\chi_{\sigma \nu,c}$,  yields that
\begin{align}
T_{\sigma}(\chi_{\nu,c} \widetilde\psi) = \chi_{\sigma \nu,c} T_{\sigma} \widetilde\psi,\nonumber
\end{align}
where, by \eqref{DefSsigma} and Proposition \ref{prop4.5},
$\widetilde\psi=S(z)_{\nu}^{*}\psi = G(z)_{\nu}^{*} \Bra{v_{\nu}} \psi \in D(T_{\sigma})$. Therefore,
\begin{align}
-\Ket{u_{\sigma}} T_{\sigma} \chi_{\nu,c}  S(z)_{\nu}^{*} \psi = \Big( \Ket{u_\sigma} \Bra{v_{\nu}} \otimes [\chi_{\sigma \nu,c}\Theta(z)_{\sigma \nu}] 
\Big) \psi,\nonumber
\end{align}
so the limit operator in \eqref{Gesconv} agrees with the desired limit operator in \eqref{compactlim}. \newline
It remains to show that the left side of \eqref{Gesconv} agrees with the  left side of \eqref{compactlim} with $\Lambda_{\eps}^{h}(z)_{\sigma \nu}$ given by \eqref{DefLambdah}. 
If  $\sigma=(k,l)$ with $3 \leq k<l\leq N$, then this follows immediately from $U_{\eps} \KK_{\sigma} \chi_{\nu,c} =  (1 \otimes \chi_{\sigma \nu,c}) U_{\eps} \KK_{\sigma}$. If $\sigma=(1,l)$ with $l \geq 3$, then, by some abuse of notation,
\begin{align}
U_{\eps} \KK_{\sigma} \chi_{\nu,c} = \chi_{c}\left(x_2-R+ \frac{\eps m_l r}{m_1+ m_l} \right) U_{\eps} \KK_{\sigma}. \label{Comm1l}
\end{align}
For the purpose of computing the limit in \eqref{Gesconv}, the right side of \eqref{Comm1l} may be replaced by $\chi_c(|x_2-R|) U_{\eps} \KK_{\sigma}$ because $\chi_c$ is Lipschitz, $u_{\sigma}$ has compact support
and, by Proposition \ref{prop4.4}, $\|\Lambda_{\eps}^{h}(z)_{\sigma \nu}\|$ is uniformly bounded in $\eps>0$. Again, \eqref{compactlim} follows. The remaining case $\sigma=(2,l)$ with $l \geq 3$ is treated similarly.
\end{proof}

\subsection{Convergence of $\Lambda_{\eps}(z)^{-1}$}     \label{sec4.3}
In Sections \ref{sec4.1} and \ref{sec4.2} we have seen that  $(\Lambda_{\eps}(z)_{\textup{diag}})^{-1}$ and a regularized version of $\Lambda_{\eps}(z)_{\textup{off}}$ have limits as $\eps \rightarrow 0$. This will now allow us
to prove invertibility of $\Lambda_{\eps}(z)$ for small $\eps > 0$ and large enough $z>0$, and existence of
$\limepsr \Lambda_{\eps}(z)^{-1}$. We claim that
\begin{align}
\lim_{\eps \to 0} \,(\Lambda_{\eps}(z)^{-1})_{\sigma \nu} = 
\begin{cases} \dfrac{\Ket{u_{\sigma}}\Bra{v_{\nu}}}{\Braket{u_{\sigma}| v_{\sigma}} \Braket{u_{\nu}| v_{\nu}}} \otimes (\Theta(z)^{-1})_{\sigma \nu}  \qquad \quad &\textup{if}\; \sigma, \nu \in \JJ \\
\;\; 0 \qquad &\textup{else},
\end{cases}
\label{Lambdafac2}
\end{align}
where $\Theta(z)=(\Theta(z)_{\sigma \nu})_{\sigma, \nu \in \JJ}$ is invertible in the reduced Hilbert space $\XX$ defined by \eqref{DefHred}.

To prove \eqref{Lambdafac2}, we introduce some auxiliary operators. Let $\Pi: \widetilde{\XX} \rightarrow L^2\left(\R^{2}, \d{r} \right) \otimes \XX$ denote the matrix operator defined by
\begin{align}
\Pi_{\sigma \nu}:=\begin{cases}
\delta_{\sigma \nu}  \qquad \qquad &\textup{if} \; \sigma,  \nu \in \JJ \\  0  &\textup{if}\; \sigma \in \JJ, \nu \in  \II \backslash \JJ
\end{cases}\nonumber
\end{align}
and let $U=(U_{\sigma \nu})_{\sigma, \nu \in \JJ}$, where $U_{\sigma \nu} \in \LL(L^2(\R^2))$ is given by
\begin{align}
U_{\sigma \nu}:= \dfrac{\Ket{u_{\sigma}}\Bra{v_{\nu}}}{\Braket{u_{\sigma}| v_{\sigma}}\Braket{u_{\nu}| v_{\nu}}}.\nonumber
\end{align}
Let  $\Theta(z)_{\textup{diag}}$ and $\Theta(z)_{\textup{off}}$ denote the operators in $\XX$ defined in terms of the
components $\Theta(z)_{\sigma \sigma}$ and $\Theta(z)_{\sigma \nu}$, $\sigma \neq \nu$, that we have introduced in Sections \ref{sec4.1} and \ref{sec4.2.2}, respectively, and let 
\begin{align}
\Theta(z) := \Theta(z)_{\textup{diag}} + \Theta(z)_{\textup{off}}. \nonumber
\end{align}
Proposition \ref{prop4.5} implies that $\Theta(z)_{\textup{off}} \in \LL(\XX)$, while, by \eqref{DefTheta(z,P)diag},
 $\Theta(z)_{\textup{diag}}$ is unbounded in $\XX$. With the help of $\Pi$ and $U$, Equation \eqref{Lambdadiaginvlim} can now be written as 
\begin{align}
\lim_{\eps \to 0}(\Lambda_{\eps}(z)_{\textup{diag}})^{-1}= \Pi^{*} \left( U \circ(\Theta(z)_{\textup{diag}})^{-1} \right) \Pi. \label{Lambdadiaglim}
\end{align}
Here, the operator product $A \circ B$ is defined by $(A\circ B)_{ij}= A_{ij} \otimes B_{ij}$, a notation inspired by the Hadamard-Schur product of matrices.

The following proposition proves \eqref{Lambdafac2} in these new notations:
\begin{prop}\label{prop4.9}
There exist $\eps_0, z_0>0$ such that the operator matrix $\Lambda_{\eps}(z)$ is invertible in $\widetilde{\XX}$ for all $\eps \in (0, \eps_0)$ and $z > z_0$, and then
\begin{align}
\lim_{\eps \to 0} \,\Lambda_{\eps}(z)^{-1} = \Pi^{*}(U \circ \Theta(z)^{-1})\Pi, \label{Lambdafac2new}
\end{align}
where  $\Theta(z)$ is a closed and invertible operator in $\XX$.
\end{prop}

\begin{proof}
Recall from \eqref{Lambdadecomp}-\eqref{DefLambdaoff} that
$\Lambda_{\eps}(z) = \Lambda_{\eps}(z)_{\textup{diag}}  + \Lambda_{\eps}(z)_{\textup{off}}$, where $ \Lambda_{\eps}(z)_{\textup{diag}}$ and $\Lambda_{\eps}(z)_{\textup{off}}$ is the diagonal and off-diagonal part of the operator-valued matrix $\Lambda_{\eps}(z) = (g_{\eps, \sigma}^{-1} \delta_{\sigma \nu}- \phi_{\eps}(z)_{\sigma \nu})_{\sigma, \nu \in \II}$, respectively. By Proposition \ref{prop4.4}, $\| \Lambda_{\eps}(z)_{\textup{off}}\| \leq C_{\textup{off}}<\infty$ holds uniformly in $\eps,z > 0$. Furthermore, \eqref{Lambdadiaglim} and the definition \eqref{DefTheta(z,P)diag} of $\Theta(z, \underline{P}_{\sigma})_{\sigma \sigma}$ imply that $\Lambda_{\eps}(z)_{\textup{diag}}$ is invertible and $\|(\Lambda_{\eps}(z)_{\textup{diag}})^{-1} \| \leq (2C_{\textup{off}})^{-1}$, provided that $z > z_0$ and $\eps \in (0, \eps_0)$ for sufficiently large $z_0>0$ and sufficiently small $\eps_0>0$. We conclude that  $\Lambda_{\eps}(z)$ is invertible and
\begin{align}
\Lambda_{\eps}(z)^{-1} = \left(1+ (\Lambda_{\eps}(z)_{\textup{diag}})^{-1}  \Lambda_{\eps}(z)_{\textup{off}} \right)^{-1}(\Lambda_{\eps}(z)_{\textup{diag}})^{-1}, \qquad \qquad z> z_0,\; \eps \in (0, \eps_0). \label{Lambdainvid}
\end{align}

Next, we claim that $\limepsr (\Lambda_{\eps}(z)_{\textup{diag}})^{-1}  \Lambda_{\eps}(z)_{\textup{off}}$ exists for all large enough $z>0$, which, by \eqref{DefLambdadiag}, is equivalent to the assertion that $\limepsr (g_{\eps, \sigma}^{-1}- \phi_{\eps}(z)_{\sigma \sigma})^{-1}\Lambda_{\eps}(z)_{\sigma \nu}$ exists for all pairs $\sigma, \nu \in \II, \sigma \neq \nu$. Without loss of generality, we may assume that $\sigma=(1,2) \neq (k,l) = \nu$ and in the following we drop the index $\sigma$ in the diagonal contributions, i.e. $\phi_{\eps}(z):=\phi_{\eps}(z)_{\sigma \sigma}$, $V:=V_{\sigma}$, $g_{\eps}:=g_{\eps, \sigma}$, $a:=a_{\sigma}$ etc. 
If $\int  V(r)\, \d{r}/(2\pi) <  a$, then it follows from Lemma \ref{lm4.2} and from the fact that, by Proposition \ref{prop4.4}, $\|\Lambda_{\eps}(z)_{\sigma \nu}\| \leq C(\sigma, \nu) $ is uniformly bounded in $\eps,z > 0$ that, for large enough $z>0$,
\begin{align}
\limepsr \left(g_{\eps}^{-1}- \phi_{\eps}(z)\right)^{-1}\Lambda_{\eps}(z)_{\sigma \nu}=0, \qquad \qquad \nu \neq (1,2).\nonumber
\end{align}
If $\int V(r)\,\d{r} =  2 \pi a$, then it is more subtle to prove existence of the limit $\eps \to 0$. We claim that, for all large enough $z>0$,
\begin{align}
\limeps \left(g_{\eps}^{-1}- \phi_{\eps}(z)\right)^{-1}\Lambda_{\eps}(z)_{\sigma \nu} =  \left( \frac{\Ket{u}\Bra{v}}{\Braket{u|v}^2} \otimes D(z)^{-1} \right) \Lambda(z)_{\sigma \nu}, \qquad \quad \nu \neq (1,2),  \label{LimCritical}
\end{align}
where $D(z,\underline{P})$ is defined by \eqref{DefD(z,P)} with $\alpha=\alpha_{\sigma}$ and $\Lambda(z)_{\sigma \nu} \in \LL(\widetilde{\XX}_{\nu}, \widetilde{\XX}_{\sigma})$ is defined by \eqref{DefLambda(z)}. To prove this, we start by showing that large momenta may be discarded: As in Section \ref{sec4.1}, let $\underline{P}=(P,p_3,...,p_N)$ be conjugate to $(R,x_3,...,x_N)$ and let $\eta_K \,(K>0)$ denote multiplication with the characteristic function of the ball $\{\underline{P} \, | \, |\underline{P}| \leq K \}$ in Fourier space. Now, observe that \eqref{DefQ} implies $Q \geq \lambda|\underline{P}|^2$ with $\lambda:=1/\max(m_1+m_2,m_3,...,m_N)>0$, so for fixed $\delta >0$ and $z> z_0$, Lemma \ref{lm4.2} gives
\begin{align}
\left\| \left(g_{\eps}^{-1}\mspace{-1mu}- \phi_{\eps}(z)\right)^{-1}\mspace{-2mu} (1 \otimes (1-\eta_K))\, \Lambda_{\eps}(z)_{\sigma \nu} \right\|
&\leq \widetilde{C} \max \left\{ |\ln(\eps)|^{-1},\ln(\mu(z+\lambda K^2))^{-1} \right\} \,C(\sigma, \nu) 
\nonumber\\ &< \delta/3, \label{LargeKest}
\end{align}
provided that $K>0$ is large enough and $\eps >0$ is small enough. Furthermore, using that, by Proposition \ref{prop4.3},
\begin{align}
\limeps \left(g_{\eps}^{-1}- \phi_{\eps}(z)\right)^{-1} = \frac{\Ket{u}\Bra{v}}{\Braket{u|v}^2} \otimes D(z)^{-1}, \label{Diaginvconv}
\end{align}
similar estimates also show that
\begin{align}
\left\|  \left(\frac{\Ket{u}\Bra{v}}{\Braket{u|v}^2} \otimes D(z)^{-1}\right)  (1 \otimes (1-\eta_K)) \, \Lambda(z)_{\sigma \nu} \right\| < \delta/3 \label{Diaginvconv0}
\end{align}
for large enough $K>0$. In view of \eqref{LargeKest} and \eqref{Diaginvconv0}, it is clear that \eqref{LimCritical} will follow if we can show that
\begin{align}
&\limeps \left(g_{\eps}^{-1}- \phi_{\eps}(z)\right)^{-1} \mspace{-2mu} (1 \otimes \eta_K) \,\Lambda_{\eps}(z)_{\sigma \nu}= \left(\frac{\Ket{u}\Bra{v}}{\Braket{u|v}^2} \otimes D(z)^{-1}\right)\mspace{-2mu} (1 \otimes \eta_K) \,\Lambda(z)_{\sigma \nu} \label{LimCriticalK}
\end{align}
for any fixed $K>0$, $\nu \neq (1,2)$ and large enough $z>0$. To further simplify \eqref{LimCriticalK}, we decompose
\begin{align}
\eta_K = \eta_K\,\chi_{\sigma \nu,c} +
\eta_K\,(1-\chi_{\sigma \nu,c}),  \qquad \qquad c>0,\label{etadecomp}
\end{align}
where $\chi_{\sigma \nu,c}=\chi_{(1,2) \nu,c}$ is defined by \eqref{Defchisigmanu}. Furthermore, for $\nu=(k,l) \neq (1,2)=\sigma$, we set
\begin{align}
\gamma_{\sigma \nu}(R,x_3,...,x_N):= \begin{cases}
|R-x_l|^{1/2} \qquad \quad &\textup{if} \; k \in \{1,2\}, l \geq 3 \\
|x_k-x_l|^{1/2}  \qquad \quad &\textup{if} \; 3 \leq k < l \leq N
\end{cases}\nonumber
\end{align}
and we note that the definition of $\chi_{\sigma \nu,c}$ implies that
\begin{align}
\left\|\eta_K(1-\chi_{\sigma \nu,c}) \right\| = \left\|(1-\chi_{\sigma \nu,c}) \eta_K\right\|
&\leq \left\|(1-\chi_{\sigma \nu,c}) \gamma_{\sigma \nu} \right\| \left\|\gamma_{\sigma \nu}^{-1}(-\Delta'+1)^{-1}(-\Delta'+1)\eta_K\right\| \nonumber \\ &\leq  (2c)^{1/2}\left\|\gamma_{\sigma \nu}^{-1}(-\Delta'+1)^{-1}\right\|(K^2+1), \label{etachiest}
\end{align}
where 
\begin{align}
-\Delta':=-\Delta_R+ \sum_{i=3}^N (-\Delta_{x_i}). \nonumber
\end{align}
As  $|\cdot|^{-1/2} \in L^2(\R^2) + L^{\infty}(\R^2)$, a standard result (see e.g. \cite[Theorem 11.1]{Teschl}) yields that $\left\| \gamma_{\sigma \nu}^{-1}(-\Delta'+1)^{-1}\right\|< \infty$, so the right side of \eqref{etachiest} vanishes as $c \to 0$. Hence, with the help of \eqref{etadecomp} and another $\delta/3$-argument, we see that the proof of \eqref{LimCriticalK} can be reduced to the claim that, for any fixed $c,K>0$, $\nu \neq (1,2)$ and large enough $z>0$,
\begin{align}
&\limeps \left(g_{\eps}^{-1}\mspace{-2mu}-\mspace{-2mu} \phi_{\eps}(z)\right)^{-1}\mspace{-3mu}(1 \otimes(\eta_K\chi_{\sigma \nu,c}))\Lambda_{\eps}(z)_{\sigma \nu} = \left( \frac{\Ket{u}\mspace{-2mu} \Bra{v}}{\Braket{u|v}^2} \otimes D(z)^{-1}\mspace{-2mu}\right)\mspace{-3mu} (1 \otimes(\eta_K \chi_{\sigma \nu,c})) \Lambda(z)_{\sigma \nu}. \nonumber
\end{align}
However, this is immediate from \eqref{Diaginvconv} and Proposition \ref{prop4.8}, so \eqref{LimCriticalK} and thus \eqref{LimCritical} are established. 

Since the choice $\sigma=(1,2)$ in the analysis above was immaterial, we conclude that, for all pairs $\sigma, \nu \in \II, \sigma \neq \nu$,
\begin{align}
\limeps \left(g_{\eps, \sigma}^{-1}- \phi_{\eps}(z)_{\sigma \sigma}\right)^{-1}\Lambda_{\eps}(z)_{\sigma \nu}  =  \begin{cases} \left( \dfrac{\Ket{u_{\sigma}}\Bra{v_{\sigma}}}{\Braket{u_{\sigma}|v_{\sigma}}^2} \otimes (\Theta(z)_{\sigma \sigma})^{-1}\right) \Lambda(z)_{\sigma \nu}  &\textup{if} \; \sigma \in \JJ \\
\qquad 0 \;  \hspace*{4.5cm} &\textup{else}.\label{allpairslim}
\end{cases} 
\end{align}
Let the operator  $\Lambda(z)_{\textup{off}}$ in $\widetilde{\XX}$ be defined by $\left(\Lambda(z)_{\textup{off}}\right)_{\sigma \nu}=  \Lambda(z)_{\sigma \nu} (1-\delta_{\sigma \nu})$. Then, in the notation of \eqref{Lambdadiaglim}, \eqref{allpairslim} takes the form
\begin{align}
\limepsr (\Lambda_{\eps}(z)_{\textup{diag}})^{-1}  \Lambda_{\eps}(z)_{\textup{off}}=\Pi^{*} \left( U \circ(\Theta(z)_{\textup{diag}})^{-1} \right) \Pi\, \Lambda(z)_{\textup{off}}=:L(z). \label{Diagofflimit}
\end{align}
From \eqref{Lambdainvid}, \eqref{Lambdadiaglim} and \eqref{Diagofflimit}, it follows that
\begin{align}
\limepsr \Lambda_{\eps}(z)^{-1} =  \left(1+ L(z) \right)^{-1} \Pi^{*} \left( U \circ(\Theta(z)_{\textup{diag}})^{-1} \right) \Pi,
\label{Lambdainvlim}
\end{align}
provided that $z> z_0$ with $z_0>0$ large enough.

To simplify the right side of \eqref{Lambdainvlim}, we first note that the inverse $(1+ L(z))^{-1}$ is only needed on $\ran \, \Pi^{*}$, which, in view of \eqref{Diagofflimit}, is left invariant by $L(z)$.  Explicitly, we have that
\begin{align}
\left(1+ L(z) \right)^{-1} \Pi^{*} =  \Pi^{*}\left(1+ \left[ U \circ (\Theta(z)_{\textup{diag}})^{-1} \right] \Pi\, \Lambda(z)_{\textup{off}} \Pi^{*}\right)^{-1}.\label{inv(1+L(z))}
\end{align}
Secondly, by Proposition \ref{prop4.8}, we have the factorization property
\begin{align}
\left(\Lambda(z)_{\textup{off}}\right)_{\sigma \nu}=\Lambda(z)_{\sigma \nu}=\Ket{u_{\sigma}}\Bra{v_{\nu}} \otimes \Theta(z)_{\sigma \nu}, \qquad \qquad \sigma \neq \nu, \label{Lambdaofffac}
\end{align}
where $\Theta(z)_{\sigma \nu} \in \LL(\XX_{\nu}, \XX_{\sigma})$ for $\sigma \neq \nu$.
From \eqref{Lambdaofffac}, it follows that 
\begin{align}
\left[ U \circ (\Theta(z)_{\textup{diag}})^{-1} \right] \Pi\, \Lambda(z)_{\textup{off}} \Pi^{*}= \widetilde{U} \circ \left[ (\Theta(z)_{\textup{diag}})^{-1}\Theta(z)_{\textup{off}} \right], \label{Thetainvid1}
\end{align}
where the components of $ \widetilde{U}=(\widetilde{U}_{\sigma \nu})_{\sigma, \nu \in \JJ}$
are defined by
\begin{align}
\widetilde{U}_{\sigma \nu} := \frac{\Ket{u_{\sigma}} \mspace{-2mu}\Bra{v_{\nu}}}{\Braket{u_{\sigma}| v_{\sigma}}} \nonumber
\end{align}
and the identity $\widetilde{U}_{\sigma \nu}= U_{\sigma \lambda} \Ket{u_{\lambda}} \mspace{-2mu}\Bra{v_{\nu}}$ was used.
Furthermore, a direct computation, using the identity $(\widetilde{U} \circ A)(\widetilde{U} \circ B)=(\widetilde{U} \circ (AB)) $,  shows that, on $\ran\,\Pi$,
\begin{align}
\left(1 + \widetilde{U} \circ \left[ (\Theta(z)_{\textup{diag}})^{-1}\Theta(z)_{\textup{off}} \right]\right)^{-1}
=1-\widetilde{U} \circ \left[ \Theta(z)^{-1}\Theta(z)_{\textup{off}}\right]. \label{Thetainvid2}
\end{align} 
Indeed, $\Theta(z)^{-1}=(1+(\Theta(z)_{\textup{diag}})^{-1}\Theta(z)_{\textup{off}})^{-1}(\Theta(z)_{\textup{diag}})^{-1}$ exists for large enough $z>0$, since our analysis in Section \ref{sec4.1} shows that $\lim_{z \rightarrow \infty}  (\Theta(z)_{\textup{diag}})^{-1}=0$ and, by  Proposition  \ref{prop4.5}, $ \|\Theta(z)_{\sigma \nu}\|$ with $\sigma \neq \nu$ and hence $\|\Theta(z)_{\textup{off}}\|$ are uniformly bounded in $z>0$. Now, \eqref{Lambdainvlim} combined with \eqref{inv(1+L(z))}, \eqref{Thetainvid1} and \eqref{Thetainvid2} yields
\begin{align}
\limepsr \Lambda_{\eps}(z)^{-1} &= \Pi^{*} \left(1 + \widetilde{U} \circ \left[ (\Theta(z)_{\textup{diag}})^{-1}\Theta(z)_{\textup{off}} \right]\right)^{-1}\left( U \circ(\Theta(z)_{\textup{diag}})^{-1} \right) \Pi
\nonumber \\
&= \Pi^{*}\left(1-\widetilde{U} \circ \left[ \Theta(z)^{-1}\Theta(z)_{\textup{off}}\right]\right)\left( U \circ(\Theta(z)_{\textup{diag}})^{-1} \right) \Pi \nonumber \\
&=  \Pi^{*}\left(U \circ \left[ (\Theta(z)_{\textup{diag}})^{-1} -\Theta(z)^{-1}\Theta(z)_{\textup{off}} (\Theta(z)_{\textup{diag}})^{-1}\right] \right)\Pi   \nonumber \\
&=  \Pi^{*}\left(U \circ \Theta(z)^{-1}\right) \Pi, \nonumber
\end{align} 
where the last equation follows from the  second resolvent identity and the second to last equation used $\widetilde{U}_{\sigma \lambda}U_{\lambda \nu}= U_{\sigma \nu}$.
\end{proof}

\section{Proof of Theorem \ref{theo1}. and lower bound on $\sigma(H)$} \label{sec5}
Recall from Corollary \ref{kor:KK} that for any $z \in \rho(H_0)$ 
we have that $z \in \rho(H_{\eps})$ if and only if $\Lambda_{\eps}(z)$ is invertible in $\widetilde{\XX}$ and in this case the resolvent $(H_{\eps}+z)^{-1}$ can be expressed by the Konno-Kuroda formula
\begin{align}
(H_{\eps}+z)^{-1} = R_0(z) + \sum_{\sigma, \nu \in \II} 
\left(A_{\eps, \sigma}R_0(\overline{z})\right)^{*} (\Lambda_{\eps}(z)^{-1})_{\sigma \nu}\, J_{\nu} A_{\eps, \nu} R_0(z). \label{KK2}
\end{align}
In virtue of Proposition \ref{prop4.9}, there exist $ \eps_0, z_0 >0$ such that $\Lambda_{\eps}(z)$ is invertible for all $z > z_0$ and $\eps \in (0, \eps_0)$. This implies that $ (z_0,\infty) \subset \rho(H_{\eps}) \cap \rho(H_0)$ for any $\eps \in (0, \eps_0)$ and $(H_{\eps}+z)^{-1}$ is then given by \eqref{KK2}. Moreover, Proposition \ref{prop4.9} also
shows that $\lim_{\eps \to 0} \,\Lambda_{\eps}(z)^{-1}$ exists for any $z > z_0$ and, by the componentwise version \eqref{Lambdafac2} of \eqref{Lambdafac2new},   
\begin{align}
\lim_{\eps \to 0} \,(\Lambda_{\eps}(z)^{-1})_{\sigma \nu} = 
\begin{cases} \dfrac{\Ket{u_{\sigma}}\Bra{v_{\nu}}}{\Braket{u_{\sigma}| v_{\sigma}} \Braket{u_{\nu}| v_{\nu}}} \otimes (\Theta(z)^{-1})_{\sigma \nu}  \qquad \quad &\textup{if}\; \sigma, \nu \in \JJ \\
\;\; 0 \qquad &\textup{else}.\nonumber
\end{cases}
\end{align}
Here, $\Theta(z)$ is a closed and invertible operator in the reduced Hilbert space 
$\XX$ defined by \eqref{DefHred}. 
Since $J_{\nu}=\sgn(V_{\nu})$ is bounded and since, by Proposition \ref{prop3.3}, $\limepsr A_{\eps, \nu}R_0(z)=S(z)_{\nu}$ exists for all $\nu$, we conclude that we can take the limit $\eps \rightarrow 0$ on the right side of \eqref{KK2}. We find that
$\limepsr (H_{\eps}+z)^{-1}=R(z)$ for all $z > z_0$, where 
\begin{align}
R(z):= R_0(z) + \sum_{\sigma, \nu \in \JJ} 
S(\overline{z})_{\sigma}^{*} \left(  \frac{\Ket{u_{\sigma}}\Bra{v_{\nu}}}{\Braket{u_{\sigma}| v_{\sigma}} \Braket{u_{\nu}| v_{\nu}}} \otimes (\Theta(z)^{-1})_{\sigma \nu} \right) J_{\nu} S(z)_{\nu}. \label{DefR(z)}
\end{align}

Expression \eqref{DefR(z)} can be simplified as follows: By Proposition \ref{prop3.3}, we have that  $S(z)_{\sigma}\psi = v_{\sigma} \otimes (G(z)_{\sigma}\psi)$ with $G(z)_{\sigma}\in \LL(\HH, \XX_{\sigma})$ and, similarly, $J_{\nu} S(z)_{\nu}\psi=u_{\nu} \otimes (G(z)_{\nu}\psi)$. Hence, \eqref{DefR(z)} takes the form
\begin{align}
R(z)&= R_0(z) + \sum_{\sigma, \nu \in \JJ} 
G(\overline{z})_{\sigma}^{*} \,(\Theta(z)^{-1} )_{\sigma \nu} \,G(z)_{\nu} = R_0(z) + G(\overline{z})^{*} \,\Theta(z)^{-1} \,G(z), \label{DefR(z)2}
\end{align}
where $G(z) \in  \LL(\HH,  \XX)$ is given by $G(z)\psi=( G(z)_{\sigma} \psi )_{\sigma \in \JJ}$.
This expression for $R(z)$ agrees with the right side of \eqref{Hres}. Moreover, it also shows that $R(z)$ only depends on $V_{\sigma}\;(\sigma \in \JJ)$ via the parameters $\beta_{\sigma}$ that are defined by \eqref{Defbetasigma}.

To prove Theorem \ref{theo1} for $z>z_0$, it remains to show that $R(z)$ indeed defines the resolvent $(H+z)^{-1}$ of a self-adjoint operator $H$. This follows from the Trotter-Kato Theorem (see e.g. \cite[Theorem 5]{DR} and \cite[Theorem VIII.22]{RS1}\footnote{ In \cite{RS1} existence of $\limepsr (H_{\eps}+z)^{-1}$ in two points $z$ with $\pm \Ima(z)>0$ is assumed, but the proof can be adapted to the case where the limit exists for all $z$ from a non-empty open interval.}), provided that we can show that $\ran\,R(z)$ is dense in $\HH$. Let $T: D(H_0) \rightarrow \XX$ be given by $T \psi=(T_{\sigma} \psi)_{\sigma \in \JJ}$, where the trace $T_{\sigma}: D(T_{\sigma}) \subset \HH \rightarrow \XX_{\sigma}$ is defined by \eqref{DefTsigma} on the domain $D(T_{\sigma}) \supset D(H_0)$. For compactly supported, smooth functions $\psi$, $T_{\sigma}\psi$ agrees with \eqref{trace}, so $\ker\,T$ contains the dense set $C^{\infty}_0(\R^{2N} \setminus \Gamma) \subset \HH$,
where 
\begin{align}
\Gamma:= \{(x_1,...,x_N) \in \R^{2N} \; | \; \exists i\neq j \;\textup{with}\; x_i=x_j \}\nonumber
\end{align}
denotes the union of the contact planes of the $N$ particles. 
Let $\phi \in (\ran\,R(z))^{\perp}$. 
Then, for any $\psi \in \ker\,T$, it follows from \eqref{DefR(z)2} and $G(z)=TR_0(z)$ that
\begin{align}
0 = \Braket{\phi| R(z)(H_0+z) \psi } &= \Braket{\phi| \psi } + 
\Braket{\phi| G(\overline{z})^{*} \,\Theta(z)^{-1} \,T \psi } =  \Braket{\phi| \psi }. \nonumber
\end{align}
As $\ker\,T \subset  \HH$ is dense, this implies that $\phi=0$, i.e.  $\ran\,R(z) \subset \HH $ is dense. Now, by the Trotter-Kato Theorem, there exists a self-adjoint operator $H$ such that $R(z)=(H+z)^{-1}$ for $z > z_0$ and, moreover, $H_{\eps} \rightarrow H$ in the norm resolvent sense. This completes the proof of Theorem \ref{theo1} for $z>z_0$.

To prove Theorem \ref{theo1} for all $z \in \rho(H) \cap \rho(H_0)$, it suffices to verify the hypotheses of 
\cite[Theorem 2.19]{CFP2018}. To this end, we first need to define $\Theta(z)$ for all $z \in \rho(H_0) = \C \setminus (-\infty, 0]$. For the diagonal parts $\Theta(z)_{\sigma \sigma}$, this is achieved by \eqref{DefTheta(z,P)diag} with $\ln$ denoting the principal branch of the logarithm. To define the off-diagonal parts  for $z \in \rho(H_0)$, we use  $\Theta(z)_{\sigma \nu}:=-T_{\sigma} G(\overline{z})_{\nu}^{*}$ for $\sigma \neq \nu$, which agrees with \eqref{DefThetaoff} for $z>0$, but, a priori, may be an unbounded operator for other values of $z \in \rho(H_0)$. From $G(z)_{\nu}=T_{\nu}R_0(z)$ and from the first resolvent identity, we get
\begin{align}
G(\overline{z})_{\nu}^{*}= G(\overline{w})_{\nu}^{*} + (w-z) R_0(z)G(\overline{w})_{\nu}^{*}, \qquad \qquad z, w \in \rho(H_0). \label{Gnudiff}
\end{align}
Choosing $w \in (0, \infty)$, we know from Proposition  \ref{prop4.5} that $\ran\,G(\overline{w})_{\nu}^{*} \subset D(T_{\sigma})$ and that $T_{\sigma}G(\overline{w})_{\nu}^{*}$ is a bounded operator. Likewise, $T_{\sigma}R_0(z)$ is a bounded operator and hence, by \eqref{Gnudiff},  $\Theta(z)_{\sigma \nu}=-T_{\sigma} G(\overline{z})_{\nu}^{*}$ is a bounded operator  defined on $\XX_{\nu}$ for all $z \in \rho(H_0)$. Moreover,
\begin{align}
\Theta(z)_{\sigma \nu} = \Theta(w)_{\sigma \nu} + (z-w) G(z)_{\sigma}G(\overline{w})_{\nu}^{*}, \qquad \qquad \sigma \neq \nu, \quad z,w \in \rho(H_0). \label{Thetaoffid}
\end{align}
By Proposition \ref{prop5.1}, below, the extended operator $\Theta(z)$ satisfies
the hypotheses of  \cite[Theorem 2.19]{CFP2018}, so the proof of Theorem \ref{theo1} is complete.$ \hfill \square$

\begin{prop} \label{prop5.1}
Let $w,z \in \rho(H_0)$. Then the operator $\Theta(z)$ has the following properties:
\begin{enumerate}[label=(\roman*)]
\item  $\Theta(z)^{*}= \Theta(\overline{z})$
\item  $\Theta(z) = \Theta(w) + (z-w) G(z)G(\overline{w})^{*}$
\item  $ 0 \in \rho(\Theta(z))$ for some $z \in \rho(H_0)$
\end{enumerate} 
\end{prop}
\noindent \textit{Remark.} $(ii)$ implies that $D=D(\Theta(z))$ is independent of $z \in \rho(H_0)$.
\begin{proof}
Property $(iii)$ has already been verified in Proposition \ref{prop4.9}. To prove $(ii)$, we first recall that  $\Theta(z)$  and $G(z)$ have 
the components $ \Theta(z)_{\sigma \nu } \; (\sigma, \nu \in \JJ) $ and  $G(z)_{\sigma} \; (\sigma \in \JJ)$, respectively. Hence, we have to verify that
\begin{align}
\Theta(z)_{\sigma \nu} = \Theta(w)_{\sigma \nu} + (z-w) G(z)_{\sigma}G(\overline{w})_{\nu}^{*}, \qquad \qquad \sigma, \nu \in \JJ. \label{Thetaid}
\end{align}
If $\sigma \neq \nu$, then this has been shown in \eqref{Thetaoffid}. In the case of $\sigma= \nu$, we assume that $\sigma=(1,2)$ for notational simplicity. Equations \eqref{DefTsigma} and \eqref{H0H0Tilde} imply that
\begin{align}
G(z)_{\sigma}G(\overline{w})_{\sigma}^{*} = \tau \KK_{\sigma} R_0(z) \left( \tau \KK_{\sigma} R_0(\overline{w}) \right)^{*} = \tau (\widetilde{H}_0+z)^{-1} \left( \tau (\widetilde{H}_0+\overline{w})^{-1} \right)^{*},\nonumber
\end{align}
where $\widetilde{H}_0$ is defined by \eqref{DefH0Tilde}. Using now the definition \eqref{Deftau} of $\tau$ together with a Fourier transform in $(R,x_3,...,x_N)$, we find that $(z-w) G(z)_{\sigma}G(\overline{w})_{\sigma}^{*}$ acts pointwise in $\underline{P}=(P,p_3,...,p_N)$ by multiplication with
\begin{align}
\dfrac{z-w}{4\pi^2} \; \int\limits_{\R^2} \left( \frac{p^2}{\mu_{\sigma}} + Q +z \right)^{-1} \mspace{-3mu}\left(\frac{p^2}{\mu_{\sigma}} + Q +w\right)^{-1} \mspace{-3mu} \d{p} =  \frac{\mu_{\sigma}}{4 \pi} \left[ \ln(z+Q) - \ln(w+Q)  \right], \label{DefMultOp}
\end{align}
where $Q \geq 0$ is defined by \eqref{DefQ}. Now \eqref{Thetaid}, for $\sigma=\nu=(1,2)$, follows from \eqref{DefMultOp} and from the definition \eqref{DefTheta(z,P)diag} of $\Theta(z, \underline{P}_{\sigma} )_{\sigma \sigma}$.

To prove $(i)$, we use the fact that $\Theta(z)$ does not depend on the particular choices of $V_{\sigma}$, $a_{\sigma}$ and $b_{\sigma}$ as long as the parameters $\beta_{\sigma}$ and the set $\JJ$ remain unchanged.
For $\sigma \in \JJ$ we choose $V_{\sigma}>0$,  $a_\sigma = \int V_{\sigma}(r)\,\d{r} / (2 \pi)$ and $b_{\sigma}$ to solve \eqref{Defbetasigma} and \eqref{Defalphasigma} for the given values of $\beta_{\sigma}$. For $\sigma \in \II \setminus \JJ$ we set $V_{\sigma}=0$. Moreover, we choose  $z \in (z_0, \infty)$ so that \eqref{Lambdafac2new} holds. Then $u_{\sigma}= v_{\sigma}$ for all pairs $\sigma \in \II$ and hence
$\Lambda_{\eps}(z)$ is self-adjoint for all $\eps>0$. Now, it follows from \eqref{Lambdafac2new}, or more directly from \eqref{Lambdafac2}, that $\Theta(z)^{-1}$ is self-adjoint, too. Therefore, $\Theta(z)$ is self-adjoint for $z \in (z_0, \infty)$ and $(i)$ for general $z \in \rho(H_0)$ now easily follows from $(ii)$.
\end{proof}

Our results on $\Theta(z)_{\sigma \nu}$ imply the following lower bound on $\sigma(H)$ (see also \cite{Figari, DR}).

\begin{prop} \label{prop5.2}
Let $H$ denote the Hamiltonian from Theorem \ref{theo1}. Then, with $N_{\JJ}:=|\JJ|$, $ \mu^{-}:= \min_{\sigma \in \JJ}\mu_{\sigma}$, $ \beta^{-}:= \min_{\sigma \in \JJ}\beta_{\sigma}$ and $m:= \max_{i=1,...,N} m_i$, it holds that
\begin{align}
\inf \sigma(H) \geq  -  \exp\left( \frac{\pi m}{\mu^{-}} (N_{\JJ} -1) - \frac{\beta^{-}}{\pi} \right). \label{SigmaHLower}
\end{align}
\end{prop}

\begin{proof} Without restriction, we can assume that $N_{\JJ}\geq 1$ because for $H=H_0$ the bound \eqref{SigmaHLower} is obvious.
From  \cite[Theorem 2.19]{CFP2018} and from Proposition \ref{prop5.1}, we know that a point $z \in \rho(H_0)$ belongs to $\rho(H)$ if and only if $\Theta(z)$ is invertible in $\XX$. 
Since $\Theta(z)$ is self-adjoint for $z>0$, it suffices to show that $\Theta(z)$ is bounded from below by a positive constant for $z >z_0:=\exp( \pi m ( N_{\JJ} -1)/ \mu^{-} - \beta^{-}/\pi)$. To this end,
we use the definition \eqref{DefTheta(z,P)diag} of $\Theta(z)_{\sigma \sigma}$ and,  for $\sigma \neq \nu$, we use the bound $\| \Theta(z)_{\sigma \nu} \| \leq m/4$ from Proposition \ref{prop4.5}. As $z>z_0$ implies that $\ln(z) + \beta^{-}/\pi > 0$, we find that, for all $w = (w_{\sigma})_{\sigma \in \JJ} \in D(\Theta(z))$, 
\begin{align}
\Braket{w | \Theta(z) w}&=\sum_{\sigma \in \JJ} \Braket{w_{\sigma} | \Theta(z)_{\sigma \sigma} w_{\sigma}} + \sum_{\substack{ \sigma, \nu \in \JJ \\ \sigma \neq \nu}} \Braket{w_{\sigma} | \Theta(z)_{\sigma \nu} w_{\nu}} \nonumber \\
& \geq  \frac{\mu^{-}}{4\pi}  \left(\ln(z) + \frac{\beta^{-}}{\pi}\right)  \sum_{\sigma \in \JJ} \|w_{\sigma} \|^2 - \frac{m}{4} \sum_{\substack{ \sigma, \nu \in \JJ \\ \sigma \neq \nu}} \|w_{\sigma} \| \|w_{\nu}\| \nonumber \\
& \geq  \frac{\mu^{-}}{4\pi}  \left(\ln(z) + \frac{\beta^{-}}{\pi}\right)\| w \|^2 - \frac{m}{8} \sum_{\substack{ \sigma, \nu \in \JJ \\ \sigma \neq \nu}} ( \|w_{\sigma} \|^2 +  \|w_{\nu}\|^2) \nonumber \\ 
&= \left[ \frac{\mu^{-}}{4\pi}  \left(\ln(z) + \frac{\beta^{-}}{\pi} 
\right)- 
 \frac{m}{4} \left( N_{\JJ} -1 \right) \right] \|w\|^2. \nonumber
\end{align}
The expression in brackets is positive for $z > z_0=\exp( \pi m ( N_{\JJ} -1)/ \mu^{-} - \beta^{-}/\pi)$, which proves \eqref{SigmaHLower}.
\end{proof}

\section{The quadratic form of the Hamiltonian} \label{sec6}
In this section we determine the quadratic form of $H$  and we show, in the case of $N$ particles 
of mass one, that it agrees with a quadratic form introduced in \cite{Figari}. This proves that $H$ is the TMS Hamiltonian of Dell'Antonio et al. (see \cite[Eqs. (5.3), (5.4)]{Figari}). 

 We start by deriving an explicit formula for the quadratic form of $H$ restricted to $D(H)$:
\begin{lemma} \label{lm6.1}
Let $H$ denote the self-adjoint operator from Theorem \ref{theo1} and let $z \in \rho(H_0) \cap \rho(H)$. Then, for any $\psi \in D(H)$, it holds that 
\begin{align}
\Braket{\psi | H \psi}= \Braket{\psi -  G(z)^{*}w | (H_0+z)(\psi - G(\overline{z})^{*} w) }+ \Braket{w | \Theta(z) w} -z \|\psi\|^2, \label{HForm}
\end{align} 
where  $w=(w_{\sigma})_{\sigma \in \JJ} \in D(\Theta(z))$ is uniquely determined by \eqref{TMS1} and \eqref{TMS2}.
\end{lemma}
\begin{proof}
By \eqref{Hact}, we have that $(H+z) \psi = (H_0+z)\psi_0$, which yields that
\begin{align}
\Braket{\psi | (H+z) \psi} &= \Braket{\psi | (H_0+z) \psi_0} \nonumber \\ &= \Braket{\psi- G(z)^{*} w  | (H_0+z)  \psi_0} + \Braket{w|G(z) (H_0+z)  \psi_0 }. \label{HFormid}
\end{align}
By the definition of $G(z)$ and by \eqref{TMS2}, $G(z) (H_0+z)  \psi_0= T\psi_0 = \Theta(z)w$, so it follows from \eqref{TMS1} and \eqref{HFormid} that
\begin{align}
\Braket{\psi | (H+z) \psi} &= \Braket{\psi- G(z)^{*} w  | (H_0+z)  \left(\psi - G(\overline{z})^{*} w \right)} + \Braket{w|  \Theta(z)w}, \nonumber
\end{align}
which proves \eqref{HForm}.
\end{proof}

Next, we are going to determine an explicit description of the closure of the quadratic form from Lemma \ref{lm6.1}. From \eqref{DefTheta(z,P)diag}, it follows that $\Theta(z)_{\sigma \sigma} \geq  c_1\ln(z) + c_2$ in operator sense, where $c_1>0$ and $c_2 \in \R$ depend on the pair $\sigma \in \JJ$ but not on $z \in (0,\infty)$. Combining this with the fact that $\|\Theta(z)_{\textup{off}}\|$ is uniformly bounded in $z \in (0,\infty)$, we see that $\Theta(z) \geq c >0$ for sufficiently large $z \in \rho(H) \cap (0,\infty)$. For such $z$, we introduce a quadratic form $q$ with domain
\begin{align}
D(q):=\left\{\psi \in \HH \,\Big\vert\, \exists w \in D(\Theta(z)^{1/2}): \psi -  G(z)^{*}w \in H^1(\R^{2N})  \right\}\nonumber
\end{align}
by 
\begin{align}
q(\psi):= \|(H_0+z)^{1/2}(\psi -  G(z)^{*}w)\|^2 + \| \Theta(z)^{1/2} w\|^2 -z \|\psi\|^2, \label{Defq}
\end{align}  
which agrees with the  right side of \eqref{HForm} if $\psi \in D(H) \subset D(q)$. Similarly to \cite[Lemma 3.2]{Figari}, one verifies that $G(z)^{*}w \notin H^1(\R^{2N})$ for any $w \in D(\Theta(z)^{1/2}) \setminus \{0\}$, so $w=w^{\psi}$ is uniquely determined by $\psi \in D(q)$ and, in particular, $q$ is well-defined. Furthermore, it is not hard to see that $q$ is bounded from below and closed. Hence, there exists a unique self-adjoint operator $H_{q}$ associated with $q$. We are going to show that $H_{q}=H$. Let $\psi \in D(H_{q})$ be fixed. Then, by the definitions of $q$ and $H_q$, for all $\phi \in D(q)$,
\begin{align}
\Braket{(H_{q}+z) \psi| \phi} =& \Braket{(H_0+z)^{1/2 }(\psi -  G(z)^{*}w^{\psi})|(H_0+z)^{1/2}(\phi - G(z)^{*} w^{\phi})} \nonumber \\
&+ \Braket{\Theta(z)^{1/2} w^{\psi} | \Theta(z)^{1/2}w^{\phi}}. \label{Hqform}
\end{align}
Choosing $w^{\phi}=0$, wee see that any $\phi \in H^1(\R^{2N})$ belongs to $D(q)$, and 
hence, for $\phi \in H^1(\R^{2N})$, Eq. \eqref{Hqform} becomes
\begin{align}
\Braket{(H_{q}+z) \psi| \phi} = \Braket{(H_0+z)^{1/2}(\psi -  G(z)^{*}w^{\psi}) | (H_0+z)^{1/2}\phi}.
\end{align}
This equation shows that $\psi_0:= \psi - G(z)^{*} w^{\psi} \in D(H_0)$, and that
\begin{align}
(H_{q}+z)\psi=(H_0+z)\psi_0. \label{Hqact}
\end{align}
We have thus verified condition \eqref{TMS1} of Corollary \ref{cor2}. It remains to check condition \eqref{TMS2}. 

Given $w \in D(\Theta(z)^{1/2})$, we choose $\phi:= G(z)^{*}w$, so that $\phi - G(z)^{*}w = 0 \in H^1(\R^{2N})$. It follows that $\phi \in D(q)$, and from Eqs. \eqref{Hqform} and \eqref{Hqact}, we see that, for all $w \in D(\Theta(z)^{1/2})$,
\begin{align}
\Braket{ \Theta(z)^{1/2} w^{\psi}  | \Theta(z)^{1/2}  w}  &=  \Braket{(H_{q}+z) \psi|  \phi} \nonumber \\&=  \Braket{(H_{0}+z) \psi_0|  G(z)^{*}w} \nonumber \\ &= \Braket{G(z)(H_{0}+z) \psi_0| w} = \Braket{T\psi_0| w}. \nonumber
\end{align}
This equation implies that $w^{\psi} \in D(\Theta(z))$ and that $\Theta(z)w^{\psi}=T \psi_0$, 
which is condition \eqref{TMS2}. Corollary \ref{cor2} now shows that $\psi \in D(H)$ and, in view of Eq. \eqref{Hqact}, that $H_q \subset H$. From the self-adjointness of $H_q$ and $H$, we conclude that $H_q=H$.

In the case of $m_i=1$ for $i=1,...,N$, we are going to show that $q$ agrees with a quadratic form introduced in \cite[Eqs. (2.13)-(2.16)]{Figari}. 
For given $\beta=(\beta_{\sigma})_{\sigma \in \II}$, the quadratic form in \cite{Figari} is denoted by $\textup{F}_{\beta}$ and, in our notation, it is defined by
\begin{align}
\textup{F}_{\beta}(\psi):=  \| \nabla(\psi -  G_z^{2N} \ast \xi^{\psi})\|^2 + z \| \psi -  G_z^{2N} \ast \xi^{\psi} \|^2 + \Phi^{z,1}_{\beta}(\xi^{\psi}) + \Phi^{z,2}(\xi^{\psi}) - z \| \psi \|^2 \label{DefFbeta}
\end{align}
on the domain
\begin{align}
D(\textup{F}_{\beta})=\left\{\psi \in \HH \,\Big\vert\, \exists \xi^{\psi} \in D(\Phi^{z,1}_{\beta}):\psi -  G_z^{2N} \ast \xi^{\psi} \in H^1(\R^{2N})  \right\}.\nonumber
\end{align}
The right side of \eqref{DefFbeta} is independent of the particular choice of $z \in (0,\infty)$ and $\xi^{\psi}=(\xi_{\sigma})_{\sigma \in \JJ}$ is a collection of ,,charges`` $\xi_{\sigma} \in L^2(\R^{2(N-1)})$, which are uniquely determined by $\psi \in D(\textup{F}_{\beta})$.
For $\sigma=(i,j)$,  $\xi_{\sigma}$ may be interpreted as a function on the hyperplane $x_i=x_j$ and the convolution $G_z^{2N} \ast \xi^{\psi}$ is to be understood in the sense that
\begin{align}
\widehat{G_z^{2N} \ast \xi^{\psi}}(p_1,...,p_N) = \mspace{-5mu} \sum_{\sigma=(i,j) \in \JJ} \left(z+ \sum_{n=1}^{N} p_n^2 \right)^{\mspace{-5mu}-1} \mspace{-12mu} \cdot \widehat{\xi_{\sigma}}\left(p_1,...,p_{i-1},\frac{p_i+p_j}{\sqrt{2}},...\widehat{p}_j ..., p_N \right) \label{DefGConv}
\end{align}
(cf. the proof of \cite[Lemma 3.2 $(a)$]{Figari}). Moreover, by \cite[Theorem 3.3]{Figari}, $\textup{F}_{\beta}$ is bounded from below and closed on $D(\textup{F}_{\beta})$. 

Assuming that $m_i=1$ for $i=1,...,N$ and that $z \in \rho(H) \cap (0,\infty)$ is so large that $\Theta(z) \geq c >0$, we now write the various contributions to the right side of \eqref{Defq} in a form that will allow us to compare them to their counterparts in \eqref{DefFbeta}. Recall from Section \ref{sec5} that $G(z) \in  \LL(\HH,  \XX)$ is given by $G(z) \psi=(G(z)_{\sigma} \psi)_{\sigma \in \JJ}$ with $G(z)_{\sigma}=T_{\sigma}R_0(z)$. For a given $w=(w_{\sigma})_{\sigma \in \JJ} \in \XX$, this implies that
\begin{align}
G(z)^{*}w = \sum_{\sigma \in \JJ} G(z)_{\sigma}^{*} w_{\sigma}, \label{G(z)Adjoint}
\end{align}
where, by \eqref{G_(k,l)(z)*Fourier} for $\sigma=(i,j)$,
\begin{align}
(\widehat{G(z)_{\sigma}^{*} w_{\sigma}})\left(p_1,...,p_N\right) = \frac{1}{2\pi }\left(z+ \sum_{n=1}^{N}p_n^2\right)^{\mspace{-5mu}-1} \mspace{-12mu}\cdot \widehat{w_{\sigma}}\left(p_i+p_j,p_{1}, ...  \widehat{p}_i... \widehat{p}_j ..., p_N\right). \label{G_(i,j)(z)*Fourier}
\end{align}
From \eqref{DefGConv}, \eqref{G(z)Adjoint} and \eqref{G_(i,j)(z)*Fourier}, we infer that 
\begin{align}
G(z)^{*}w = G_z^{2N} \ast \xi, \qquad \qquad \quad w=(w_{\sigma})_{\sigma \in \JJ}, \xi=(\xi_{\sigma})_{\sigma \in \JJ}, \label{G(z)*Identify} 
\end{align}
where $\xi_{\sigma}$ agrees with $w_{\sigma}$ up to a permutation of the arguments and rescaling:
\begin{align}
\xi_{\sigma}(x_1,...\widehat{x}_j ..., x_N) = \frac{1}{4 \pi} \, w_{\sigma}\left(\frac{x_i}{\sqrt{2}},x_1,...\widehat{x}_i... \widehat{x}_j ..., x_{N}\right), \qquad  \qquad \sigma=(i,j), \label{Defxisigma}
\end{align}
which, after Fourier transform, is equivalent to
\begin{align}
\widehat{\xi_{\sigma}}(p_1,...\widehat{p}_j ..., p_N) =  \frac{1}{2 \pi} \, \widehat{w_{\sigma}}\left(\sqrt{2}p_i,p_1,...\widehat{p}_i... \widehat{p}_j ..., p_{N}\right), \qquad \qquad  \sigma=(i,j). \label{DefxisigmaP}
\end{align}
Conversely, given $\xi_{\sigma} \in L^2(\R^{2(N-1)})$ for all $\sigma \in \JJ$, the identity \eqref{Defxisigma} uniquely determines $w=(w_{\sigma})_{\sigma \in \JJ} \in \XX$ and again \eqref{G(z)*Identify} holds true. In particular, it follows that
\begin{align}
\| (H_0+z)^{1/2}(\psi- G(z)^{*} w) \|^2 =\| \nabla(\psi -  G_z^{2N} \ast \xi)\|^2 + z \| \psi -  G_z^{2N} \ast \xi \|^2,\nonumber
\end{align}
provided that $\psi- G(z)^{*} w \in H^1(\R^{2N})$.
The third and the fourth terms in \eqref{DefFbeta} correspond to the diagonal and off-diagonal contributions, respectively, in 
\begin{align}
\|\Theta(z)^{1/2} w \|^2=\sum_{\sigma \in \JJ}\| (\Theta(z)_{\sigma \sigma})^{1/2} w_{\sigma} \|^2 + \sum_{\substack{ \sigma, \nu \in \JJ \\ \sigma \neq \nu}} \Braket{w_{\sigma} | \Theta(z)_{\sigma \nu} w_{\nu}}.\nonumber
\end{align}
We begin with the diagonal ones. Recall that for  $\sigma=(i,j)$ the operator $\Theta(z)_{\sigma \sigma}$ acts pointwise in $\underline{P}_{\sigma}=(P,p_1,...\widehat{p}_i...\widehat{p}_j ...,p_N)$ by multiplication with $\Theta(z, \underline{P}_{\sigma})_{\sigma \sigma}$ defined by \eqref{DefTheta(z,P)diag}. For given $w=(w_{\sigma})_{\sigma \in \JJ} \in D(\Theta(z)^{1/2})$, it follows that
\begin{align}
\sum_{\sigma \in \JJ} \| (\Theta(z)_{\sigma \sigma})^{1/2} w_{\sigma} \|^2
= \sum_{\sigma \in \JJ} \int \mspace{-3mu} \d{\underline{P}_{\sigma}} \, \Theta(z, \underline{P}_{\sigma})_{\sigma \sigma} \left|\widehat{w_{\sigma}}\left(\underline{P}_{\sigma}\right)\right|^2 , \nonumber
\end{align}
which agrees with  $\Phi^{z,1}_{\beta}(\xi)$ defined by \cite[Eq. (2.15)]{Figari}, where $\beta=(\beta_{\sigma})_{\sigma \in \JJ}$ and $\xi \in D(\Phi^{z,1}_{\beta})$ are given by \eqref{Defbetasigma} and \eqref{DefxisigmaP}, respectively. Conversely, given $\xi \in D(\Phi^{z,1}_{\beta})$, a straightforward computation shows that \eqref{DefxisigmaP} defines a vector $w \in D(\Theta(z)^{1/2})$ and that $\Phi^{z,1}_{\beta}(\xi)=\sum_{\sigma \in \JJ} \| (\Theta(z)_{\sigma \sigma})^{1/2} w_{\sigma} \|^2$. In particular, we see that $w \in D(\Theta(z)^{1/2})$ if and only if 
\eqref{DefxisigmaP} defines a vector $\xi \in D(\Phi^{z,1}_{\beta})$ and hence, using \eqref{G(z)*Identify}, it follows that $D(q)=D(\textup{F}_{\beta})$. 
It remains to examine the contributions of the off-diagonal operators  $\Theta(z)_{\sigma \nu} \in \LL(\XX_{\nu}, \XX_{\sigma})$. Here, the key is the identity (cf. Eq. \eqref{DefThetaoff})
\begin{align}
 \Theta(z)_{\sigma \nu}= -T_{\sigma}G(z)_{\nu}^{*}, \qquad \qquad \sigma \neq \nu. \label{DefThetaoff2}
\end{align}
Recall from \eqref{TsigmaFourier} that, for $\sigma=(i,j)$ and $\psi \in D(T_{\sigma})$,  
\begin{align}
&(\widehat{T_{\sigma}\psi})\left(P,p_1,...\widehat{p_i}...\widehat{p_j}...,p_N\right)\nonumber \\
& \quad= \frac{1}{2\pi}\mspace{-3mu}\int\mspace{-3mu} \d{p} \; \widehat{\psi}\mspace{-2mu}\left(p_1,...,p_{i-1}, \frac{P}{2}-p,p_{i+1},...,p_{j-1}, \frac{P}{2}+p , p_{j+1}, ..., p_{N} \right). \label{TsigmaFourier2}
\end{align}
Now, using first \eqref{DefThetaoff2}, \eqref{TsigmaFourier2} and the substitution
\begin{align}
P:=p_i+p_j, \quad p:= \frac{p_j -p_i}{2}, \nonumber
\end{align}
and inserting the identity \eqref{G_(i,j)(z)*Fourier} for $G(z)_{\nu}^{*} w_{\nu}$ thereafter, we obtain that
\begin{align}
 \sum_{\substack{ \sigma, \nu \in \JJ \\ \sigma \neq \nu}}  \Braket{w_{\sigma} | \Theta(z)_{\sigma \nu} w_{\nu}  } &= -(2 \pi)^{-2} \sum_{\substack{ \sigma, \nu \in \JJ \\ \sigma \neq \nu}} \int \mspace{-3mu} \d{p_1} \cdots  \d{p_N} \;\overline{\widehat{w_{\sigma}}}\left(p_i+p_j,p_{1}, ...  \widehat{p}_i... \widehat{p}_j ..., p_N\right) \nonumber \\ & \;\;\; \cdot \left(z+ \sum_{n=1}^{N}p_n^2 \right)^{ \mspace{-3mu}-1} \mspace{-6mu} \widehat{w_{\nu}}\left(p_k+p_l,p_{1}, ...  \widehat{p}_k... \widehat{p}_l ..., p_N\right), \label{ThetaoffForm}
\end{align}
where the sum runs over all $\sigma=(i,j)\neq(k,l)= \nu$. In view of \eqref{DefxisigmaP}, \eqref{ThetaoffForm} coincides with the expression for $\Phi^{z,2}(\xi)$ from \cite[Eq. (2.16)]{Figari}. This completes the proof that $q=\textup{F}_{\beta}$.

\appendix
\setcounter{theorem}{0}
\setcounter{equation}{0}

\vspace*{1cm}
\LARGE \noindent
\textbf{Appendix}\normalsize
\section{Properties of the Green's function}  \label{secA}
This appendix collects facts and estimates on the Green's function of $-\Delta+z: H^2(\R^d) \rightarrow L^2(\R^d)$. 

For $d \in \N$ and $z \in \C$ with $\Rea(z)>0$, let the function $G^d_{z}: \R^d \rightarrow \C$ be defined by
\begin{align}
G^d_{z}(x) := \int\limits_{0}^{\infty} (4\pi t)^{-d/2}  \exp\left( -\Scalefrac{x^2}{4t} - zt\right) \textup{dt} . \label{DefGreen}
\end{align}
Note that $G^d_{z}$ has a singularity at $x=0$ unless $d=1$.  The proof of the following lemma is left to the reader.
\begin{lemma}\label{lmA1}
Let $d \in \N$ and $z \in \C$ with $\Rea(z)>0$ be given. Then $G^d_{z}$ defined by \eqref{DefGreen} satisfies the following properties:
\begin{enumerate}[label=(\roman*)]
\item $G^d_{z} \in L^1(\R^d)$ and $\|G^d_{z} \|_{L^1}\leq \Rea(z)^{-1}$ 
\item The Fourier transform of $G^d_{z}$ is given by $\widehat{G^d_{z}}(p)=(2\pi)^{-d/2}(p^2+z)^{-1}$
\item $G^d_{z}$ is the Green's function of $-\Delta + z$, i.e. $(-\Delta + z)^{-1}f = G^d_{z} \ast f$ holds for all $f \in L^2(\R^d)$
\item  $G^d_{z} \in L^2(\R^d)$ if and only if $d \in \{1,2,3\}$ 
\item $G^d_{z}$ is spherically symmetric, i.e. $G^d_{z}(x)$ only depends on $\left|x\right|$, and  for $z \in (0,\infty)$ it is positive and 
strictly monotonically decreasing both as a function of $\left|x\right|$ and $z$
\item  Let $d_1,d_2\in \N$ with $d_1+d_2=d$ and decompose $x=(x_1,x_2) \in \R^{d_1} \times \R^{d_2}$. If $x_1 \neq 0$ or $d_1=1$, then $G^d_z(x_1,\cdot) \in L^1(\R^{d_2})$ and the Fourier transform with respect to $x_2$ is
\begin{align}
\widehat{G^d_{z}}(x_1,p_2)=(2\pi)^{-d_2/2} \; G^{d_1}_{z+p_2^2}(x_1). \nonumber
\end{align}
In particular, we have that
\begin{align}
 \int\limits_{\R^{d_2}}^{} \mspace{-3mu}  G^d_{z}(x_1,x_2) \; \textup{d}x_2 = G_{z}^{d_1}(x_1). \nonumber
\end{align}
\end{enumerate}
\end{lemma} 

\hspace*{5mm} It follows from \eqref{DefGreen} that $G^d_z(x) \rightarrow \infty$ as $|x| \rightarrow 0$ in dimensions $d \geq2$. However, the next lemma shows that, apart from $x=0$, $G^d_z$ is smooth and exponentially decaying as $|x| \rightarrow \infty$:

\begin{lemma}\label{lmA2} 
Let $d \in \mathbb{N}$ and  $z \in \mathbb{C}$ with $\Rea(z)>0$. 
Then $G^d_z$ is smooth in $\R^d \setminus \{0\}$ and, for any multi-index $\alpha \in \N_{0}^d$, $\partial^{\alpha} G^d_z$ is exponentially decaying as $|x| \rightarrow \infty$. 
\end{lemma}\medskip

\begin{proof}
First, we claim that, for all $c>0$, $k \in \R$, $l \geq 0$ and $z' \in (0, \sqrt{\Rea(z)})$,
\begin{align}
\sup_{|x| \geq c}\left(\int\limits_{0}^{\infty} 
t^k  |x|^{l} \exp\left( -\Scalefrac{x^2}{4t} - \Rea(z)t + z'|x|\right)  \textup{dt} \right)  < \infty. \label{supbound}
\end{align}
To prove this, we choose $\delta>0$ so small that $z' \leq (1-\delta)\sqrt{\Rea(z)}$. Then the estimate
\begin{align}
\exp\left( -\Scalefrac{x^2}{4t} - \Rea(z)t + z'|x|\right) &\leq 
\exp\left( -\Scalefrac{x^2}{4t} - \Rea(z)t + (1-\delta)\sqrt{\Rea(z)} |x| \right) \nonumber \\ & = \exp\left( -(1-\delta) \left( \Scalefrac{|x|}{2\sqrt{t}} - \sqrt{\Rea(z)t} \right)^2- \delta \left( \Scalefrac{x^2}{4t} +  \Rea(z)t \right)  \right) \nonumber \\
& \leq  \exp\left( -\frac{\delta x^2}{8t} - \frac{ \delta c^2}{8t} -\delta \Rea(z) t \right)\nonumber
\end{align}
holds for all $t>0$ and $|x| \geq c$. With the help of 
$ C_{l,\delta}:= \sup_{s \geq 0}\left(s^{l/2} \exp\left( -\delta s\right)    \right) < \infty$,  \eqref{supbound} now follows from
\begin{align}
\sup_{|x| \geq c}&\left(\int\limits_{0}^{\infty}  
t^k  |x|^{l} \exp\left( -\Scalefrac{x^2}{4t} - \Rea(z)t  + z'|x| \right) \d{t} \right)
\nonumber \\
& \leq 8^{l/2} \sup_{|x| \geq c}\left(\int\limits_{0}^{\infty} 
t^{k+l/2}  \left(\frac{x^{2}}{8t}\right)^{l/2} \exp\left( -\frac{\delta x^2}{8t} - \frac{ \delta c^2}{8t}-\delta \Rea(z)t \right) \d{t} \right) \nonumber \\
& \leq 8^{l/2} C_{l,\delta} \; \int\limits_{0}^{\infty} 
t^{k+l/2}\exp\left( -\frac{ \delta c^2}{8t} - \delta \Rea(z)t\right) \d{t} < \infty,\nonumber
\end{align}
where the last integral is finite because the term $\exp\left( - \delta c^2/(8t)\right)$ cancels a possible divergence at $t=0$. \eqref{supbound} allows us to change the order of differentiation and integration in the representation \eqref{DefGreen} of $G^d_z$. We find that, for $x\neq 0$, $\partial^{\alpha}G^d_z(x)$ is a sum of terms of the form
\begin{align}
\int\limits_{0}^{\infty}
(4\pi t)^{-d/2} \, \frac{x^{\alpha'}}{(2t)^n} \exp\left( -\Scalefrac{x^2}{4t} - zt\right)  \d{t}, \nonumber
\end{align}
where $ 0 \leq n \leq |\alpha|$ and $\alpha' \in \N_0^d$ is a multi-index with $| \alpha'| \leq | \alpha|$.
In view of \eqref{supbound}, all of these terms are exponentially decaying as $|x| \rightarrow \infty$.
\end{proof}

In the proof of Proposition \ref{prop4.4} two integral operators $F$ and $B$ depending on $G_{z}^{d}$  
were introduced. We close this section with the bypassed proof of their boundedness:  

\begin{lemma}\label{lmA3}
Let $z,m>0$. Then the operator $F: L^2(\R^2 \times \R^2) \rightarrow L^2(\R^2 \times \R^2)$
that is defined in terms of the kernel 
\begin{align}
K_{F}(x,y;x',y'):= G_{z}^6\left(\sqrt{m}(x-x'),\sqrt{m}(x-y'),\sqrt{m}(y-x')\right)
\nonumber
\end{align}
is bounded with $\| F \| \leq  (4 \sqrt{2}m^2)^{-1}$. Similarly, the kernel
\begin{align}
K_B(x,y,w;x',y',w'):= G_{z}^8\left(\sqrt{m}(x-y'),\sqrt{m}(x-w'),\sqrt{m}(y-x'),\sqrt{m}(w-x')\right) \nonumber
\end{align}
defines an operator $B \in \LL(L^2(\R^2 \times \R^2\times \R^2 ))$ with $\| B \| \leq  (4 \sqrt{2} m^3)^{-1}$.
\end{lemma}

\begin{proof}
The bounds for $\| F \|$ and $\| B \|$ are both based on the Schur test. In the case of $F$, 
the Schur test, in a general version, yields that 
\begin{align}
\| F \| &\leq \esssup_{x,y}\left(h(x,y)^{-1} \int \mspace{-3mu} \d{x'} \, \d{y'} \; h(x',y') K_F(x,y;x',y')  \right) \label{Schurtestgen}
\end{align}
for any measurable test function $h: \R^2 \times \R^2 \rightarrow (0,\infty)$, provided that the right side is finite.
We choose the test function $h(x,y):=|x-y|^{-1}$, so after scaling we may assume that $m=1$. To evaluate the right side of \eqref{Schurtestgen}, we insert the integral representation \eqref{DefGreen} of the Green's function  $G^{6}_{z}$ and we substitute $x-y' \rightarrow y'$. This results in the estimate
\begin{align}
 &\esssup_{x,y}\left(|x-y| \int \mspace{-3mu} \d{x'} \, \d{y'} \; |x'-y'|^{-1}\, G^{6}_{z}\left(x-x',x-y',y-x'\right) \right)\nonumber \\ &= \esssup_{x,y}\left(|x-y| \int \mspace{-3mu} \d{x'} \, \int\limits_0^{\infty} \mspace{-3mu} \d{t} \; (4\pi t)^{-3}
\exp\left( -\Scalefrac{1}{4t} \left((x-x')^2+(y-x')^2\right) -zt \right)  \right. \nonumber \\
&\quad \cdot \int \mspace{-3mu} \textup{d}y'  \; |y'+x'-x|^{-1} \exp\left(- \frac{|y'|^2}{4t} \right) \Biggs). \label{Schurtest1}
\end{align}
Using a rearrangement inequality (see \cite[Theorem 3.4]{AnalysisLL}), 
the last integral is bounded by
\begin{align}
 \int \mspace{-3mu} \textup{d}y'  \; |y'+x'-x|^{-1} \exp\left(- \frac{|y'|^2}{4t} \right) 
 & \leq  \int \mspace{-3mu} \textup{d}y'  \; |y'|^{-1} \exp\left(- \frac{|y'|^2}{4t} \right) \nonumber \\
 &= 2 \pi \int\limits_0^{\infty} \mspace{-3mu} \d{s} \; \exp\left(- \frac{s^2}{4t} \right) = 
 \sqrt{4\pi^3t}. \label{lastintbound}
\end{align}
For the remaining integral, we use
$(x-x')^2+(y-x')^2 = \Scalefrac{1}{2} \left( (2x'-x-y)^2+(x-y)^2 \right)$
together with the substitution  $2x'-x-y \rightarrow 2x'$. This yields that
\begin{align}
&\pi  \esssup_{x,y}\left(|x-y| \int \mspace{-3mu} \d{x'} \, \int\limits_0^{\infty} \mspace{-3mu} \d{t} \; (4\pi t)^{-5/2}
\exp\left( -\Scalefrac{1}{4t} \left((x-x')^2+(y-x')^2\right) -zt \right)\right)  \nonumber \\ &\;\;\;
\leq \pi \esssup_{x,y}\left(|x-y| \int \mspace{-3mu} \d{x'} \, \int\limits_0^{\infty} \mspace{-3mu} \d{t} \; (4\pi t)^{-5/2}
\exp\left( -\frac{|x'|^2}{2t} - \frac{(x-y)^2}{8t}  \right)\right) \nonumber \\
&\;\;\;= (16 \sqrt{\pi})^{-1}  \esssup_{x,y}\left(|x-y| \int\limits_0^{\infty} \d{t} \; t^{-3/2} \exp\left(- \frac{(x-y)^2}{8t}  \right) \right)\nonumber \\
&\;\;\;=   (4 \sqrt{2\pi})^{-1} \int \limits_{0}^{\infty} \d{t} \; t^{-3/2} \exp\left( - t^{-1}\right). \label{intid1}
\end{align}
In the last equality, the $\esssup$ cancels after the scaling $t\rightarrow (x-y)^2t/8$. 
The remaining integral can be explicitly computed with the help of the substitution $\widetilde{t}:=t^{-1}$. We find that
\begin{align}
 \int \limits_{0}^{\infty} \d{t} \; t^{-3/2} \exp\left( - t^{-1} \right)  = 
 \int \limits_{0}^{\infty} \d{\tilde{t}} \; \tilde{t}^{-1/2} \exp\left(-\tilde{t}\,\right) = \Gamma(1/2)= \sqrt{\pi}, \label{Gammaint}
\end{align}
where $\Gamma$ denotes the Gamma function. Using \eqref{lastintbound}, \eqref{intid1} and \eqref{Gammaint} to bound the right side of \eqref{Schurtest1}, the Schur test yields that $F$ is bounded and its norm satisfies the desired estimate.

The proof in case of the operator $B$ is similar, but for the convenience of the reader we go into details here. First, applying the Schur test with the test function $h(y,w)=|y-w|^{-1}$ gives
\begin{align}
\| B \| &\leq \esssup_{x,y,w}\left(|y-w| \int \mspace{-3mu} \d{x'} \, \d{y'}\, \d{w'} \; |y'-w'|^{-1} K_B(x,y,w;x',y',w')  \right), \nonumber
\end{align}
provided that the right side is finite. After scaling, we may again assume that $m=1$.
Next, inserting the integral formula \eqref{DefGreen} for $G^{8}_{z}$ and substituting $x-y' \rightarrow y'$, it follows that
\begin{align}
\| B \| &\leq  \esssup_{x,y,w} \mspace{-2mu}\left(|y-w| \int \mspace{-3mu} \d{x'} \,\d{w'} \int\limits_0^{\infty} \mspace{-3mu} \d{t} \; (4\pi t)^{-4}
\exp\left( - \Scalefrac{1}{4t} \left[(x-w')^2+(y-x')^2+ (w-x')^2\right]\right)  \right. \nonumber \\
&\quad \cdot \int \mspace{-3mu} \textup{d}y'  \; |x-y'-w'|^{-1} \exp\left(- \frac{|y'|^2}{4t} \right) \Biggs). \label{Schurtest2}
\end{align}
Similarly to \eqref{lastintbound}, a rearrangement inequality now yields that
\begin{align}
 \int \mspace{-3mu} \textup{d}y'  \; |x-y'-w'|^{-1} \exp\left(- \frac{|y'|^2}{4t} \right) 
 &\leq \sqrt{4\pi^3t}. \label{lastintbound2}
\end{align}
In the remaining integral, we use $(y-x')^2+(w-x')^2 = \Scalefrac{1}{2} \left( (2x'- w- y)^2+(y-w)^2 \right) $
in combination with the substitutions  $2x'- w-  y \rightarrow 2x'$ and $x-w' \rightarrow w'$. This leads to
\begin{align}
&\pi \esssup_{x,y,w}\mspace{-2mu}\left(|y-w| \int \mspace{-3mu}  \d{x'} \,\d{w'} \int\limits_0^{\infty} \mspace{-3mu} \d{t} \; (4\pi t)^{-7/2}
\exp\mspace{-2mu}\left( - \Scalefrac{1}{4t} \left[(x-w')^2+(y-x')^2+  (w-x')^2\right]\right)  \mspace{-2mu}  \right)  \nonumber \\ &\;= \pi \esssup_{x,y,w}\mspace{-2mu}\left(|y-w| \int \mspace{-3mu}  \d{x'} \,\d{w'} \int\limits_0^{\infty} \mspace{-3mu} \d{t} \; (4\pi t)^{-7/2}
\exp\mspace{-2mu}\left(- \Scalefrac{1}{4t} \left(|w'|^2+2|x'|^2\right) - \frac{(y-w)^2}{8t}  \right) \mspace{-2mu} \right) \nonumber \\
&\;= (16\sqrt{\pi} )^{-1}  \esssup_{y,w}\left(|y-w| \int\limits_0^{\infty} \d{t} \; t^{-3/2} \exp\left(- \frac{(y-w)^2}{8t}  \right) \right)\nonumber \\
&\; =  (4 \sqrt{2\pi})^{-1} \int \limits_{0}^{\infty} \d{t} \; t^{-3/2} \exp\left( - t^{-1} \right), \label{intid2}
\end{align}
where the $\esssup$ vanishes after the scaling $t\rightarrow (y-w)^2t/8$. 
By \eqref{Gammaint}, the last integral in \eqref{intid2} is equal to $\sqrt{\pi}$, so the desired estimate for $\| B\|$ follows from \eqref{Schurtest2}, \eqref{lastintbound2} and \eqref{intid2}.
\end{proof}

\section{Konno-Kuroda formula} \label{secB}
In this section we sketch the proof of the Konno-Kuroda resolvent identity, see \cite[Eq. (2.3)]{KoKu1966}, for operators of the type \eqref{AbstractHam}. The main difference between Theorem \ref{thm:KK} below and \cite{KoKu1966} is that we do not assume that $\phi(z)$ defined by \eqref{Defphi(z)} extends to a compact operator for some (and hence all) $z \in \rho(H_0)$.

Let $\HH$ and $\widetilde{\XX}$ be arbitrary complex Hilbert spaces, let $H_0\geq 0$ be a self-adjoint operator in $\HH$ and let $A:D(A)\subset\HH\to \widetilde{\XX}$ be densely defined and closed with $D(A)\supset D(H_0)$. Let $J\in \LL(\widetilde{\XX})$ be a self-adjoint isometry and let $B=JA$. Suppose that $BD(H_0)\subset D(A^{*})$ and that $A^{*}A$ and $A^{*}B$ are $H_0$-bounded with relative bound less than one. Then  
\begin{equation}\label{AbstractHam}
      H = H_0 - A^{*}B
\end{equation}
is self-adjoint on $D(H_0)$ by the Kato-Rellich Theorem. For $z\in \rho(H_0)$, let $R_0(z):=(H_0+z)^{-1}$ and let $\phi(z):D(A^{*})\subset \widetilde{\XX} \to \widetilde{\XX}$ be defined by
\begin{align}
\phi(z) := BR_0(z)A^{*}. \label{Defphi(z)}
\end{align}
Note that $D(A^{*})\subset \widetilde\XX$ is dense because $A$ is closed.  The resolvent $(H+z)^{-1}$ and the operator $\phi(z)$ are related by the following theorem:

\begin{theorem}\label{thm:KK}
Let the above hypotheses be satisfied and let $z\in \rho(H_0)$. Then $\phi(z)$ is a bounded operator. The operator $1-\phi(z)$ is invertible if and only if $z\in \rho(H_0)\cap\rho(H)$, and then
\begin{align}
    (H + z)^{-1} &= R_0(z) + R_0(z)A^{*} (1-\phi(z))^{-1} BR_0(z)\label{HResolvent}\\
   (1-\phi(z))^{-1} &= 1+ B(H+z)^{-1}A^{*}.\label{phi(z)Resolvent}
\end{align}
\end{theorem}

\noindent \emph{Remark.} Note that $(1-\phi(z))^{-1}$ leaves $D(A^{*})$ invariant. This follows from \eqref{phi(z)Resolvent} and from the assumption $BD(H_0)\subset D(A^{*})$.

\begin{proof} $ $ \medskip 

\noindent \emph{Step 1.}  $AR_0(c)^{1/2}$ is bounded for $c>0$, and  $A(H+c)^{-1/2}$ is bounded for $c>0$ large enough. 

As $A^{*}A$ is $H_0$-bounded with relative bound less than one, the operator $H_0-A^*A$ is bounded from below. This implies that, for all $\psi\in D(H_0)$ and all $c>0$,
$$
     \|A\psi\|^2 \leq \|(H_0+c)^{1/2}\psi\|^2 + C \|\psi\|^2 \nonumber
$$
with some constant $C>0$. Due to the fact that $D(H_0)$ is dense in $D(H_0^{1/2})$ and $A$ is closed, this bound extends to all $\psi\in D(H_0^{1/2})$ by an approximation argument. In particular, $D(A)\supset D(H_0^{1/2})$ and the first statement of Step 1 follows. 
The second statement is a consequence of the first and the fact that $H$ and $H_0$ have equivalent form norms, which implies that $D(H^{1/2}) = D(H_0^{1/2})$.\medskip

\noindent \emph{Step 2.} If $z\in \rho(H_0)$, then $\phi(z)$ is a bounded operator, and if $z\in \rho(H)$, then 
$$
     M(z)  := B(H+z)^{-1}A^{*} \nonumber
$$
is a bounded operator too.

This easily follows from Step 1 and from the first resolvent identity.\medskip

\noindent \emph{Step 3.} If $z\in \rho(H_0)\cap \rho(H)$, then $(1-\phi(z))$ is invertible and $1+M(z)=(1-\phi(z))^{-1}$.

Both $\phi(z)$ and $M(z)$ leave $D(A^{*})$ invariant and on this subspace, by straightforward computations using the second resolvent identity, $(1-\phi(z)) (1+M(z)) = 1 = (1+M(z)) (1-\phi(z))$.\medskip

\noindent \emph{Step 4.} If $z\in \rho(H_0)$ and $1-\phi(z)$ is invertible, then $z\in \rho(H)$,  $ (1-\phi(z))^{-1}$ leaves $D(A^{*})$ invariant and \eqref{HResolvent} holds.

By Step 3, $ (1-\phi(i))^{-1} = 1+M(i)$, which leaves $D(A^{*})$ invariant. Now suppose that $z\in \rho(H_0)$ and
that $1-\phi(z)$ has a bounded inverse. Then
\begin{align*}
(1-\phi(z))^{-1}  &=   (1-\phi(i))^{-1} +   (1-\phi(i))^{-1} (\phi(z)-\phi(i))  (1-\phi(z))^{-1}\\
&=(1-\phi(i))^{-1} + (i-z)  (1-\phi(i))^{-1} BR_0(i)(AR_0(\bar z))^{*} (1-\phi(z))^{-1}.\nonumber
\end{align*}
Since $BR_0(i): \HH \rightarrow D(A^{*})$, it follows that $(1-\phi(z))^{-1}$ leaves $D(A^{*})$ invariant as well, and
$$
    R(z) := R_0(z) + R_0(z)A^{*} (1-\phi(z))^{-1} BR_0(z) \nonumber
$$
is well-defined. Now it is a matter of straightforward computations to show that $(H+z)R(z)=1$ on $\HH$ and that $R(z)(H+z)=1$ on $D(H)$. 
\end{proof}

Assume that $\widetilde{\XX}=\bigoplus_{i=1}^{I} \widetilde{\XX}_i$, where $I \in \N$ and $\widetilde{\XX}_i$ are complex Hilbert spaces for $i=1,...,I$. Then we consider operators of the more general form
\begin{align}
H = H_0 - \sum_{i=1}^{I}  g_{i}\, A_{i}^{*}J_{i}A_{i}, 
\label{DefHgen}
\end{align}
where $0 \neq g_i \in \R$, $A_i:D(A_i)\subset\HH\to \widetilde{\XX}_i$ are densely  defined and closed with $D(A_i)\supset D(H_0)$ and $J_i$ are self-adjoint isometries in $\widetilde{\XX}_i$ for $i=1,...,I$.
Suppose that $J_iA_i D(H_0)\subset D(A_i^{*})$ and that $A_i^{*}A_i$ and $A_i^{*}J_iA_i$ are infinitesimally $H_0$-bounded for $i=1,...,I$, which ensures that $H$ is self-adjoint on $D(H_0)$. Now, observe that operators of the form \eqref{DefHgen} can be reduced to the form \eqref{AbstractHam}, where $A:D(A)\subset\HH\to \widetilde{\XX}$ is the closure of the operator $A_0: D(H_0) \rightarrow \widetilde{\XX}$ defined by $A_0\psi=(|g_i|^{1/2}A_i\psi)_{i=1}^{I}$  and $J \in \LL(\widetilde{\XX})$ is defined by $J(\psi_i)_{i=1}^I=(\sgn(g_i)J_i\psi_i)_{i=1}^{I}$. Hence, Theorem \ref{thm:KK} implies the following corollary:

\begin{kor}\label{kor:KK}
Let the above hypotheses be satisfied and let $z\in \rho(H_0)$. Then
\begin{align}
\Lambda(z)_{ij}:=g_i^{-1}\delta_{ij}-J_i A_i R_0(z) A_j^{*}, \qquad \quad i,j=1,...,N \nonumber
\end{align}
defines a bounded operator $\Lambda(z)=(\Lambda(z)_{ij})_{i,j=1}^{I}$ in $\widetilde{\XX}$. Moreover, $\Lambda(z)$ is invertible if and only if $z\in \rho(H_0)\cap\rho(H)$, and then \begin{align}
(H + z)^{-1} &= R_0(z) + \sum_{i,j=1}^{I} \left(A_{i}R_0(\overline{z})\right)^{*} (\Lambda(z)^{-1})_{ij}\, J_{j} A_{j} R_0(z),\nonumber\\
\Lambda(z)^{-1} &= \left( g_i \delta_{ij}+ g_ig_jJ_iA_i(H+z)^{-1}A_j^{*}\right)_{i,j=1}^I.\nonumber
\end{align}
\end{kor}

\noindent
\textbf{Acknowledgement.} Initially, Ulrich Linden was involved in the project as well. He was the first of us to sit down and do some preliminary computations and estimates in the special case of a 2d Bose gas.  This helped tremendously to get the project off the ground. - M.~H.~was supported partially by the \emph{Deutsche Forschungsgemeinschaft (DFG)} through the Research Training Group 1838: \emph{Spectral Theory and Dynamics of Quantum Systems}.


\end{document}